\documentclass{article}

\usepackage{microtype}
\usepackage{graphicx}
\usepackage{subfig}
\usepackage{booktabs} 
\usepackage{multirow}
\usepackage{enumitem}
\usepackage{tabularray}

\usepackage{hyperref}

\usepackage[accepted]{icml2023}

\usepackage{amsmath}
\usepackage{amssymb}
\usepackage{mathtools}
\usepackage{amsthm}

\usepackage[capitalize,noabbrev]{cleveref}

\theoremstyle{plain}
\newtheorem{theorem}{Theorem}[section]

\newtheorem{lemma}[theorem]{Lemma}
\newtheorem{corollary}[theorem]{Corollary}
\theoremstyle{definition}

\newtheorem{assumption}[theorem]{Assumption}
\theoremstyle{remark}

\allowdisplaybreaks
\usepackage[disable,textsize=tiny]{todonotes}

\icmltitlerunning{Bootstrap in High Dimension with Low Computation}

\begin{document}

\twocolumn[
\icmltitle{Bootstrap in High Dimension with Low Computation}



\icmlsetsymbol{equal}{*}

\begin{icmlauthorlist}
\icmlauthor{Henry Lam}{col}
\icmlauthor{Zhenyuan Liu}{col}
\end{icmlauthorlist}

\icmlaffiliation{col}{Department of Industrial Engineering and Operations Research, Columbia University, New York, NY, USA}

\icmlcorrespondingauthor{Henry Lam}{khl2114@columbia.edu}

\icmlkeywords{bootstrap, high-dimensional, computation effort, finite-sample}

\vskip 0.3in
]



\printAffiliationsAndNotice{}  

\begin{abstract}
The bootstrap is a popular data-driven method to quantify statistical uncertainty, but for modern high-dimensional problems, it could suffer from huge computational costs due to the need to repeatedly generate resamples and refit models. We study the use of bootstraps in high-dimensional environments with a small number of resamples. In particular, we show that with a recent ``cheap" bootstrap perspective, using a number of resamples as small as one could attain valid coverage even when the dimension grows closely with the sample size, thus strongly supporting the implementability of the bootstrap for large-scale problems. We validate our theoretical results and compare the performance of our approach with other benchmarks via a range of experiments.

\end{abstract}

\section{Introduction\label{sec:introduction}}


 
The bootstrap is a widely used method for statistical uncertainty quantification, notably confidence interval construction \cite{efron1994introduction,davison1997bootstrap,shao2012jackknife,hall1988bootstrap}. Its main idea is to resample data and use the distribution of resample estimates to approximate a sampling distribution. Typically, this approximation requires running many Monte Carlo replications to generate the resamples and refit models. This is affordable for classical problems, but for modern large-scale problems, this repeated fitting could impose tremendous computation concerns. This issue motivates an array of recent works to curb the computation effort, mostly through a ``subsampling'' perspective that fits models on smaller data sets in the bootstrap process, e.g., \citet{kleiner2012big,lu2020uncertainty,giordano2019swiss,schulam2019can,alaa2020discriminative}.

In contrast to subsampling, we consider in this paper the reduction in bootstrap computation cost by using a fewer number of Monte Carlo replications or resamples. In particular, we target the following question: \emph{Is it possible to run a valid bootstrap for high-dimensional problems with very little Monte Carlo computation?} While conventional bootstraps rely heavily on adequate resamples, recent work \cite{lam2022cheap,lam2022conference} shows that it is possible to reduce the resampling effort dramatically, even down to one Monte Carlo replication. The rough idea of this  ``cheap'' bootstrap is to exploit the approximate independence among the original and resample estimates,  instead of their distributional closeness utilized in the conventional bootstraps. We will leverage this recent idea in this paper. However, since \citet{lam2022cheap,lam2022conference} is based purely on asymptotic derivation, giving an affirmative answer to the above question requires the study of new finite-sample bounds to draw understanding on bootstrap behaviors jointly in terms of problem dimension $p$, sample size $n$ and number of resamples $B$.

To this end, our main theoretical contribution in this paper is three-fold:

\textbf{General Finite-Sample Bootstrap Bounds:} We derive general finite-sample bounds on the coverage error of confidence intervals aggregated from $B$ resample estimates, where $B$ is small using the ``cheap" bootstrap idea, and $B=\infty$ for traditional quantile-based bootstrap methods including the basic and percentile bootstraps (e.g., \citet{davison1997bootstrap} Section 5.2-5.3). Our bounds reveal that, given the same primitives on the approximate normality of the original and each resample estimate, the cheap bootstrap with fixed small $B$ achieves similar coverage error bounds as conventional bootstraps using infinite resamples. This also simultaneously recovers the main result in \citet{lam2022cheap}, but stronger in terms of the finite-sample guarantee.

\textbf{Bootstrap Bounds on Function-of-Mean Models Explicit in $p$, $n$ and $B$:} We specialize our general bounds above to the function-of-mean model that is customary in the high-dimensional Berry-Esseen and central limit theorem (CLT) literature \cite{pinelis2016optimal,zhilova2020nonclassical}. In particular, our bounds explicit on $p$, $n$ and $B$ conclude vanishing coverage errors for the cheap bootstrap when $p=o(n)$, for any given $B\geq1$. Note that the function-of-mean model does not capture all interesting problems, but it has been commonly used -- and in fact, appears the only model used in deriving finite-sample CLT errors for technicality reasons. Our bounds shed light that, at least for this wide class of models, using a small number of resamples can achieve a good coverage even in a dimension $p$ growing closely with $n$.

\textbf{Bootstrap Bounds on Linear Models Independent of $p$:} We further specialize our bounds to linear functions with weaker tail conditions, which have orders independent of $p$ under certain conditions on the $L_p$ norm or Orlicz norm of the linearly scaled random variable.

In addition to theoretical bounds, we investigate the empirical performances of bootstraps using few resamples on large-scale problems, including high-dimensional linear regression, high-dimensional logistic regression, computational simulation modeling, and a real-world data set RCV1-v2 \cite{lewis2004rcv1}. To give a sense of our comparisons that support using the cheap bootstrap in high dimension, here is a general conclusion observed in our experiments: Figure \ref{figure_linear_vs_B}(a) shows the coverage probabilities of $95\%$-level confidence intervals for three regression coefficients with corresponding true values $0,2,-1$ in a 9000-dimensional linear regression (in Section \ref{sec:num}). The cheap bootstrap coverage probabilities are close to the nominal level $95\%$ even with one resample, but the basic and percentile bootstraps only attain around $80\%$ coverage with ten resamples. In this example, one Monte Carlo replication to obtain each resample estimate takes around 4 minutes in the virtual machine e2-highmem-2 in Google Cloud Platform. Therefore, the cheap bootstrap only requires 4 minutes to obtain a statistically valid interval, but the standard bootstrap methods are still far from the nominal coverage even after more than a 40-minute run. Figure \ref{figure_linear_vs_B}(b) shows the average interval widths. This reveals the price of a wider interval for the cheap bootstrap when the Monte Carlo budget is very small, but considering the low coverages in the other two methods and the fast decay of the cheap bootstrap width for the first few number of resamples, such a price appears secondary. 


\begin{figure*}[htbp]
\centering

\subfloat[Coverage probability]{\includegraphics[width=0.45\textwidth]{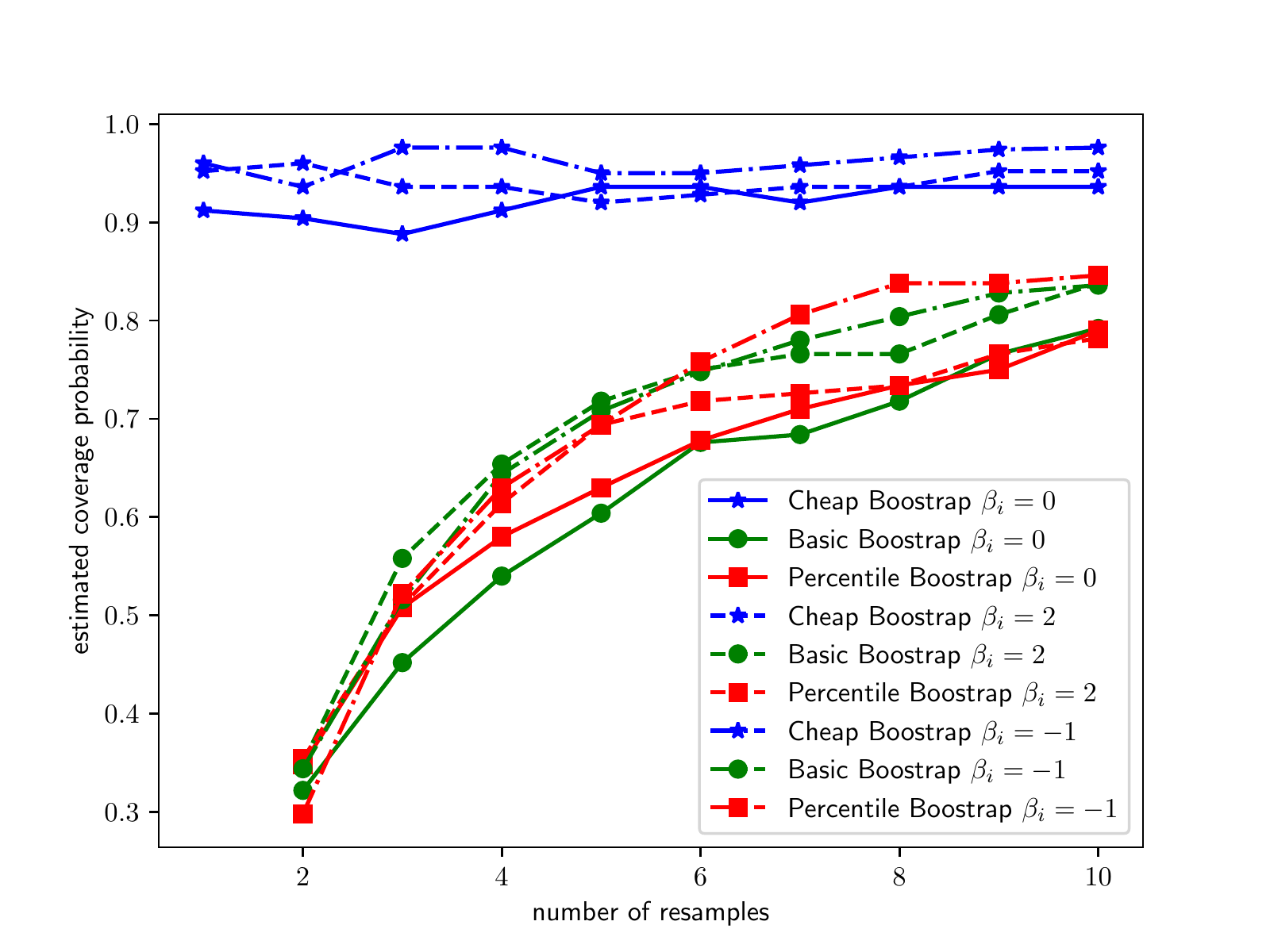}}
\subfloat[Confidence interval width]{\includegraphics[width=0.45\textwidth]{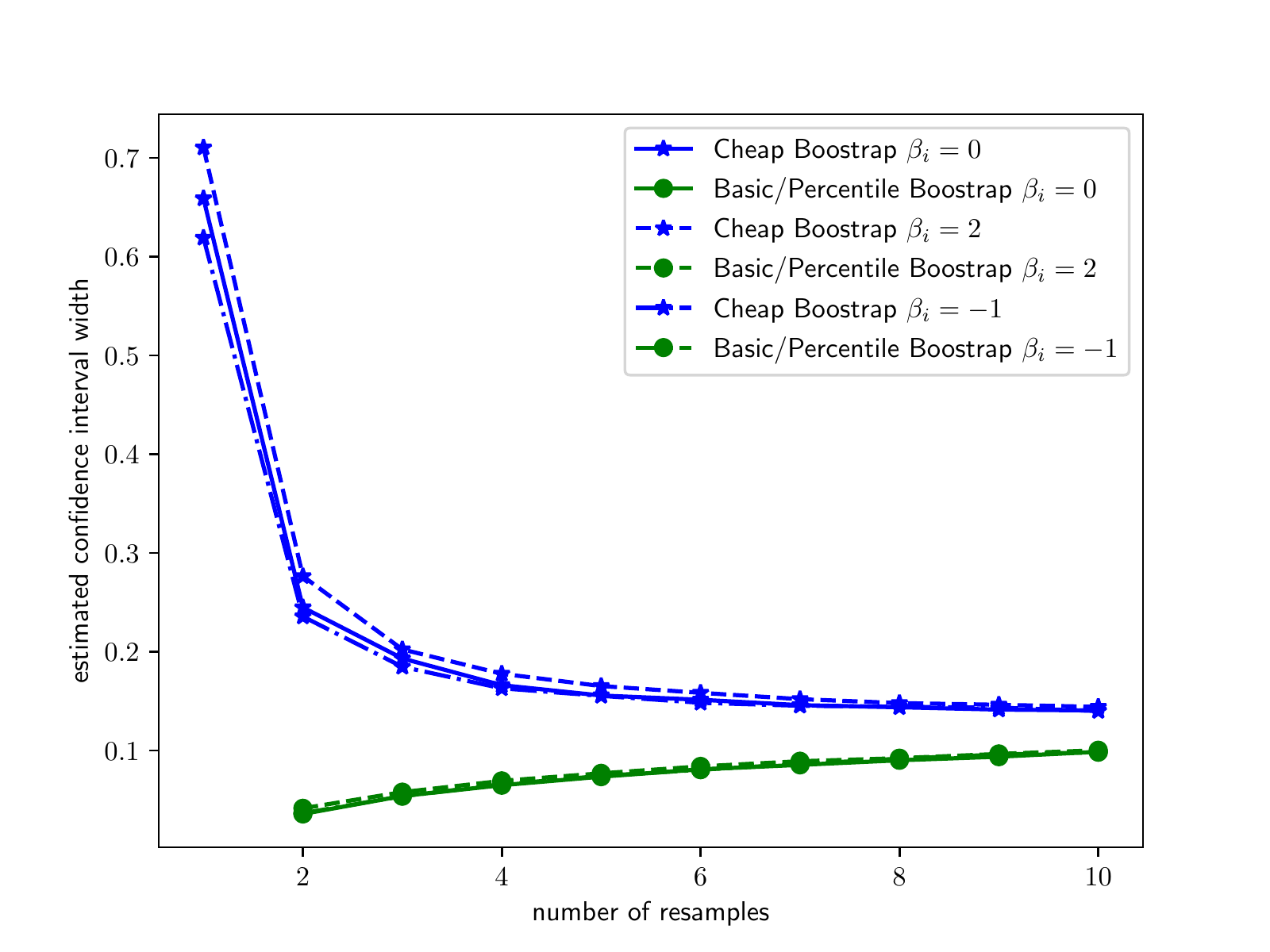}}

\caption{Empirical coverage probabilities and confidence interval widths for different numbers of resamples in a linear regression.}%
\label{figure_linear_vs_B}

\end{figure*}



\textit{Notation.} For a random vector $X$, we write $X^{k}$ as the tensor power $X^{\otimes k}$. The vector norm is taken as the usual Euclidean norm. The matrix and tensor norms are taken as the operator norm. For a square matrix $M$, $tr(M)$ denotes the trace of $M$. $I_{p\times p}$ denotes the identity matrix in $\mathbb{R}^{p\times p}$ and $\mathbf{1}_{p}$ denotes the vector in $\mathbb{R}^{p}$ whose components are all $1$. 
$\Phi$ denotes the cumulative distribution function of the standard normal. $C^2(\mathbb{R}^p)$ denotes the set of twice continuously differentiable functions on $\mathbb{R}^p$. Throughout the whole paper, we use $C>0$ (without subscripts) to denote a universal constant which may vary each time it appears. We use $C_{1},C_{2},\ldots$ to denote constants that could depend on other parameters and we will clarify their dependence when using them.


 

\section{Background on Bootstrap Methods\label{sec:method}}

We briefly review standard bootstrap methods and from there the recent cheap bootstrap. Suppose we are interested in estimating a target statistical quantity $\psi:=\psi(P_{X})$ where $\psi(\cdot):\mathcal{P}\mapsto\mathbb{R}$ is a functional defined on the probability measure space $\mathcal{P}$. Given i.i.d. data $X_{1},\ldots,X_{n}\in\mathbb{R}^{p}$ following the unknown distribution $P_{X}$, we denote the empirical distribution as $\hat{P}_{X,n}(\cdot):=(1/n)\sum_{i=1}^nI(X_i\in\cdot)$. A natural point estimator is $\hat{\psi}_{n}:=\psi(\hat{P}_{X,n})$.

To construct a confidence interval from $\hat\psi_n$, a typical beginning point is the distribution of $\hat{\psi}_{n}-\psi$ from which we can pivotize. As this distribution is unknown in general, the bootstrap idea is to approximate it using the resample counterpart, as if the empirical distribution was the true distribution. More concretely, conditional on $X_{1},\ldots,X_{n}$, we repeatedly, say for $B$ times, resample (i.e., sample with replacement) the data $n$ times to obtain resamples $\{X_{1}^{\ast b},\ldots,X_{n}^{\ast b}\},b=1,\ldots,B$. Denoting $\hat{P}_{X,n}^{\ast b}$ as the resample empirical distributions, we construct $B$ resample estimates $\hat\psi_{n}^{\ast b}:=\psi(\hat{P}_{X,n}^{\ast b})$. Then we use the $\alpha/2$ and $(1-\alpha/2)$-th quantiles of $\hat{\psi}_{n}^{\ast b}-\hat{\psi}_{n}$, called $q_{\alpha/2}$ and $q_{1-\alpha/2}$, to construct $[\hat{\psi}_{n}-q_{1-\alpha/2},\hat{\psi}_{n}-q_{\alpha/2}]$ as a $(1-\alpha)$-level confidence interval, which is known as the basic bootstrap (\citet{davison1997bootstrap} Section 5.2). Alternatively, we could also use the $\alpha/2$ and $(1-\alpha/2)$-th quantiles of $\hat{\psi}_{n}^{\ast b}$, say $q_{\alpha/2}'$ and $q_{1-\alpha/2}'$, to form $[q_{\alpha/2}',q_{1-\alpha/2}']$, which is known as the percentile bootstrap (\citet{davison1997bootstrap} Section 5.3). There are numerous other variants in the literature, such as studentization \cite{hall1988theoretical}, calibration or iterated bootstrap \cite{hall1986bootstrap,beran1987prepivoting}, and bias correction and acceleration \cite{efron1987better,diciccio1996bootstrap,diciccio1987bootstrap}, with the general goal of obtaining more accurate coverage.


All the above methods rely on the principle that $\hat{\psi}_{n}-\psi$ and $\hat{\psi}_{n}^{\ast}-\hat{\psi}_{n}$ (conditional on $X_{1},\ldots,X_{n}$) are close in distribution. Typically, this means that, with a $\sqrt n$-scaling, they both converge to the same normal distribution. In contrast, the cheap bootstrap proposed in \citet{lam2022cheap,lam2022conference} constructs a $(1-\alpha)$-level confidence interval via
\begin{equation}
\left[\hat{\psi}_{n}-t_{B,1-\alpha/2}S_{n,B},\hat{\psi}_{n}+t_{B,1-\alpha/2}S_{n,B}\right],\label{cheap CI}
\end{equation}
where $S_{n,B}^{2}=(1/B)\sum_{b=1}^{B}(\hat\psi_{n}^{\ast b}-\hat{\psi}_{n})^{2}$, and $t_{B,1-\alpha/2}$ is the $(1-\alpha/2)$-th quantile of $t_B$, the $t$-distribution with degree of freedom $B$.
The quantity $S_{n,B}^2$ resembles the sample variance of the resample estimates $\hat\psi_n^{*b}$'s, in the sense that as $B\to\infty$, $S_{n,B}^2$ approaches the bootstrap variance $Var_*(\hat\psi_n^*)$ (where $Var_*(\cdot)$ denotes the variance of a resample conditional on the data). In this way, \eqref{cheap CI} reduces to the normality interval with a ``plug-in'' estimator of the standard error term when $B$ and $n$ are both large. However, intriguingly, \emph{$B$ does not need to be large}, and $S_{n,B}^2$ is not necessarily close to the bootstrap variance. Instead, the idea is to consider the joint distribution of $\hat{\psi}_{n}-\psi,\hat\psi_{n}^{\ast1}-\hat{\psi}_{n},\ldots,\hat\psi_{n}^{\ast B}-\hat{\psi}_{n}$ that is argued to be asymptotically independent as $n\to\infty$ and $B$ fixed, which subsequently allows the construction of a pivotal $t$-statistic and gives rise to \eqref{cheap CI} for fixed $B$. A more detailed explanation on the cheap bootstrap in the low-dimensional case (i.e., $p$ is fixed) is given in Appendix \ref{sec:more_on_CB}.

\section{High-Dimensional Bootstrap Bounds\label{sec:main}}
As explained in the introduction, we aim to study the coverage performance of bootstraps in high-dimensional problems, focusing on the cheap bootstrap approach that allows the use of small computation effort.
We describe our results at three levels, first under general assumptions (Section \ref{sec:general bounds}), then more explicit bounds under the function-of-mean model and sub-Gaussianity of $X$ (Section \ref{sec:nonlinear}), finally bounds for a linear function under weaker tail assumptions on $X$ (Section \ref{sec:linear}). 

\subsection{General Finite-Sample Bounds\label{sec:general bounds}}



We have the following finite-sample bound for the cheap bootstrap: 

\begin{theorem}
\label{general_CB_coverage}Suppose we have the finite-sample accuracy for the estimator $\hat{\psi}_{n}$
\begin{equation}
\sup_{x\in\mathbb{R}}\left\vert P(\sqrt{n}(\hat{\psi}_{n}-\psi)\leq x)-\Phi(x/\sigma)  \right\vert \leq\mathcal{E}_{1},\label{general_Berry_Esseen}%
\end{equation}
and with probability at least $1-\beta$ we have the finite-sample accuracy for the bootstrap estimator $\hat{\psi}_{n}^{\ast}$%
\begin{equation}
\sup_{x\in\mathbb{R}}\left\vert P^{\ast}(\sqrt{n}(\hat{\psi}_{n}^{\ast}-\hat{\psi}_{n})\leq x)-\Phi(x/\sigma) \right\vert\leq\mathcal{E}_{2},\label{general_bootstrap_CLT}%
\end{equation}
where $\sigma>0$, $\mathcal{E}_{1}$ and $\mathcal{E}_{2}$ are deterministic quantities, and $P^*$ denotes the probability on a resample conditional on the data. Then the coverage error of \eqref{cheap CI} satisfies, for any $B\geq1$,%
\begin{align}
&\left\vert P(|\psi-\hat{\psi}_{n}|\leq t_{B,1-\alpha/2}S_{n,B})-(1-\alpha)\right\vert \nonumber\\
&\leq 2\mathcal{E}_{1}+2B\mathcal{E}_{2}+\beta.\label{error CB}
\end{align}
\end{theorem}



Condition (\ref{general_Berry_Esseen}) is a Berry-Esseen bound \cite{bentkus2003dependence,pinelis2016optimal} that gauges the normal approximation for the original estimate $\hat\psi_n$. Condition (\ref{general_bootstrap_CLT}) manifests a similar normal approximation for the resample estimate $\hat\psi_n^*$, and has been a focus in the high-dimensional CLT literature \cite{zhilova2020nonclassical,lopes2022central,chernozhukov2020nearly}. Both conditions (in their asymptotic form) are commonly used to establish the validity of standard bootstrap methods. Theorem \ref{general_CB_coverage} shows that, under these conditions, the coverage error of the cheap bootstrap interval \eqref{cheap CI} with any $B\geq1$ can be controlled. Note that Theorem \ref{general_CB_coverage} is very general in the sense that there is no direct assumption applied on the form of $\psi(\cdot)$ -- All we assume is approximate normality in the sense of (\ref{general_Berry_Esseen}) and (\ref{general_bootstrap_CLT}). Due to technical delicacies, in the bootstrap literature, finite-sample or higher-order coverage errors are typically obtainable only with specific models \cite{hall2013bootstrap,zhilova2020nonclassical,lopes2022central,chernozhukov2020nearly}, the most popular being the function-of-mean model (see Section \ref{sec:nonlinear}) or even simply the sample mean. In contrast, the bound in Theorem \ref{general_CB_coverage} that concludes the sufficiency in using a very small $B$ is a general statement that does not depend on the delicacies of $\psi(\cdot)$. Moreover, by plugging in suitable bounds for $\mathcal E_1,\mathcal E_2,\beta$ under regularity conditions, Theorem \ref{general_CB_coverage} also recovers the main result (Theorem 1) in \citet{lam2022cheap}. The detailed proof of Theorem \ref{general_CB_coverage} is in Appendix \ref{sec:proofs} (and so are the proofs of all other theorems). Below we give a sketch of the main argument.

\textbf{Proof sketch of Theorem \ref{general_CB_coverage}.} \emph{Step 1: }We write the coverage probability as the expected value (with respect to data) of a multiple integral with respect to the distributions of $\sqrt{n}(\hat{\psi}_{n}^{\ast}-\hat{\psi}_{n})$ (denoted by $Q^{\ast}$, conditional on data), i.e.,%
\begin{align}
&  P(|\psi-\hat{\psi}_{n}|\leq t_{B,1-\alpha/2}S_{n,B})\nonumber\\
&  =P\left(  \left\vert \frac{\sqrt{n}(\hat{\psi}_{n}-\psi)}{\sqrt{\frac{1}{B}\sum_{b=1}^{B}(\sqrt{n}(\hat{\psi}_{n}^{\ast b}-\hat{\psi}_{n}))^{2}}}\right\vert \leq t_{B,1-\alpha/2}\right) \nonumber\\
&  =E\left[  \int_{\frac{|\sqrt{n}(\hat{\psi}_{n}-\psi)|}{\sqrt{\sum_{b=1}^{B}z_{b}^{2}/B}}\leq t_{B,1-\alpha/2}}dQ^{\ast}(z_{B})\cdots dQ^{\ast}(z_{1})\right].\label{multiple_integral}%
\end{align}

\emph{Step 2: }Suppose (\ref{general_bootstrap_CLT}) happens and denote this event by $\mathcal{A}$ which satisfies $P(\mathcal{A}^{c})\leq\beta$. For each $b=1,\ldots,B$, given all other $z_{b^{\prime}},b^{\prime}\neq b$, the integration region is of the form $z_{b}\in(-\infty,-q]\cup\lbrack q,\infty)$ for some $q\geq0$. Then we can replace the distribution $Q^{\ast}$ by the distribution of $N(0,\sigma^{2})$ (denoted by $P_{0}$) with controlled error given in (\ref{general_bootstrap_CLT}) and obtain%
\begin{align}
&  E\left[  \int_{\frac{|\sqrt{n}(\hat{\psi}_{n}-\psi)|}{\sqrt{\sum_{b=1}^{B}z_{b}^{2}/B}}\leq t_{B,1-\alpha/2}}dQ^{\ast}(z_{B})\cdots dQ^{\ast}(z_{1})\right]  \nonumber\\
&  =E\left[  \int_{\frac{|\sqrt{n}(\hat{\psi}_{n}-\psi)|}{\sqrt{\sum_{b=1}^{B}z_{b}^{2}/B}}\leq t_{B,1-1-\alpha/2}}dP_{0}(z_{B})\cdots dP_{0}(z_{1})\right] \nonumber\\
&+R_{1},\label{sketch_expansion1}%
\end{align}
where $|R_{1}|\leq2B\mathcal{E}_{2}+\beta$ accounts for the error from (\ref{general_bootstrap_CLT}) and the small probability event $\mathcal{A}^{c}$.

\emph{Step 3: }Following the same logic in Step 2 and noticing that the integration region for $\sqrt{n}(\hat{\psi}_{n}-\psi)$ is $[-q,q]$ for some $q\geq0$, we can also replace the distribution of $\sqrt{n}(\hat{\psi}_{n}-\psi)$ by the distribution $P_{0}$ with controlled error $|R_{2}|\leq2\mathcal{E}_{1}$ according to (\ref{general_Berry_Esseen}):%
\begin{align}
&  E\left[  \int_{\frac{|\sqrt{n}(\hat{\psi}_{n}-\psi)|}{\sqrt{\sum_{b=1}^{B}z_{b}^{2}/B}}\leq t_{B,1-\alpha/2}}dP_{0}(z_{B})\cdots dP_{0}(z_{1})\right]  \nonumber\\
&  =\int_{\frac{|z_0|}{\sqrt{\sum_{b=1}^{B}z_{b}^{2}/B}}\leq t_{B,1-\alpha/2}}dP_{0}(z_{B})\cdots dP_{0}(z_{1})dP_{0}(z_{0})\nonumber\\
&+R_{2}\nonumber\\
&=1-\alpha+R_{2}.\label{sketch_expansion2}%
\end{align}

\emph{Step 4: }Plugging (\ref{sketch_expansion1}) and (\ref{sketch_expansion2}) back into (\ref{multiple_integral}), we can express the coverage probability as a sum of the nominal level and the remainder term:%
\[
P(|\psi-\hat{\psi}_{n}|\leq t_{B,1-\alpha/2}S_{n,B})=1-\alpha+R_{1}+R_{2}%
\]
with error $|R_{1}+R_{2}|\leq2\mathcal{E}_{1}+2B\mathcal{E}_{2}+\beta$. This gives our conclusion.\hfill\qed

Theorem \ref{general_CB_coverage} is designed to work well for small $B$ (our target scenario), but deteriorates when $B$ grows. However, in the latter case, we can strengthen the bound to cover the large-$B$ regime with additional conditions on the variance estimator (see Appendix \ref{sec:further_CB}).

We compare with standard basic and percentile bootstraps using $B=\infty$. Below is a generalization of \citet{zhilova2020nonclassical} which focuses only on the basic bootstrap under the function-of-mean model. 

\begin{theorem}
\label{general_other_coverages}Suppose the conditions in Theorem \ref{general_CB_coverage} hold. If $q_{\alpha/2}$, $q_{1-\alpha/2}$ are the $\alpha/2$-th and $(1-\alpha/2)$-th quantiles of $\hat{\psi}_{n}^{\ast}-\hat{\psi}_{n}$ respectively given $X_{1},\ldots,X_{n}$, then a finite-sample bound on the basic bootstrap coverage error is%
\begin{align}
&|P(\hat{\psi}_{n}-q_{1-\alpha/2}\leq\psi\leq\hat{\psi}_{n}-q_{\alpha/2})-(1-\alpha)|\nonumber\\
&\leq2\mathcal{E}_{1}+2\mathcal{E}_{2}+2\beta.\label{error basic}
\end{align}
If $q_{\alpha/2}'$, $q_{1-\alpha/2}'$ are the $\alpha/2$-th and $(1-\alpha/2)$-th quantiles of $\hat{\psi}_{n}^{\ast}$ respectively given $X_{1},\ldots,X_{n}$, then a finite-sample bound on the percentile bootstrap coverage error is%
\begin{equation}
|P(q_{\alpha/2}'\leq\psi\leq q_{1-\alpha/2}')-(1-\alpha)|\leq2\mathcal{E}_{1}+2\mathcal{E}_{2}+2\beta.\label{error percentile}
\end{equation}

\end{theorem}



In view of Theorems \ref{general_CB_coverage} and \ref{general_other_coverages}, the cheap bootstrap with any fixed $B$ can achieve the same order of coverage error bound as the basic and percentile bootstraps with $B=\infty$, in the sense that
\begin{equation}
(1/2)\ \text{EB}_\text{Quantile}\leq\text{EB}_\text{Cheap}\leq B\ \text{EB}_\text{Quantile},\label{bound comparison}
\end{equation}
where $\text{EB}_\text{Cheap}$ is the RHS error bound of \eqref{error CB} and $\text{EB}_\text{Quantile}$ is that of \eqref{error basic} or \eqref{error percentile}. This shows that, to attain a good coverage that is on par with standard basic/percentile bootstraps, it suffices to use the cheap bootstrap with a small $B$ which could save computation dramatically. 

Besides coverage, another important quality of confidence interval is its width. To this end, note that for any fixed $B$, (\ref{general_bootstrap_CLT}) ensures that $\sqrt{n}S_{n,B}\Rightarrow\sigma\sqrt{\chi_{B}^{2}/B}$ (unconditionally as $n\rightarrow\infty$ with proper  model assumptions) . Therefore, the half-width of (\ref{cheap CI}) is approximately $t_{B,1-\alpha/2}\sigma\sqrt{\chi_{B}^{2}/(nB)}$ with expected value 
\begin{equation}
E\left[  t_{B,1-\alpha/2}\sigma\sqrt{\frac{\chi_{B}^{2}}{nB}}\right]=t_{B,1-\alpha/2}\sigma\sqrt{\frac{2}{Bn}}\frac{\Gamma((B+1)/2)}{\Gamma(B/2)},\label{half_width}
\end{equation}
where $\Gamma(\cdot)$ is the gamma function. Since the dimensional impact is hidden in $\sigma$ which is a common factor in the expected width as $B$ varies, we can see $p$ does not affect the relative width behavior as $B$ changes. In particular, from \eqref{half_width} the inflation of the expected width relative to the case $B=\infty$ is 417.3\% for $B=1$, and dramatically reduces to 94.6\%, 24.8\% and 10.9\% for $B=2,5,10$, thus giving an interval with both correct coverage and short width using a small computation budget $B$.

In the next sections, we will apply Theorem \ref{general_CB_coverage} to obtain explicit bounds for specific high-dimensional models. Here, in relation to \eqref{bound comparison}, we briefly comment that the order of the coverage error bounds for these models is of order $1/\sqrt n$, both for the cheap bootstrap (which we will derive) and state-of-the-art high-dimensional bootstrap CLT. This is in contrast to the typical $1/n$ coverage error in two-sided bootstrap confidence intervals in low dimension (see \citet{hall2013bootstrap} Section 3.5 for quantile-based bootstraps and \citet{lam2022cheap} Section 3.2 for the cheap bootstrap).

\subsection{Function-of-Mean Models \label{sec:nonlinear}}

We now specialize to the function-of-mean model $\psi=g(\mu)$ for a mean vector $\mu=E[X]\in\mathbb{R}^{p}$ and smooth function $g:\mathbb{R}^{p}\mapsto\mathbb{R}$, which allows us to construct more explicit bounds. The original estimate $\hat{\psi}_{n}$ and resample estimate $\hat\psi_{n}^{\ast b}$ are now given by $g(\bar{X}_{n})$ and $g(\bar{X}_{n}^{\ast b})$ respectively, where $\bar X_n$ denotes the sample mean of data and $\bar X_n^{*b}$ denotes the resample mean of $X_1^{*b},\ldots,X_n^{*b}$. We assume:
\begin{assumption}\label{smoothness_assumption}
The function $g(x)\in C^2(\mathbb{R}^p)$ has Hessian matrix $H_{g}(x)$ with uniformly bounded eigenvalues, that is, $\exists$ a constant $C_{H_{g}}>0$ s.t. $\sup_{x\in\mathbb{R}^{p}}|a^{\top}H_{g}(x)a|\leq C_{H_{g}}||a||^{2},\forall a\in\mathbb{R}^{p}$.
\end{assumption}
\begin{assumption}\label{nonlinear_assumption}
$X$ is sub-Gaussian, i.e., there is a constant $\tau^{2}>0$ s.t. $E[\exp(a^{\top}(X-\mu))]\leq\exp(||a||^{2}\tau^{2}/2),\forall a\in \mathbb{R}^{p}$. %
Furthermore, $X$ has a density bounded by a constant $C_{X}$ and its covariance matrix $\Sigma$ is positive definite with the smallest eigenvalue $\lambda_{\Sigma}>0$.
\end{assumption}

Based on Theorem \ref{general_CB_coverage}, we derive the following explicit bound:

\begin{theorem}
\label{CB_coverage_nonlinear}Suppose the function $g$ satisfies Assumption \ref{smoothness_assumption} and random vector $X$ satisfies Assumption \ref{nonlinear_assumption}. Moreover, assume $||\nabla g(\mu)||>C_{\nabla g}\sqrt{p}$ for some constant $C_{\nabla g}>0$. Then we have
\begin{align*}
& \left\vert P(|g(\mu)-g(\bar{X}_{n})|\leq t_{B,1-\alpha/2}S_{n,B})  -(1-\alpha)\right\vert\\
&  \leq \frac{6}{n}+BC\left(\frac{m_{31}}{\sqrt{n}\sigma^{3}}+\frac{C_{H_{g}}m_{31}^{1/3}tr(\Sigma)}{\sqrt{n}\sigma^{2}}+\frac{C_{H_{g}}m_{32}^{2/3}}{n^{5/6}\sigma}\right. \\
& +\frac{C_{H_{g}}m_{31}^{1/3}m_{32}^{2/3}}{n\sigma^{2}}+\frac{C_{H_{g}}\tau^{2}}{C_{\nabla g}\sqrt{\lambda_{\Sigma}}}\left(1+\frac{\log n}{p}\right)  \sqrt{\frac{p}{n}}\\
&  +\frac{||E[(X-\mu)^{3}]||}{\lambda_{\Sigma}^{3/2}}\frac{1}{\sqrt{n}}+\frac{\tau^{3}}{\lambda_{\Sigma}^{3/2}}\left(  1+\frac{\log n}{p}\right)  ^{3/2}\frac{1}{\sqrt{n}}\\
&  \left.  +\frac{\tau^{4}\sqrt{p}}{\lambda_{\Sigma}^{2}n}\left(  1+\frac{\log n}{p}\right)  ^{1/2}+\frac{\tau^{2}\sqrt{p}}{\lambda_{\Sigma}n}\left(1+\frac{\log n}{p}\right)  ^{1/2}  \right. \\
& \left.+\frac{\tau^{3}\sqrt{p}}{\lambda_{\Sigma}^{3/2}n}\left(  1+\frac{\log n}{p}\right)\right)\\
&+ BC_1\left(\frac{\tau^{4}(\log n)^{3/2}}{\lambda_{\Sigma}^{2}\sqrt{n}}+\frac{\tau^{2}(\log n)^{3/2}}{\lambda_{\Sigma}\sqrt{n}}\right.\\
&  \left.+\frac{\tau^{3}}{\lambda_{\Sigma}^{3/2}\sqrt{n}}\left(  1+\frac{\log n}{p}\right)^{1/2}(\log n+\log p)\sqrt{\log n}\right)  ,
\end{align*}
where $m_{31}=E[|\nabla g(\mu)^{\top}(X-\mu)|^{3}]$, $m_{32}:=E[||X-\mu||^{3}]$, $\sigma^{2}=\nabla g(\mu)^{\top}\Sigma\nabla g(\mu)$, $C$ is a universal constant and $C_1$ is a constant only depending on $C_X$.
\end{theorem}

Theorem \ref{CB_coverage_nonlinear} is obtained by tracing the implicit quantities in Theorem \ref{general_CB_coverage} for the function-of-mean model, via extracting the dependence on problem parameters in the Berry-Esseen theorems for the multivariate delta method \cite{pinelis2016optimal} and the standard bootstrap \cite{zhilova2020nonclassical}. In particular, the sub-Gaussian assumption is required to derive finite-sample concentration inequalities, in a similar spirit as the state-of-the-art high-dimensional CLTs (e.g., \citet{chernozhukov2017central,lopes2022central}). On the other hand, the third moments such as $||E[(X-\mu)^3]||$ (operation norm of the third order tensor $E[(X-\mu)^3]$), $m_{31}$ and $m_{32}$ are due to the use of the Berry-Esseen theorem and a multivariate higher-order Berry-Esseen inequality in \citet{zhilova2020nonclassical}, which generally requires this order of moments. The bound in Theorem \ref{CB_coverage_nonlinear} can be simplified with reasonable assumptions on the involved quantities:

\begin{corollary}\label{concise_CB_coverage_nonlinear}
Suppose the conditions in Theorem \ref{CB_coverage_nonlinear} hold. Moreover, suppose that $\tau,C_{H_{g}},C_1=O(1)$, $C_{\nabla g},\lambda_{\Sigma}=\Theta(1)$, $||\nabla g(\mu)||^{2}=O(p)$, $\sigma^{2}=\Theta(p)$ and $||E[(X-\mu)^{3}]||=O(1)$. Then as $p,n\rightarrow\infty$,
\begin{align*}
& \left\vert P(|g(\mu)-g(\bar{X}_{n})|\leq t_{B,1-\alpha/2}S_{n,B})  -(1-\alpha)\right\vert\\
&  =B\times O\left(  \left(  1+\frac{\log n}{p}\right)  \sqrt{\frac{p}{n}}\right.\\
&\left.+\frac{1}{\sqrt{n}}\left(  1+\frac{\log n}{p}\right)  ^{1/2}(\log n+\log p)\sqrt{\log n}\right)  .
\end{align*}
Consequently, for any fixed $B\geq1$, the cheap bootstrap confidence interval is asymptotically exact provided $p=o(n)$, i.e.,
\[
\lim_{\substack{p,n\rightarrow\infty\\p=o(n)}}P(|g(\mu)-g(\bar{X}_{n})|\leq t_{B,1-\alpha/2}S_{n,B})=1-\alpha.
\]

\end{corollary}

In Corollary \ref{concise_CB_coverage_nonlinear}, the cheap bootstrap coverage error shrinks to $0$ as $n\rightarrow\infty$ if $p=o(n)$, i.e., the problem dimension grows slower than $n$ in any arbitrary fashion. Although there is no theoretical guarantee that the choice of $p=o(n)$ is tight, we offer numerical evidence in Section \ref{sec:num} where the cheap bootstrap works with a small $B$ when $p/n=0.09$ but it fails (i.e., over-covers the target with a quite large interval width) when $p/n=0.25$. Such a difference indicates that $p=o(n)$ can be tight in some cases. Recall that $||E[(X-\mu)^{3}]||$ denotes the operator norm of the third order tensor $E[(X-\mu)^{3}]$, and so the assumption $||E[(X-\mu)^{3}]||=O(1)$ holds if the components of $X$ are independent (or slightly weakly dependent). Other order assumptions in Corollary \ref{concise_CB_coverage_nonlinear} are natural. An example of the function-of-mean model is $g(\mu)=||\mu||^2$, used also in \citet{zhilova2020nonclassical}, whose confidence interval becomes a simultaneous confidence region for the mean vector $\mu$.

\subsection{Linear Functions\label{sec:linear}}


We consider a further specialization to linear $g$ where, instead of sub-Gaussanity of $X$, we are now able to use weaker tail conditions. Assume $g(x)=g_{1}^{\top}x+g_{2}$, where $g_{1}\in\mathbb{R}^{p}$ and $g_{2}\in\mathbb{R}$ are known. Then $g(\bar{X}_{n})$ and $g(\bar{X}_{n}^{\ast b})$ are essentially the sample mean and resample mean of i.i.d. random variables $g_{1}^{\top}X_{i}+g_{2},i=1,\ldots,n$. 

First, we consider the case where $g_{1}^{\top}(X-\mu)$ is sub-exponential, i.e., 
$||g_{1}^{\top}(X-\mu)||_{\psi_{1}}:=\inf\{\lambda>0:E[\psi_{1}(|g_{1}^{\top}(X-\mu)|/\lambda)]\leq1\}<\infty
$, 
where $||\cdot||_{\psi_{1}}$ is the Orlicz norm induced by the function $\psi_{1}(x)=e^{x}-1$. Sub-exponential property is a weaker tail condition than sub-Gaussianity; see e.g. \citet{vershynin2018high} Sections 2.5 and 2.7. Under this condition, we have:
\begin{theorem}
\label{CB_coverage_subexp}Suppose $g$ is a linear function in the form $g(x)=g_{1}^{\top}x+g_{2}$. Assume that $\sigma^{2}=g_{1}^{\top}\Sigma g_{1}>0$ and $||g_{1}^{\top}(X-\mu)||_{\psi_{1}}<\infty$. Then for any $n\geq3$, we have the following finite-sample bound on the cheap bootstrap coverage error%
\begin{align*}
& \left\vert P(|g(\mu)-g(\bar{X}_{n})|\leq t_{B,1-\alpha/2}S_{n,B})  -(1-\alpha)\right\vert\\
&  \leq \frac{C}{n}+BC  \frac{E[|g_{1}^{\top}(X-\mu)|^{3}]}{\sigma^{3}\sqrt{n}}\\
&+BC\frac{||g_{1}^{\top}(X-\mu)||_{\psi_{1}}^{4}\log^{11}(n)}{\sigma^{4}\sqrt{n}},
\end{align*}
where $C$ is a universal constant.
\end{theorem}
Note that the bootstrap in Theorem \ref{CB_coverage_subexp} effectively applies on the univariate $g_{1}^{\top}(X-\mu)$. Nonetheless, proving Theorem \ref{CB_coverage_subexp} requires tools from high-dimensional CLT \cite{lopes2022central,chernozhukov2020nearly}, as this appears the only line of work that investigates finite-sample bootstrap errors (for mean estimation). The order of the bound in terms of $p$ is controlled by $g_{1}^{\top}(X-\mu)$, and so if the latter is well-scaled by its standard deviation $\sigma$ in the sense that $E[|g_{1}^{\top}(X-\mu)/\sigma|^{3}],||g_{1}^{\top}(X-\mu)/\sigma||_{\psi_{1}}=O(1)$ (e.g., $X$ follows a multivariate normal distribution), then the order is independent of $p$, which means the error goes to $0$ for any $p$ as long as $n\rightarrow\infty$. However, if the orders of $E[|g_{1}^{\top}(X-\mu)/\sigma|^{3}]$ and  $||g_{1}^{\top}(X-\mu)/\sigma||_{\psi_{1}}$ depend on $p$, then the growing rate of $p$ must be restricted by $n$ to ensure the error goes to $0$.

Next, we further weaken the tail condition on $g_{1}^{\top}(X-\mu)$. We only assume $E[|g_{1}^{\top}(X-\mu)|^{q}]<\infty$ for some $q\geq4$. In this case, we have the following:

\begin{theorem}
\label{CB_coverage_moments}Suppose $g$ is a linear function in the form of $g(x)=g_{1}^{\top}x+g_{2}$. Assume that $\sigma^{2}=g_{1}^{\top}\Sigma g_{1}>0$ and $E[|g_{1}^{\top}(X-\mu)|^{q}]<\infty$ for some $q\geq4$. Then for any $n\geq3$, we have the following finite-sample bound on the cheap bootstrap coverage accuracy%
\begin{align*}
& \left\vert P(|g(\mu)-g(\bar{X}_{n})|\leq t_{B,1-\alpha/2}S_{n,B})  -(1-\alpha)\right\vert\\
&  \leq \frac{B  C_{1}\sqrt{\log n}}{n^{1/2-3/(2q)}}\max\left\{  E[|g_{1}^{\top}(X-\mu)/\sigma|^{q}]^{1/q},\right.\\
& \left.\sqrt{E[|g_{1}^{\top}(X-\mu)/\sigma|^{4}]}\right\}    +C\frac{E[|g_{1}^{\top}(X-\mu)|^{3}]}{\sigma^{3}\sqrt{n}},
\end{align*}
where $C$ is a universal constant and $C_{1}$ is a constant depending only on $q$.
\end{theorem}

The implication of Theorem \ref{CB_coverage_moments} on the choice of $p$ is similar to Theorem \ref{CB_coverage_subexp}. In parallel to the above, explicit finite-sample bounds for standard quantile-based bootstrap methods can also be obtained by means of Theorem \ref{general_other_coverages} under the assumptions in Theorems \ref{CB_coverage_nonlinear}, \ref{CB_coverage_subexp} or \ref{CB_coverage_moments} (see Appendix \ref{sec:explicit_others}). 

\section{Numerical Experiments\label{sec:num}}


We consider various high-dimensional examples:

\textbf{Ellipsoidal estimation:} The estimation target is $g(\mu)=||\mu||^2$, where $\mu$ is the mean of $X\in\mathbb R^p$ with ground-truth distribution $N(0.02\mathbf{1}_p,0.01I_{p\times p})$. Sample size $n=10^{5}$ and dimension $p=2.5\times10^{4}$. 

\textbf{Sinusoidal estimation:} The estimation target is $g(\mu)=\sum_{i=1}^{p}\sin(\mu_{i})$, where $\mu=(\mu_i)_{i=1}^{p}$ is the mean of $X\in\mathbb R^p$ with ground-truth distribution $N(\mathbf{0},0.01I_{p\times p})$. Sample size $n=10^{5}$ and dimension $p=2.5\times10^{4}$. 

\textbf{Linear regression with independent covariates:} Consider the true model $Y=X^{\top}\beta+\varepsilon$, where $X\in\mathbb{R}^{p}$ follows $N(\mathbf{0},0.01I_{p\times p})$ and $\varepsilon\sim N(0,1)$ independent of $X$. The first, second and last $1/3$ components of $\beta=(\beta_{i})_{i=1}^{p}$ are $0,2,-1$ respectively. We estimate $\beta$ given i.i.d. data $(X_{i},Y_{i})_{i=1}^{n}$ with $n=10^{5}$ and $p=9000$.


\textbf{Logistic regression with independent covariates:} Consider the true model $Y\in\{0,1\}$, $X\in\mathbb{R}^{p}$, $P(Y=1|X)=\exp(X^{\top}\beta)/(1+\exp(X^{\top}\beta))$, where $X\sim N(\mathbf{0},0.01I_{p\times p})$. The first $300$ components of $\beta_{i}$'s are $1$, the second $300$ components $-1$ and the rest $0$. As suggested in \citet{sur2019modern}, we choose such values of $\beta_{i}$'s to make sure $Var(X^{\top}\beta)=6$ does not increase with $p$ so that $P(Y=1|X)$ is not trivially equal to $0$ or $1$ in most cases. We estimate $\beta$ given i.i.d. data $(X_{i},Y_{i})_{i=1}^{n}$ with $n=10^{5}$ and $p=9000$.

\textbf{Stochastic simulation model:} Consider a stochastic computer communication model used to calculate the steady-state average message delay (\citet{cheng1997sensitivity,lin2015single,lam2022subsampling}; see Appendix \ref{sec:computer model} for details). This problem can be cast as computing $\psi(P_1,\ldots,P_p)$ where $\psi$ represents this expensive simulation model (due to the need to run very long time in order to reach steady state) and $P_j$'s denote the input distributions, $p=13$ in total. The data sizes for observing these 13 input models range from $3\times10^{4}$ to $6\times10^{4}$. 

\textbf{A real data example:} We run logistic regression on the RCV1-v2 data in \citet{lewis2004rcv1}. This dataset contains $n=804414$ manually categorized newswire stories with a total of $p=47236$ features. ``Economics'' (``ECAT'') is chosen as the $+1$ label. We target coefficient estimation.

\textbf{Linear regression with dependent covariates:} We consider the same linear regression setup as before but with two different distributions of $X$. One is $X\sim N(\mathbf{0},0.01\Sigma)$ where the components of $\Sigma$ are $\Sigma_{ij}=0.8^{|i-j|}$. For this distribution, the $i$-th component and $j$-th component of $X$ are more dependent when $i$ and $j$ are closer to each other. The other distribution is $X\sim N(\mathbf{0},0.01AA^{\top})$, where $A$ is a random matrix whose components are i.i.d. from $U(0,1)$. The two cases are referred to as ``exponential decay'' and ``random covariance matrix'' respectively.

\textbf{Logistic regression with dependent covariates:} We consider the same logistic regression setup as before but change the distribution of $X$ into the exponential decay distribution mentioned above. Here we do not consider the random covariance matrix case because the significant noisiness of $X$ makes $Var(X^{\top}\beta)$ so large that $P(Y=1|X)$ is trivially equal to $0$ or $1$ in most cases, which is also avoided in other work (e.g., \citet{sur2019modern}). 

\textbf{Ridge regression with $p>n$:} Consider the true linear model $Y=X^{\top}\beta+\varepsilon$, where $\varepsilon\sim N(0,1)$ is independent of $X$. The first, second and last $1/3$ components of $\beta=(\beta_{i})_{i=1}^{p}$ are $0,2,-1$ respectively. We consider three distributions of $X$. The first one is $X\sim N(\mathbf{0},0.01I_{p\times p})$ (referred to as ``independent''). The other two are exponential decay distribution and random covariance matrix mentioned above. Given i.i.d. data $(X_{i},Y_{i})_{i=1}^{n}$ with $n=8000$ and $p=9000$, we estimate $\beta$ by means of ridge regression which minimizes $||Y-X\beta||^2+\lambda||\beta||^2$. 

\textbf{(Regularized) logistic regression with varying dimensions $p$:} We use the same setup of logistic regression with independent covariates as before (e.g. $n=10^5$) but use $p\in\{12000,15000,18000,21000,25000\}$ to test the valid boundary of $p$ (i.e., the maximum $p$ that makes cheap bootstrap work) in this problem. Moreover, we also run a regularized version by adding the $\ell_2$ regularization term $||\beta||^2/2$ to the log likelihood function of logistic regression to see the effect of regularization.

\textbf{Setups and comparison benchmarks. }
In each example above, our targets are $95\%$-level confidence intervals for the target parameters. We test four bootstrap confidence intervals: 1) cheap bootstrap \eqref{cheap CI}; 2) basic bootstrap described in Section \ref{sec:method}; 3) percentile bootstrap described in Section \ref{sec:method}; 4) standard error bootstrap that uses standard normal quantile and standard deviation of $\hat\psi_{n}^{\ast b}$'s in lieu of $t_{B,1-\alpha/2}$ and $S_{n,B}$ respectively in \eqref{cheap CI}. For each setup except the real-data example, we run $1000$ experimental repetitions, each time generating a new data set from the ground truth distribution and construct the intervals. We report the empirical coverage and average interval width over these repetitions. For examples with more than one target estimation quantity, we further average the coverages and widths over all these targets. For the high-dimensional linear regression with independent covariates, we additionally show a box plot of the coverage probabilities and confidence interval widths of each individual $\beta_i$. We vary the number of resamples $B$ from 1 to 10 in all examples and report the running time (i.e., model fitting time for one point estimate and $B$ bootstrap estimates; the time for outputting the confidence intervals using these estimates is negligible compared to the model fitting time) in the virtual machine e2-highmem-2 in Google Cloud Platform. Some examples have larger scale and thus are run in the virtual machine e2-highmem-8 with larger memory and better CPU, whose running time will be starred (*).



\textbf{Results and discussions. }Tables \ref{num_table}-\ref{table_different_p_l2} and Figure \ref{figure_linear} describe our results (Tables \ref{dependent_table}-\ref{table_different_p_l2} are delegated to Appendix \ref{sec:tables} due to the limitation of space), where we report $B=1,2,5,10$. ``CB'', ``BB'', ``PB'' and ``SEB'' in Figure \ref{figure_linear} stand for the cheap bootstrap, basic bootstrap, percentile bootstrap and standard error bootstrap respectively.

\emph{Coverage probability: }According to Table \ref{num_table}, the cheap bootstrap performs the best in terms of the coverage probabilities in almost all cases (except the real-data example where we cannot validate and only report the interval widths). In all but three entries, the cheap bootstrap gives the closest coverages to the nominal 95\% among all considered bootstrap methods, and in all but three entries the cheap bootstrap coverages are above 95\%. In contrast, other approaches are substantially below the nominal level except for very few cases with $B=10$. For example, in the ellipsoidal estimation, cheap bootstrap coverage probabilities are above 95\% for all considered $B$'s, while the highest coverage among other bootstrap methods is 82.1\% even for $B=10$. These observations corroborate with theory since unlike standard bootstrap methods, the cheap bootstrap gives small coverage errors even with very small $B$. Note that when $B=1$, the entries of other bootstrap methods are all ``N.A.'' since quantile-based approaches cannot even output two distinct finite numbers using one resample, and standard error bootstrap uses $B-1$ in the denominator of the sample variance. Similar results are also observed from Tables \ref{dependent_table} and \ref{table_ridge}, which confirm that even for distributions with dependent components and ridge regression with $p>n$, the cheap bootstrap continues to work and outperform other bootstrap methods. In terms of the individual plot in Figure \ref{figure_linear}, the cheap bootstrap has coverages close to the nominal level 95\% for almost all $\beta_i$'s for any $B$. On the other hand, standard error bootstrap coverages are above 90\% only when $B=10$ while the quantile-based bootstrap coverages are still below 85\% for most of the $\beta_i$'s even for $B=10$. 

\emph{Interval width: }Cheap bootstrap intervals are wider than other bootstrap intervals. However, these widths appear to decay very fast for the first few $B$'s. In all examples, they decrease by around $2/3$ from $B=1$ to $B=2$ and by around $4/5$ from $B=1$ to $B=10$, hinting a quick enhancement of statistical efficiency as the computation budget increases. Note that while the other bootstrap intervals are shorter, they fall short in attaining the nominal coverage. It is thus reasonable to see the larger widths of the cheap bootstrap intervals which appear to push up the coverages by the right amounts.

\emph{Valid boundary of $p$:} Table \ref{table_different_p} displays the results of logistic regression with varying $p$. We observe that the confidence widths have a dramatic increase when $p$ increases to 21000 from 18000, which hints that the validity boundary of $p$ under this setting should be at most 21000. Moreover, even for $p=15000$ and $p=18000$, the coverage probabilities are almost 100\% which is overly large compared to the nominal level 95\%. Combining with our existing favorable results for $p=9000$, for this problem it appears the validity boundary could be indeed around 9000. As a comparison, Table \ref{table_different_p_l2} displays the results of regularized logistic regression. We can see although the interval widths seem to converge with the regularization, the coverage probabilities are still overly large for large $p$. For this regularized example, the validity boundary could be around 12000 or 15000, which is larger than the previous boundary. This contrast shows that regularization could be helpful to achieve a valid confidence interval when $p$ is large. But no matter whether we add the regularization, it seems the validity boundary is always around the order $p=o(n)$ for logistic regression. The results here, combined with the favorable results for ridge regression with $p>n$ in Table \ref{table_ridge}, provide evidence that the precise validity threshold of $p$ could be model-dependent.



\begin{figure}[ht]
\centering

\subfloat[Coverage probability]{\includegraphics[width=0.45\textwidth]{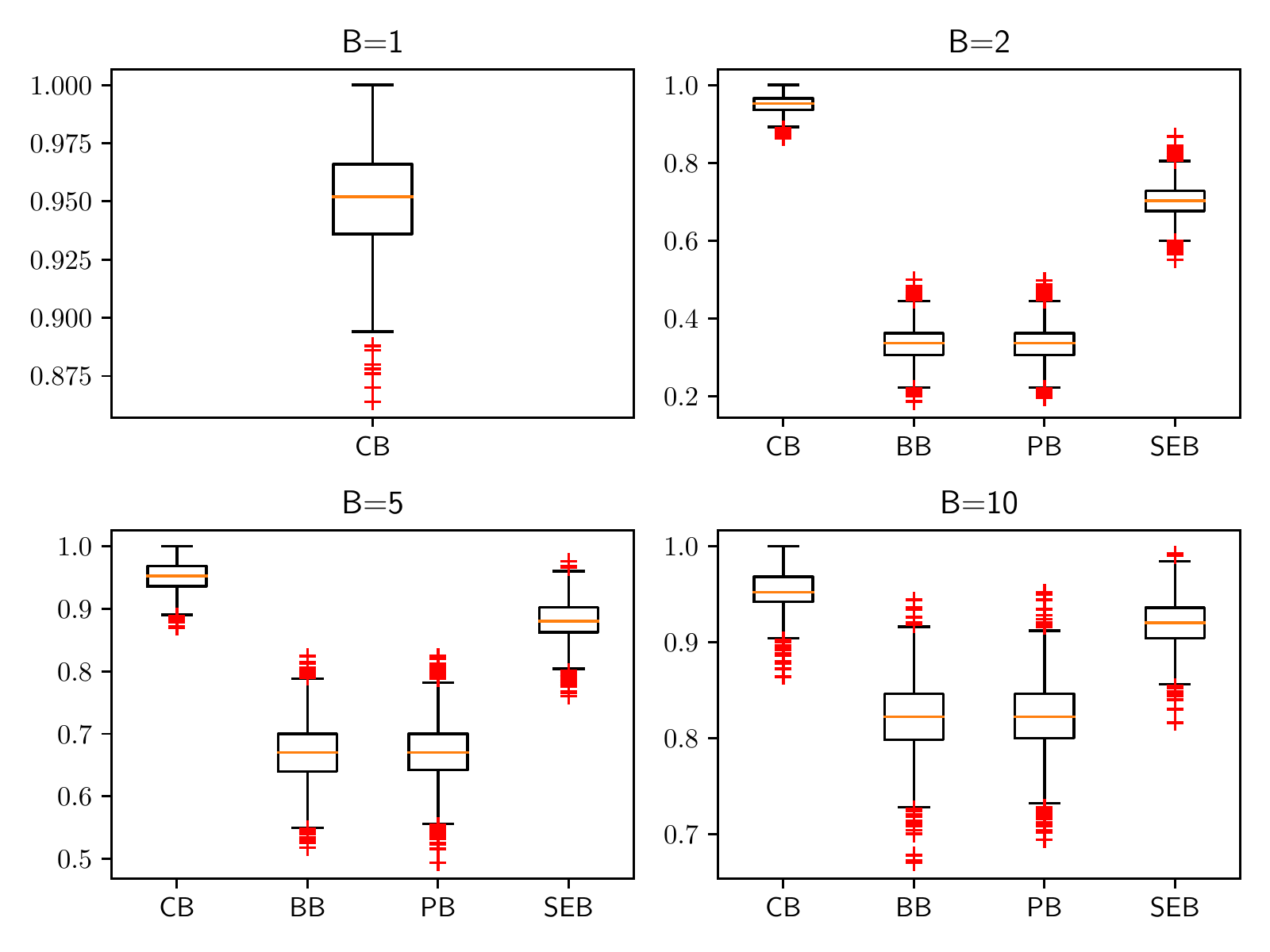}}\\
\subfloat[Confidence interval width]{\includegraphics[width=0.45\textwidth]{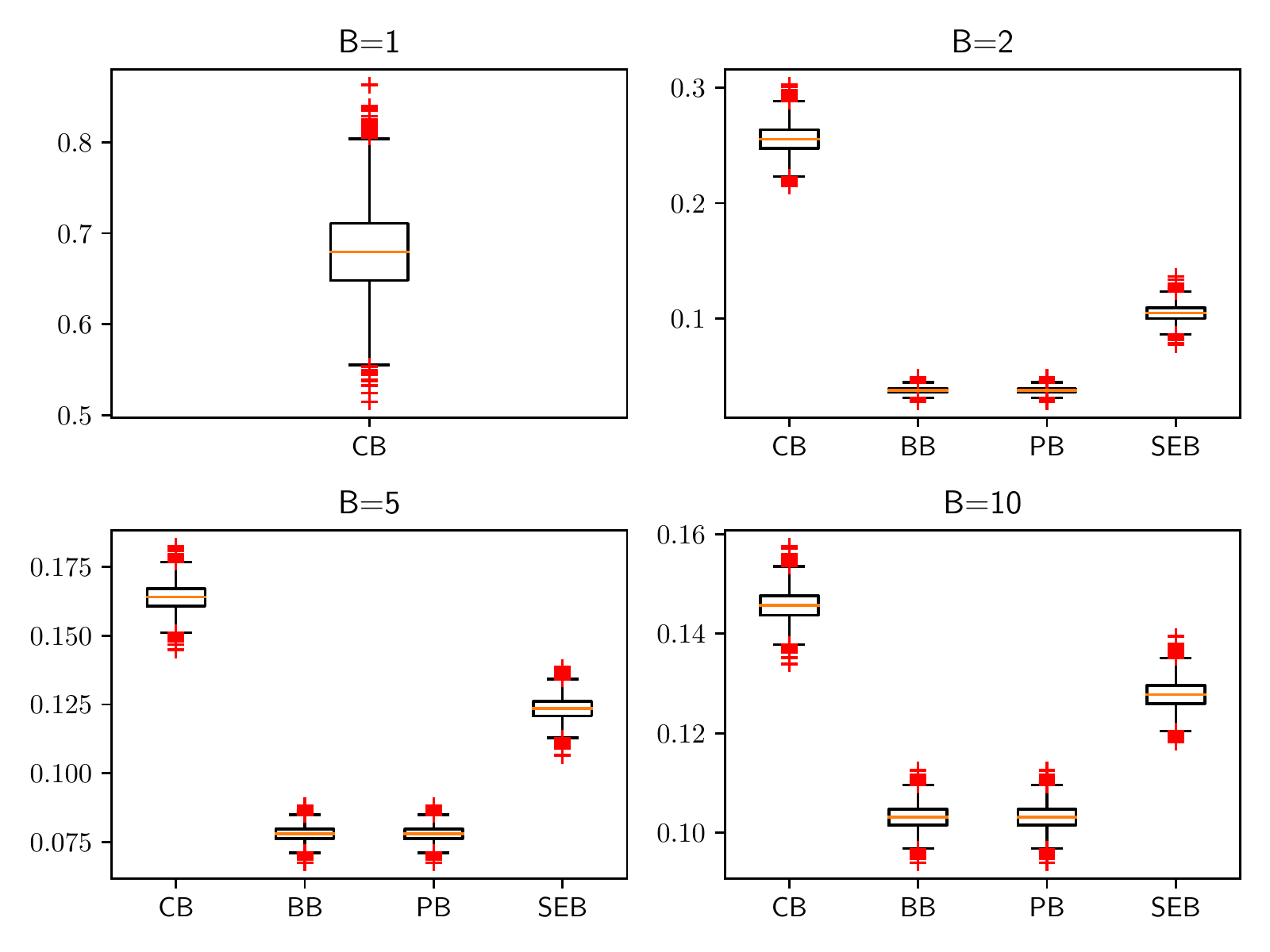}}

\caption{Box plots of empirical coverage probabilities and confidence interval widths of all $\beta_i$'s for different numbers of resamples in a linear regression.}%
\label{figure_linear}

\end{figure}

\begin{table*}[ht]
\centering
\caption{Coverage probabilities (Pro.), confidence interval widths (Wid.) and running time (unit: second) of the first six numerical examples. The closest coverage probability to the nominal 95\% level among all methods in each setting is \textbf{bold}.}
\medskip
\label{num_table}
\resizebox{2\columnwidth}{!}{
\begin{tblr}{
  cells = {c},
  cell{1}{1} = {r=2}{},
  cell{1}{2} = {r=2}{},
  cell{1}{3} = {c=2}{},
  cell{1}{5} = {c=2}{},
  cell{1}{7} = {c=2}{},
  cell{1}{9} = {c=2}{},
  cell{1}{11} = {r=2}{},
  hline{1,27} = {-}{0.08em},
  hline{3,7,11,15,19,23} = {1-10}{0.03em},
  hline{3,7,11,15,19,23} = {11}{},
}
Example      & $B$ & Cheap Bootstrap &                     & Basic Bootstrap &                     & Percentile Bootstrap &                     & Standard Error Bootstrap &                     & Running Time \\
             &     & Pro.            & Wid.                & Pro.            & Wid.                & Pro.                 & Wid.                & Pro.                     & Wid.                &              \\
~            & 1   & \textbf{96.0\%} & 0.069               & N.A.            & N.A.                & N.A.                 & N.A.                & N.A.                     & N.A.                & 9*             \\
Ellipsoidal  & 2   & \textbf{97.3\%} & 0.026               & 32.2\%          & 0.002               & 5.5\%                & 0.002               & 55.1\%                   & 0.006               & 15*             \\
estimation   & 5   & \textbf{97.4\%} & 0.016               & 66.0\%          & 0.005               & 13.6\%               & 0.005               & 70.1\%                   & 0.007               & 36*             \\
             & 10  & \textbf{97.5\%} & 0.014               & 82.1\%          & 0.006               & 20.8\%               & 0.006               & 73.6\%                   & 0.008               & 70*             \\
             & 1   & \textbf{94.4\%} & 0.958               & N.A.            & N.A.                & N.A.                 & N.A.                & N.A.                     & N.A.                & 7*             \\
Sinusoidal   & 2   & \textbf{95.2\%} & 0.384               & 29.6\%          & 0.051               & 35.2\%               & 0.051               & 63.2\%                   & 0.142               & 12*             \\
estimation   & 5   & \textbf{93.6\%} & 0.248               & 71.2\%          & 0.117               & 66.4\%               & 0.117               & 86.4\%                   & 0.187               & 27*             \\
             & 10  & \textbf{94.4\%} & 0.222               & 84.0\%          & 0.156               & 83.2\%               & 0.156               & 89.6\%                   & 0.196               & 52*             \\
Linear       & 1   & \textbf{95.1\%} & 0.68                & N.A.            & N.A.                & N.A.                 & N.A.                & N.A.                     & N.A.                & 443          \\
regression   & 2   & \textbf{95.1\%} & 0.256               & 33.5\%          & 0.038               & 33.5\%               & 0.038               & 70.2\%                   & 0.105               & 666          \\
(independent & 5   & \textbf{95.2\%} & 0.164               & 67.0\%          & 0.078               & 67.0\%               & 0.078               & 88.1\%                   & 0.123               & 1337         \\
covariates)  & 10  & \textbf{95.2\%} & 0.146               & 82.2\%          & 0.103               & 82.2\%               & 0.103               & 92.1\%                   & 0.128               & 2454         \\
Logistic     & 1   & \textbf{96.1\%} & 2.866               & N.A.            & N.A.                & N.A.                 & N.A.                & N.A.                     & N.A.                & 50           \\
regression   & 2   & \textbf{96.9\%} & 1.074               & 39.7\%          & 0.147               & 31.7\%               & 0.147               & 73.4\%                   & 0.407               & 81           \\
(independent & 5   & \textbf{97.9\%} & 0.685               & 77.9\%          & 0.302               & 63.3\%               & 0.302               & 91.0\%                   & 0.479               & 175          \\
covariates)  & 10  & 98.4\%          & 0.609               & 91.9\%          & 0.400               & 77.7\%               & 0.400               & \textbf{94.6\%}          & 0.496               & 331          \\
             & 1   & \textbf{96.9\%} & 1.757$\times 10^{-3}$ & N.A.            & N.A.                & N.A.                 & N.A.                & N.A.                     & N.A.                & 1             \\
Stochastic   & 2   & \textbf{98.8\%} & 6.417$\times 10^{-4}$ & 21.9\%          & 6.962$\times 10^{-5}$ & 47.0\%               & 6.962$\times 10^{-5}$ & 68.7\%                   & 1.930$\times 10^{-4}$ & 2             \\
simulation   & 5   & 99.7\%          & 4.044$\times 10^{-4}$ & 43.2\%          & 1.428$\times 10^{-4}$ & \textbf{90.4\%}      & 1.428$\times 10^{-4}$ & 87.1\%                   & 2.269$\times 10^{-4}$ & 3             \\
             & 10  & 100\%           & 3.591$\times 10^{-4}$ & 55.6\%          & 1.915$\times 10^{-4}$ & 99.8\%               & 1.915$\times 10^{-4}$ & \textbf{92.6\%}          & 2.375$\times 10^{-4}$ & 5             \\
             & 1   & N.A.            & 3.594               & N.A.            & N.A.                & N.A.                 & N.A.                & N.A.                     & N.A.                & 156          \\
Real         & 2   & N.A.            & 1.361               & N.A.            & 0.201               & N.A.                 & 0.201               & N.A.                     & 0.556               & 233          \\
data         & 5   & N.A.            & 0.877               & N.A.            & 0.414               & N.A.                 & 0.414               & N.A.                     & 0.658               & 464          \\
             & 10  & N.A.            & 0.779               & N.A.            & 0.547               & N.A.                 & 0.547               & N.A.                     & 0.682               & 849          
\end{tblr}
}
\end{table*}

\section{Discussions and Other Connections}
In this paper, we show how to run the bootstrap that achieves statistically valid coverage with very low cost, in terms of the number of resamples, even when the problem dimension grows closely with the data size. This is made possible by using sample-resample independence from a recent ``cheap" bootstrap perspective. We provide general finite-sample coverage error bounds to support our validity, and specialize these bounds to explicit models including function-of-mean models with sub-Gaussian tails and linear models with weaker tail conditions. We discuss how our approach can operate with as low as one resample in attaining valid interval coverage in large sample. At the same time, the interval constructed with one resample tends to be wide, but fortunately shrinks quickly as the number of resamples increases from one even slightly. We run a wide set of numerical experiments to validate our theory and show the outperformance of our method compared to other benchmarks. As our numerical experiments hint that the tight growth rate of $p$ in terms of $n$ could be model-dependent, a future investigation is to establish more precise finite-sample bounds that are tight in some appropriate sense.


We close this paper by positioning our results in the broader literature. First, our work is related to bootstrap coverage analysis. The commonest approach is to use the Edgeworth expansion that reveals the asymptotic higher-order terms in the coverage errors; see the comprehensive monograph \cite{hall2013bootstrap}. It is only until recently when finite-sample bounds on bootstrap appear, mostly in the high-dimensional CLT literature where the target is sample mean (\citet{chernozhukov2017central,lopes2022central,chernozhukov2020nearly} and references therein). They aim to prove a uniform finite-sample bound of normal approximation of the (bootstrap) sample mean over all hyperrectangles. An alternative approach is to use Stein's method \cite{fang2021high}.

Second, within the bootstrap framework, various approaches have been proposed to reduce the Monte Carlo sampling effort by, e.g., variance reduction such as importance sampling \cite{booth1993simple}, or analytic approximation especially when applying iterated bootstraps \cite{booth1994monte,lee1995asymptotic}. These methods, however, require additional knowledge such as an explicit way to calculate variance, or focus on tail estimation issue. The closest work to the cheap bootstrap idea we utilize in this paper is \citet{hall1986number} who investigates the number of resamples for one-sided bootstrap intervals. Nonetheless, \citet{hall1986number} suggests a minimum of 19 for $B$ in a 95\%-level interval, obtained via an order-statistics calculation, which is still much larger than the minimal choice $B=1$ in the cheap bootstrap.

\section*{Acknowledgement}
We gratefully acknowledge support from the National Science Foundation under grants CAREER CMMI-1834710 and IIS-1849280.

\bibliography{highdim_cheap.bib}
\bibliographystyle{icml2023}

\newpage
\appendix
\onecolumn

\section{A More Detailed Explanation on the Cheap Bootstrap in the Low-Dimensional Case}\label{sec:more_on_CB}
This section aims to give more details on the cheap bootstrap method in the low-dimensional case discussed in Section \ref{sec:method}. 

Suppose we are interested in estimating a target statistical quantity $\psi:=\psi(P_{X})$ where $\psi(\cdot):\mathcal{P}\mapsto\mathbb{R}$ is a functional defined on the probability measure space $\mathcal{P}$. Given i.i.d. data $X_{1},\ldots,X_{n}\in\mathbb{R}^{p}$ following the unknown distribution $P_{X}$, we denote the empirical distribution as $\hat{P}_{X,n}(\cdot):=(1/n)\sum_{i=1}^nI(X_i\in\cdot)$. A natural point estimator is $\hat{\psi}_{n}:=\psi(\hat{P}_{X,n})$. 

The cheap bootstrap confidence interval for $\psi$ is constructed as follows. Conditional on $X_{1},\ldots,X_{n}$, we independently resample, i.e., sample with replacement, the data for $B$ times to obtain resamples $\{X_{1}^{\ast b},\ldots,X_{n}^{\ast b}\},b=1,\ldots,B$. Denoting $\hat{P}_{X,n}^{\ast b}$ as the resample empirical distributions, we construct $B$ resample estimates $\hat\psi_{n}^{\ast b}:=\psi(\hat{P}_{X,n}^{\ast b})$. A $(1-\alpha)$-level confidence interval is then given by
\begin{equation}
\left[\hat{\psi}_{n}-t_{B,1-\alpha/2}S_{n,B},\hat{\psi}_{n}+t_{B,1-\alpha/2}S_{n,B}\right],\label{cheap CI Appendix}
\end{equation}
where $S_{n,B}^{2}=(1/B)\sum_{b=1}^{B}(\hat\psi_{n}^{\ast b}-\hat{\psi}_{n})^{2}$, and $t_{B,1-\alpha/2}$ is the $(1-\alpha/2)$-th quantile of $t_B$, the $t$-distribution with degree of freedom $B$. Theorem 1 in \citet{lam2022cheap} shows that, under conditions on par with standard bootstrap methods, \eqref{cheap CI Appendix} is an asymptotically exact $(1-\alpha)$-level confidence interval \emph{for any fixed $B\geq1$}, i.e., 
\begin{equation}
P(\psi\in[\hat{\psi}_{n}-t_{B,1-\alpha/2}S_{n,B},\hat{\psi}_{n}+t_{B,1-\alpha/2}S_{n,B}])\to1-\alpha, \text{\ \ as\ \ }n\to\infty\label{asymptotic exact}
\end{equation}
where $P$ is the probability with respect to both the data and the randomness in the resampling process. 

Here we explain the asymptotic argument that gives rise to \eqref{asymptotic exact}. Under suitable conditions, the sampling distribution of an estimate $\hat\psi_n$ and the distribution of a resample estimate $\hat\psi_n^*$ are approximately equal. More formally, they are equal in the asymptotic sense of two CLTs $\sqrt{n}(\hat{\psi}_{n}-\psi)\overset{d}{\rightarrow}N(0,\sigma^{2})$ for some $\sigma^{2}>0$, and $\sqrt{n}(\hat{\psi}_{n}^{\ast}-\hat{\psi}_{n})\overset{d}{\rightarrow}N(0,\sigma^{2})$  (conditional on $X_{1},\ldots,X_{n}$ in probability) for the same $\sigma^{2}$. By means of a conditional argument, we can combine the two aforementioned CLTs to obtain the following joint convergence%
\begin{equation}
\sqrt{n}(\hat{\psi}_{n}-\psi,\hat\psi_{n}^{\ast1}-\hat{\psi}_{n},\ldots,\hat\psi_{n}^{\ast B}-\hat{\psi}_{n})\overset{d}{\rightarrow}(\sigma Z_{0},\sigma Z_{1},\ldots,\sigma Z_{B}), \text{\ \ as\ \ }n\to\infty \label{joint_convergence}%
\end{equation}
where $Z_{b},b=0,\ldots,B$ are i.i.d. standard normal. From \eqref{joint_convergence}, we can establish the convergence of a pivotal $t$-statistic $(\hat\psi_n-\psi)/S_{n,B}\overset{d}{\rightarrow}t_B$ which gives \eqref{asymptotic exact}. The above shows that, with $B$ fixed as small as $1$, \eqref{cheap CI Appendix} already offers a coverage close to the nominal level as $n\to\infty$. In this argument, the approximation accuracy of $(\hat\psi_n-\psi)/S_{n,B}$ by the $t_B$ random variable is crucial. However, in the high-dimensional case when $p\rightarrow\infty$ as $n\rightarrow\infty$, the joint CLT (\ref{joint_convergence}) may not hold and thus the techniques in this paper are needed to establish the validity of the cheap bootstrap method.

\section{Additional Theoretical Results}

This section provides additional theoretical results. Appendix \ref{sec:further_CB} establishes an alternative finite-sample bound for the cheap bootstrap that generalizes Theorem \ref{general_CB_coverage} to cover the large-$B$ regime. Section \ref{sec:explicit_others} provides finite-sample bounds for standard quantile-based bootstrap methods under the conditions in Sections \ref{sec:nonlinear} and \ref{sec:linear}.

\subsection{Further Finite-Sample Bound for the Cheap Bootstrap\label{sec:further_CB}}

The following result generalizes Theorem \ref{general_CB_coverage} to include both the small and large-$B$ regimes:
\begin{theorem}
\label{general_CB_coverage2}Suppose we have the finite sample accuracy for the estimator $\hat{\psi}_{n}$
\begin{equation}
\sup_{x\in\mathbb{R}}\left\vert P(\sqrt{n}(\hat{\psi}_{n}-\psi)\leq x)-\Phi\left(  \frac{x}{\sigma}\right)  \right\vert \leq\mathcal{E}_{1},\label{appendix:general_Berry_Esseen}%
\end{equation}
and with probability at least $1-\beta$ we have the finite sample accuracy for the bootstrap estimator $\hat{\psi}_{n}^{\ast}$%
\[
\sup_{x\in\mathbb{R}}\left\vert P^{\ast}(\sqrt{n}(\hat{\psi}_{n}^{\ast}-\hat{\psi}_{n})\leq x)-\Phi\left(  \frac{x}{\sigma}\right)  \right\vert\leq\mathcal{E}_{2},
\]
where $\sigma>0$ and $P^{\ast}$ denotes the probability on a resample conditional on the data. Further, suppose that the following concentration inequality holds%
\begin{equation}
P\left(  \left\vert \sqrt{\frac{1}{B}\sum_{b=1}^{B}(\sqrt{n}(\hat{\psi}_{n}^{\ast b}-\hat{\psi}_{n}))^{2}}-\sigma\right\vert \geq\mathcal{E}_{3}\right)  \leq\mathcal{E}_{4},\label{concentration_sigma}%
\end{equation}
where $\mathcal{E}_{3}$ is deterministic and $\sigma-\mathcal{E}_{3}>0$. Then we have the following finite sample bound on the cheap bootstrap coverage error%
\begin{align*}
| &  P(\psi\in\lbrack\hat{\psi}_{n}-t_{B,1-\alpha/2}S_{n,B},\hat{\psi}_{n}+t_{B,1-\alpha/2}S_{n,B}])-(1-\alpha)|\\
&  \leq\min\left\{  2\mathcal{E}_{1}+2B\mathcal{E}_{2}+\beta,2\mathcal{E}_{1}+2\mathcal{E}_{4}+\sqrt{\frac{2}{\pi}}|t_{B,1-\alpha/2}-z_{1-\alpha/2}|+\sqrt{\frac{2}{\pi}}\frac{\mathcal{E}_{3}}{\sigma}t_{B,1-\alpha/2}\right\}  ,
\end{align*}
where $z_{1-\alpha/2}$ is the $(1-\alpha/2)$-quantile of the standard normal.
\end{theorem}

The finite sample accuracy in Theorem \ref{general_CB_coverage2} consists of two parts. The first one $2\mathcal{E}_{1}+2B\mathcal{E}_{2}+\beta$ works well when $B$ is small as shown in Sections \ref{sec:nonlinear} and \ref{sec:linear} but it deteriorates when $B$ grows. In contrast, the second part
\begin{equation}
2\mathcal{E}_{1}+2\mathcal{E}_{4}+\sqrt{\frac{2}{\pi}}|t_{B,1-\alpha/2}-z_{1-\alpha/2}|+\sqrt{\frac{2}{\pi}}\frac{\mathcal{E}_{3}}{\sigma}t_{B,1-\alpha/2}\label{second_bound}%
\end{equation}
vanishes as $B,n\rightarrow\infty$ but does not if $B$ is bounded even if $n\rightarrow\infty$. Its behavior for bounded $B$ is easy to see: The third term $\sqrt{2/\pi}|t_{B,1-\alpha/2}-z_{1-\alpha/2}|$ in (\ref{second_bound}) is bounded away from zero if $B$ is bounded and thus (\ref{second_bound}) never converges to zero even if $n\rightarrow\infty$. To explain why (\ref{second_bound}) vanishes as $B,n\rightarrow\infty$, first note that the first term $2\mathcal{E}_{1}$ is independent of $B$ and satisfies $\mathcal{E}_{1}\rightarrow0$ as $n\rightarrow\infty$ by the Berry-Esseen Theorem for a reasonable model $\psi(\cdot)$ such as the function-of-mean model in Section \ref{sec:nonlinear}. Second, notice that $\sqrt{(1/B)\sum_{b=1}^{B}(\sqrt{n}(\hat{\psi}_{n}^{\ast b}-\hat{\psi}_{n}))^{2}}$ is the bootstrap estimator of the asymptotic standard deviation $\sigma$. Therefore, (\ref{concentration_sigma}) is the concentration inequality for the bootstrap principle applied to the estimation of $\sigma$ and would hold with a choice of $\mathcal{E}_{3}$ and $\mathcal{E}_{4}$ satisfying $\ \mathcal{E}_{3}\rightarrow0,\mathcal{E}_{4}\rightarrow0$ as $B,n\rightarrow\infty$. Lastly, since $t_{B}\overset{d}{\rightarrow}N(0,1)$ as $B\rightarrow\infty$, by Lemma 21.2 in \citet{van2000asymptotic}, we have $t_{B,1-\alpha/2}\rightarrow z_{1-\alpha/2}$ as $B\rightarrow\infty$. Therefore, we can see the second part  (\ref{second_bound}) converges to zero as $B,n\rightarrow\infty$ at any rate.

Under concrete assumptions on $X$ as in Sections \ref{sec:nonlinear} and \ref{sec:linear}, explicit forms of $\mathcal{E}_{3}$ and $\mathcal{E}_{4}$ depending on $B$, $n$ and the distribution of $X$ can be derived, based on similar arguments as the explicit bounds in Theorems \ref{CB_coverage_nonlinear}-\ref{CB_coverage_moments}. Then by studying the order of these explicit bounds with respect to $B,p$ and $n$, we can deduce a proper growth rate of dimension $p=p(B,n)$ which ensures a vanishing error as $B,n\rightarrow\infty$. The concentration inequality (\ref{concentration_sigma}) seems unexplored in the literature and we leave it as future work.

\subsection{Explicit Finite-Sample Bounds for Quantile-Based Bootstrap Methods\label{sec:explicit_others}}

In this section, in parallel to Theorems \ref{CB_coverage_nonlinear}-\ref{CB_coverage_moments} for the cheap bootstrap, we provide a few explicit bounds for standard quantile-based bootstrap methods under the same conditions.

The first result is in parallel to Theorem \ref{CB_coverage_nonlinear} under the function-of-mean model:

\begin{theorem}
\label{other_coverage_nonlinear}Suppose the conditions in Theorem \ref{CB_coverage_nonlinear} hold. If $q_{\alpha/2}$, $q_{1-\alpha/2}$ are the $\alpha/2$-th and $(1-\alpha/2)$-th quantiles of $g(\bar{X}_{n}^{\ast})-g(\bar{X}_{n})$ respectively given $X_{1},\ldots,X_{n}$, then the finite-sample bound on the basic bootstrap coverage error is given by%
\begin{align*}
&  |P(g(\bar{X}_{n})-q_{1-\alpha/2}\leq g(\mu)\leq g(\bar{X}_{n})-q_{\alpha/2})-(1-\alpha)|\\
&  \leq \frac{12}{n}+C\left(  \frac{m_{31}}{\sqrt{n}\sigma^{3}}+\frac{C_{H_{g}}m_{31}^{1/3}tr(\Sigma)}{\sqrt{n}\sigma^{2}}+\frac{C_{H_{g}}m_{32}^{2/3}}{n^{5/6}\sigma}+\frac{C_{H_{g}}m_{31}^{1/3}m_{32}^{2/3}}{n\sigma^{2}}\right.\\
&  +\frac{C_{H_{g}}\tau^{2}}{C_{\nabla g}\sqrt{\lambda_{\Sigma}}}\left(1+\frac{\log n}{p}\right)  \sqrt{\frac{p}{n}}+\frac{||E[(X-\mu)^{3}]||}{\lambda_{\Sigma}^{3/2}}\frac{1}{\sqrt{n}}+\frac{\tau^{3}}{\lambda_{\Sigma}^{3/2}}\left(  1+\frac{\log n}{p}\right)  ^{3/2}\frac{1}{\sqrt{n}}\\
&  \left.  +\frac{\tau^{4}\sqrt{p}}{\lambda_{\Sigma}^{2}n}\left(  1+\frac{\log n}{p}\right)  ^{1/2}+\frac{\tau^{2}\sqrt{p}}{\lambda_{\Sigma}n}\left(1+\frac{\log n}{p}\right)  ^{1/2}+\frac{\tau^{3}\sqrt{p}}{\lambda_{\Sigma}^{3/2}n}\left(  1+\frac{\log n}{p}\right)  \right)  \\
&  +C_{1}\left(  \frac{\tau^{4}(\log n)^{3/2}}{\lambda_{\Sigma}^{2}\sqrt{n}}+\frac{\tau^{2}(\log n)^{3/2}}{\lambda_{\Sigma}\sqrt{n}}+\frac{\tau^{3}}{\lambda_{\Sigma}^{3/2}\sqrt{n}}\left(  1+\frac{\log n}{p}\right)^{1/2}(\log n+\log p)\sqrt{\log n}\right)  ,
\end{align*}
where $C$ is a universal constant and $C_{1}$ is a constant only depending on $C_{X}$. If $q^{\prime}_{\alpha/2}$, $q^{\prime}_{1-\alpha/2}$ are the $\alpha/2$-th and $(1-\alpha/2)$-th quantiles of $g(\bar{X}_{n}^{\ast})$ respectively given $X_{1},\ldots,X_{n}$, then the finite-sample bound on the percentile bootstrap coverage error is given by%
\begin{align*}
&  |P(q^{\prime}_{\alpha/2}\leq g(\mu)\leq q^{\prime}_{1-\alpha/2})-(1-\alpha)|\\
&  \leq \frac{12}{n}+C\left(  \frac{m_{31}}{\sqrt{n}\sigma^{3}}+\frac{C_{H_{g}}m_{31}^{1/3}tr(\Sigma)}{\sqrt{n}\sigma^{2}}+\frac{C_{H_{g}}m_{32}^{2/3}}{n^{5/6}\sigma}+\frac{C_{H_{g}}m_{31}^{1/3}m_{32}^{2/3}}{n\sigma^{2}}\right.\\
&  +\frac{C_{H_{g}}\tau^{2}}{C_{\nabla g}\sqrt{\lambda_{\Sigma}}}\left(1+\frac{\log n}{p}\right)  \sqrt{\frac{p}{n}}+\frac{||E[(X-\mu)^{3}]||}{\lambda_{\Sigma}^{3/2}}\frac{1}{\sqrt{n}}+\frac{\tau^{3}}{\lambda_{\Sigma}^{3/2}}\left(  1+\frac{\log n}{p}\right)  ^{3/2}\frac{1}{\sqrt{n}}\\
&  \left.  +\frac{\tau^{4}\sqrt{p}}{\lambda_{\Sigma}^{2}n}\left(  1+\frac{\log n}{p}\right)  ^{1/2}+\frac{\tau^{2}\sqrt{p}}{\lambda_{\Sigma}n}\left(1+\frac{\log n}{p}\right)  ^{1/2}+\frac{\tau^{3}\sqrt{p}}{\lambda_{\Sigma}^{3/2}n}\left(  1+\frac{\log n}{p}\right)  \right)  \\
&  +C_{1}\left(  \frac{\tau^{4}(\log n)^{3/2}}{\lambda_{\Sigma}^{2}\sqrt{n}}+\frac{\tau^{2}(\log n)^{3/2}}{\lambda_{\Sigma}\sqrt{n}}+\frac{\tau^{3}}{\lambda_{\Sigma}^{3/2}\sqrt{n}}\left(  1+\frac{\log n}{p}\right)^{1/2}(\log n+\log p)\sqrt{\log n}\right)  ,
\end{align*}
where $C$ is a universal constant and $C_{1}$ is a constant only depending on $C_{X}$.
\end{theorem}

Our discussion below Theorem \ref{general_other_coverages} shows that the cheap bootstrap error bound with any given $B$ and quantile-based bootstrap error bounds with $B=\infty$ only differ up to a constant. Therefore, the order analysis for the cheap bootstrap in Corollary \ref{concise_CB_coverage_nonlinear} also applies here, that is, under the conditions in Corollary \ref{concise_CB_coverage_nonlinear}, the quantile-based bootstrap coverage errors shrink to $0$ as $n\rightarrow\infty$ if $p=o(n)$.

The second result is in parallel to Theorem \ref{CB_coverage_subexp} under the sub-exponential assumption and linearity of $g$:

\begin{theorem}
\label{other_coverage_subexp}Suppose the conditions in Theorem \ref{CB_coverage_subexp} hold. If $q_{\alpha/2}$, $q_{1-\alpha/2}$ are the $\alpha/2$-th and $(1-\alpha/2)$-th quantiles of $g(\bar{X}_{n}^{\ast})-g(\bar{X}_{n})$ respectively given $X_{1},\ldots,X_{n}$, then the finite-sample bound on the basic bootstrap coverage error is given by%
\begin{align*}
&  |P(g(\bar{X}_{n})-q_{1-\alpha/2}\leq g(\mu)\leq g(\bar{X}_{n})-q_{\alpha/2})-(1-\alpha)|\\
&  \leq C\left(  \frac{1}{n}+\frac{E[|g_{1}^{\top}(X-\mu)|^{3}]}{\sigma^{3}\sqrt{n}}+\frac{||g_{1}^{\top}(X-\mu)||_{\psi_{1}}^{4}\log^{11}(n)}{\sigma^{4}\sqrt{n}}\right)  +\frac{CE[|g_{1}^{\top}(X-\mu)|^{3}]}{\sigma^{3}\sqrt{n}},
\end{align*}
where $C$ is a universal constant. If $q^{\prime}_{\alpha/2}$, $q^{\prime}_{1-\alpha/2}$ are the $\alpha/2$-th and $(1-\alpha/2)$-th quantiles of $g(\bar{X}_{n}^{\ast})$ respectively given $X_{1},\ldots,X_{n}$, then the finite-sample bound on the percentile bootstrap coverage error is given by%
\begin{align*}
&  |P(q^{\prime}_{\alpha/2}\leq g(\mu)\leq q^{\prime}_{1-\alpha/2})-(1-\alpha)|\\
&  \leq C\left(\frac{1}{n}+\frac{E[|g_{1}^{\top}(X-\mu)|^{3}]}{\sigma^{3}\sqrt{n}}+\frac{||g_{1}^{\top}(X-\mu)||_{\psi_{1}}^{4}\log^{11}(n)}{\sigma^{4}\sqrt{n}}\right) +\frac{CE[|g_{1}^{\top}(X-\mu)|^{3}]}{\sigma^{3}\sqrt{n}},
\end{align*}
where $C$ is a universal constant.
\end{theorem}

The last result is in parallel to Theorem \ref{CB_coverage_moments} under moment conditions and linearity of $g$:

\begin{theorem}
\label{other_coverage_moments}Suppose the conditions in Theorem \ref{CB_coverage_moments} hold. If $q_{\alpha/2}$, $q_{1-\alpha/2}$ are the $\alpha/2$ and $(1-\alpha/2)$-quantiles of $g(\bar{X}_{n}^{\ast})-g(\bar{X}_{n})$ respectively given $X_{1},\ldots,X_{n}$, then the finite-sample bound on the basic bootstrap coverage error is given by
\begin{align*}
&  |P(g(\bar{X}_{n})-q_{1-\alpha/2}\leq g(\mu)\leq g(\bar{X}_{n})-q_{\alpha/2})-(1-\alpha)|\\
&  \leq\frac{2}{\sqrt{n}}+C_{1}\max\left\{  E[|g_{1}^{\top}(X-\mu)/\sigma|^{q}]^{1/q},\sqrt{E[|g_{1}^{\top}(X-\mu)/\sigma|^{4}]}\right\}  \left(\frac{(\log n)^{3/2}}{\sqrt{n}}+\frac{\sqrt{\log n}}{n^{1/2-3/(2q)}}\right)\\
&  +\frac{CE[|g_{1}^{\top}(X-\mu)|^{3}]}{\sigma^{3}\sqrt{n}},
\end{align*}
where $C$ is a universal constant and $C_{1}$ is a constant depending only on $q$. If $q^{\prime}_{\alpha/2}$, $q^{\prime}_{1-\alpha/2}$ are the $\alpha/2$ and $(1-\alpha/2)$-quantiles of $g(\bar{X}_{n}^{\ast})$ respectively given $X_{1},\ldots,X_{n}$, then the finite-sample bound on the percentile bootstrap coverage error is given by%
\begin{align*}
&  |P(q^{\prime}_{\alpha/2}\leq g(\mu)\leq q^{\prime}_{1-\alpha/2})-(1-\alpha)|\\
&  \leq\frac{2}{\sqrt{n}}+C_{1}\max\left\{  E[|g_{1}^{\top}(X-\mu)/\sigma|^{q}]^{1/q},\sqrt{E[|g_{1}^{\top}(X-\mu)/\sigma|^{4}]}\right\}  \left(\frac{(\log n)^{3/2}}{\sqrt{n}}+\frac{\sqrt{\log n}}{n^{1/2-3/(2q)}}\right)\\
&  +\frac{CE[|g_{1}^{\top}(X-\mu)|^{3}]}{\sigma^{3}\sqrt{n}},
\end{align*}
where $C$ is a universal constant and $C_{1}$ is a constant depending only on $q$.
\end{theorem}

The order analysis for the cheap bootstrap also applies to Theorems \ref{other_coverage_subexp} and \ref{other_coverage_moments}. If $g_{1}^{\top}(X-\mu)$ is well-scaled by its standard deviation $\sigma$ in the sense that the $L_p$ norm and Orlicz norm $||\cdot||_{\psi_1}$ is independent of $p$, then the errors shrink to $0$ for any $p$ as $n\rightarrow\infty$. Otherwise, the growth rate of $p$ should depend on $n$ to obtain a vanishing error.

\section{Details of Numerical Experiments and Additional Numerical Results}

In this section, we present additional results and details of the experiments in Section \ref{sec:num}. We also report some additional experiments. Section \ref{sec:tables} presents tables for experimental results in Section \ref{sec:num} that have not been shown in previous sections. The following subsections, \ref{sec:logistic regression}, \ref{sec:computer model} and \ref{sec:real data} provide additional details for logistic regression with independent covariates, the computer network and the real world problem presented in Section \ref{sec:num}. Section \ref{sec:lower_level} further validates our performances by a simulation study with a lower nominal level 70\%. Finally, Section \ref{sec:trend} studies the coverage error behavior as $B$ and $n$ vary for the sinusoidal estimation.

\subsection{Additional Tables}\label{sec:tables}

Tables \ref{dependent_table}-\ref{table_different_p_l2} reports the experimental results of regression problems with dependent covariates, ridge regression and (regularized) logistic regression with different $p$ presented in Section \ref{sec:num} respecitvely.

\begin{table}[ht]
\centering
\caption{Coverage probabilities (Pro.), confidence interval widths (Wid.) and running time (unit: second) of the regression problems with dependent covariates. The closest coverage probability to the nominal 95\% level among all methods in each setting is \textbf{bold}.}
\label{dependent_table}
\medskip
\resizebox{\columnwidth}{!}{
\begin{tblr}{
  cells = {c},
  cell{1}{1} = {r=2}{},
  cell{1}{2} = {r=2}{},
  cell{1}{3} = {c=2}{},
  cell{1}{5} = {c=2}{},
  cell{1}{7} = {c=2}{},
  cell{1}{9} = {c=2}{},
  cell{1}{11} = {r=2}{},
  hline{1,3,7,11,15} = {-}{},
}
Example      & $B$ & Cheap Bootstrap &         & Basic Bootstrap &        & Percentile Bootstrap &        & Standard Error Bootstrap &        & Running Time \\
             &     & Pro.            & Wid.    & Pro.            & Wid.   & Pro.                 & Wid.   & Pro.                     & Wid.   &              \\
Linear       & 1   & \textbf{95.1\%} & 1.451   & N.A.            & N.A.   & N.A.                 & N.A.   & N.A.                     & N.A.   & 167*            \\
regression   & 2   & \textbf{95.1\%} & 0.546   & 33.6\%          & 0.081  & 33.5\%               & 0.081  & 70.3\%                   & 0.224  & 258*             \\
(exponential & 5   & \textbf{95.2\%} & 0.350   & 67.1\%          & 0.166  & 67.1\%               & 0.166  & 88.2\%                   & 0.264  & 531*             \\
decay)       & 10  & \textbf{95.2\%} & 0.311   & 82.2\%          & 0.220  & 82.2\%               & 0.220  & 92.1\%                   & 0.273  & 987*             \\
Linear       & 1   & \textbf{94.7\%} & 129.255 & N.A.            & N.A.   & N.A.                 & N.A.   & N.A.                     & N.A.   & 511          \\
regression   & 2   & \textbf{94.9\%} & 52.791  & 33.9\%          & 7.575  & 34.5\%               & 7.575  & 69.9\%                   & 20.995 & 798          \\
(random      & 5   & \textbf{94.8\%} & 30.733  & 67.8\%          & 14.875 & 66.8\%               & 14.875 & 87.5\%                   & 23.555 & 1659         \\
cov matrix)  & 10  & \textbf{95.0\%} & 27.048  & 82.1\%          & 19.398 & 82.1\%               & 19.398 & 91.9\%                   & 23.787 & 3095         \\
Logistic     & 1   & \textbf{95.3\%} & 7.411   & N.A.            & N.A.   & N.A.                 & N.A.   & N.A.                     & N.A.   & 114*             \\
regression   & 2   & \textbf{95.5\%} & 2.787   & 34.9\%          & 0.404  & 32.3\%               & 0.404  & 70.7\%                   & 1.119  & 175*             \\
(exponential & 5   & \textbf{95.8\%} & 1.787   & 69.5\%          & 0.832  & 64.8\%               & 0.832  & 88.6\%                   & 1.318  & 360*             \\
decay)       & 10  & \textbf{96.0\%} & 1.588   & 84.6\%          & 1.101  & 79.8\%               & 1.101  & 92.6\%                   & 1.364  & 667*             
\end{tblr}
}
\end{table}

\begin{table}[ht]
\centering
\caption{Coverage probabilities (Pro.), confidence interval widths (Wid.) and running time (unit: second) of the ridge regression with $p=9000$ and $n=8000$. The closest coverage probabilities to the nominal 95\% level among all methods are \textbf{bold}.}
\label{table_ridge}
\medskip
\resizebox{\columnwidth}{!}{
\begin{tblr}{
  cells = {c},
  cell{1}{1} = {r=2}{},
  cell{1}{2} = {r=2}{},
  cell{1}{3} = {c=2}{},
  cell{1}{5} = {c=2}{},
  cell{1}{7} = {c=2}{},
  cell{1}{9} = {c=2}{},
  cell{1}{11} = {r=2}{},
  hline{1,3,7,11,15,19,23,27} = {-}{},
}
Example                & $B $ & Cheap Bootstrap &        & Basic Bootstrap &       & Percentile Bootstrap &       & Standard Error Bootstrap &       & Running Time \\
                       &      & Pro.            & Wid.   & Pro.            & Wid.  & Pro.                 & Wid.  & Pro.                     & Wid.  &              \\
Ridge                  & 1    & \textbf{96.7\%} & 15.572 & N.A.            & N.A.  & N.A.                 & N.A.  & N.A.                     & N.A.  & 48           \\
regression             & 2    & \textbf{97.2\%} & 5.719  & 23.1\%          & 0.553 & 25.5\%               & 0.553 & 68.8\%                   & 1.534 & 72           \\
(independent;          & 5    & \textbf{97.5\%} & 3.595  & 47.0\%          & 1.141 & 51.8\%               & 1.141 & 86.6\%                   & 1.807 & 144          \\
$\lambda=0.1)$         & 10   & \textbf{97.6\%} & 3.171  & 60.2\%          & 1.510 & 66.0\%               & 1.510 & 90.7\%                   & 1.870 & 264          \\
Ridge                  & 1    & \textbf{96.3\%} & 13.734 & N.A.            & N.A.  & N.A.                 & N.A.  & N.A.                     & N.A.  & 48           \\
regression             & 2    & \textbf{96.7\%} & 5.082  & 27.1\%          & 0.539 & 24.8\%               & 0.539 & 68.7\%                   & 1.495 & 72           \\
(independent;          & 5    & \textbf{97.0\%} & 3.212  & 54.8\%          & 1.111 & 50.3\%               & 1.111 & 86.5\%                   & 1.761 & 144          \\
$\lambda=1)$           & 10   & \textbf{97.1\%} & 2.839  & 69.2\%          & 1.471 & 64.3\%               & 1.471 & 90.6\%                   & 1.822 & 264          \\
Ridge                  & 1    & \textbf{96.9\%} & 9.588  & N.A.            & N.A.  & N.A.                 & N.A.  & N.A.                     & N.A.  & 47           \\
regression             & 2    & \textbf{97.7\%} & 3.453  & 11.6\%          & 0.263 & 37.3\%               & 0.263 & 48.0\%                   & 0.728 & 71           \\
(exponential           & 5    & \textbf{98.4\%} & 2.143  & 23.8\%          & 0.541 & 73.9\%               & 0.541 & 59.9\%                   & 0.857 & 142          \\
decay; $\lambda=0.1)$   & 10   & \textbf{98.7\%} & 1.882  & 31.2\%          & 0.716 & 88.5\%               & 0.716 & 62.8\%                   & 0.887 & 260          \\
Ridge                  & 1    & \textbf{95.7\%} & 5.776  & N.A.            & N.A.  & N.A.                 & N.A.  & N.A.                     & N.A.  & 48           \\
regression             & 2    & \textbf{96.3\%} & 2.145  & 19.5\%          & 0.240 & 36.2\%               & 0.240 & 61.6\%                   & 0.665 & 71           \\
(exponential           & 5    & \textbf{96.9\%} & 1.360  & 39.8\%          & 0.495 & 71.8\%               & 0.495 & 78.0\%                   & 0.784 & 142          \\
decay; $\lambda=1)$     & 10   & \textbf{97.2\%} & 1.203  & 51.5\%          & 0.654 & 86.7\%               & 0.654 & 81.9\%                   & 0.811 & 261          \\
Ridge                  & 1    & \textbf{96.7\%} & 15.232 & N.A.            & N.A.  & N.A.                 & N.A.  & N.A.                     & N.A.  & 49           \\
regression             & 2    & \textbf{96.8\%} & 5.527  & 19.4\%          & 0.462 & 21.2\%               & 0.462 & 64.4\%                   & 1.281 & 74           \\
(random cov            & 5    & \textbf{96.5\%} & 3.448  & 39.6\%          & 0.953 & 43.4\%               & 0.953 & 81.5\%                   & 1.510 & 151          \\
matrix;~$\lambda=0.1)$ & 10   & \textbf{96.3\%} & 3.032  & 51.2\%          & 1.261 & 56.4\%               & 1.261 & 85.6\%                   & 1.563 & 279          \\
Ridge                  & 1    & \textbf{96.3\%} & 13.302 & N.A.            & N.A.  & N.A.                 & N.A.  & N.A.                     & N.A.  & 50           \\
regression             & 2    & \textbf{96.4\%} & 4.873  & 23.5\%          & 0.457 & 20.9\%               & 0.457 & 65.2\%                   & 1.267 & 76           \\
(random cov            & 5    & \textbf{96.2\%} & 3.060  & 47.9\%          & 0.943 & 42.8\%               & 0.943 & 82.5\%                   & 1.493 & 155          \\
matrix;~$\lambda=1)$   & 10   & \textbf{95.9\%} & 2.696  & 61.2\%          & 1.247 & 55.7\%               & 1.247 & 86.6\%                   & 1.545 & 286          
\end{tblr}
}
\end{table}

\begin{table}[ht]
\centering
\caption{Coverage probabilities (Pro.), confidence interval widths (Wid.) and running time (unit: second) of the logistic regression with $n=10^5$ and different $p$. The closest coverage probability to the nominal 95\% level among all methods in each setting is \textbf{bold}.}
\label{table_different_p}
\medskip
\resizebox{\columnwidth}{!}{
\begin{tblr}{
  cells = {c},
  cell{1}{1} = {r=2}{},
  cell{1}{2} = {r=2}{},
  cell{1}{3} = {c=2}{},
  cell{1}{5} = {c=2}{},
  cell{1}{7} = {c=2}{},
  cell{1}{9} = {c=2}{},
  cell{1}{11} = {r=2}{},
  cell{3}{1} = {r=4}{},
  cell{7}{1} = {r=4}{},
  cell{11}{1} = {r=4}{},
  cell{15}{1} = {r=4}{},
  cell{19}{1} = {r=4}{},
  hline{1,3,7,11,15,19,23} = {-}{},
}
$p $  & $B $ & Cheap Bootstrap &          & Basic Bootstrap &         & Percentile Bootstrap &         & Standard Error Bootstrap &         & Running Time \\
      &      & Pro.            & Wid.     & Pro.            & Wid.    & Pro.                 & Wid.    & Pro.                     & Wid.    &              \\
12000 & 1    & \textbf{96.7\%} & 3.697    & N.A.            & N.A.    & N.A.                 & N.A.    & N.A.                     & N.A.    & 46*          \\
      & 2    & \textbf{97.7\%} & 1.378    & 42.8\%          & 0.182   & 32.3\%               & 0.182   & 75.9\%                   & 0.504   & 76*          \\
      & 5    & 99.0\%          & 0.878    & 83.1\%          & 0.375   & 64.6\%               & 0.375   & \textbf{93.1\%}          & 0.594   & 165*         \\
      & 10   & 99.5\%          & 0.778    & \textbf{95.3\%} & 0.496   & 79.0\%               & 0.496   & 96.2\%                   & 0.615   & 314*         \\
15000 & 1    & \textbf{97.5\%} & 5.368    & N.A.            & N.A.    & N.A.                 & N.A.    & N.A.                     & N.A.    & 84*          \\
      & 2    & \textbf{98.7\%} & 1.993    & 46.3\%          & 0.251   & 33.3\%               & 0.251   & 80.2\%                   & 0.697   & 142*         \\
      & 5    & 99.8\%          & 1.268    & 88.0\%          & 0.518   & 66.4\%               & 0.518   & \textbf{95.9\%}          & 0.821   & 316*         \\
      & 10   & 100\%           & 1.123    & \textbf{96.2\%} & 0.686   & 80.9\%               & 0.686   & 98.1\%                   & 0.849   & 606*         \\
18000 & 1    & \textbf{98.7\%} & 11.473   & N.A.            & N.A.    & N.A.                 & N.A.    & N.A.                     & N.A.    & 160*         \\
      & 2    & \textbf{99.7\%} & 4.249    & 47.2\%          & 0.495   & 35.0\%               & 0.495   & 88.3\%                   & 1.372   & 285*         \\
      & 5    & 100\%           & 2.700    & 89.1\%          & 1.022   & 69.6\%               & 1.022   & \textbf{99.2\%}          & 1.619   & 659*         \\
      & 10   & 100\%           & 2.389    & \textbf{96.2\%} & 1.354   & 84.0\%               & 1.354   & 99.9\%                   & 1.676   & 1283*        \\
21000 & 1    & \textbf{100\%}  & 1272.497 & N.A.            & N.A.    & N.A.                 & N.A.    & N.A.                     & N.A.    & 143*         \\
      & 2    & 100\%           & 468.656  & 36.6\%          & 48.271  & 36.5\%               & 48.271  & \textbf{99.9\%}          & 133.799 & 242*         \\
      & 5    & \textbf{100\%}  & 296.354  & 72.5\%          & 99.900  & 72.2\%               & 99.900  & \textbf{100\%}           & 158.067 & 537*         \\
      & 10   & \textbf{100\%}  & 261.811  & 86.6\%          & 132.870 & 86.3\%               & 132.870 & \textbf{100\%}           & 163.747 & 1029*        \\
25000 & 1    & \textbf{99.9\%} & 201.611  & N.A.            & N.A.    & N.A.                 & N.A.    & N.A.                     & N.A.    & 142*         \\
      & 2    & 100\%           & 74.595   & 36.6\%          & 7.597   & 35.3\%               & 7.597   & \textbf{98.8\%}          & 21.057  & 215*         \\
      & 5    & \textbf{100\%}  & 47.205   & 72.5\%          & 15.737  & 70.1\%               & 15.737  & \textbf{100\% }          & 24.922  & 433*         \\
      & 10   & \textbf{100\%}  & 41.583   & 86.8\%          & 20.819  & 84.7\%               & 20.819  & \textbf{100\%}           & 25.750  & 797*         
\end{tblr}
}
\end{table}

\begin{table}[ht]
\centering
\caption{Coverage probabilities (Pro.), confidence interval widths (Wid.) and running time (unit: second) of the $\ell_2$-regularized logistic regression with $n=10^5$ and different $p$. The closest coverage probability to the nominal 95\% level among all methods in each setting is \textbf{bold}.}
\label{table_different_p_l2}
\medskip
\resizebox{\columnwidth}{!}{
\begin{tblr}{
  cells = {c},
  cell{1}{1} = {r=2}{},
  cell{1}{2} = {r=2}{},
  cell{1}{3} = {c=2}{},
  cell{1}{5} = {c=2}{},
  cell{1}{7} = {c=2}{},
  cell{1}{9} = {c=2}{},
  cell{1}{11} = {r=2}{},
  cell{3}{1} = {r=4}{},
  cell{7}{1} = {r=4}{},
  cell{11}{1} = {r=4}{},
  cell{15}{1} = {r=4}{},
  cell{19}{1} = {r=4}{},
  hline{1,3,7,11,15,19,23} = {-}{},
}
$p $  & $B $ & Cheap Bootstrap &       & Basic Bootstrap &       & Percentile Bootstrap &       & Standard Error Bootstrap &       & Running Time \\
      &      & Pro.            & Wid.  & Pro.            & Wid.  & Pro.                 & Wid.  & Pro.                     & Wid.  &              \\
12000 & 1    & \textbf{96.3\%} & 3.162 & N.A.            & N.A.  & N.A.                 & N.A.  & N.A.                     & N.A.  & 42*          \\
      & 2    & \textbf{97.3\%} & 1.182 & 41.0\%          & 0.161 & 32.2\%               & 0.161 & 74.7\%                   & 0.446 & 69*          \\
      & 5    & 98.4\%          & 0.755 & 80.3\%          & 0.331 & 64.4\%               & 0.331 & \textbf{92.2\%}          & 0.525 & 150*         \\
      & 10   & 98.9\%          & 0.669 & 93.5\%          & 0.438 & 78.9\%               & 0.438 & \textbf{95.6\%}          & 0.543 & 284*         \\
15000 & 1    & \textbf{96.8\%} & 3.943 & N.A.            & N.A.  & N.A.                 & N.A.  & N.A.                     & N.A.  & 83*          \\
      & 2    & \textbf{97.9\%} & 1.471 & 43.5\%          & 0.196 & 32.9\%               & 0.196 & 76.9\%                   & 0.543 & 141*         \\
      & 5    & 99.1\%          & 0.938 & 84.3\%          & 0.404 & 65.6\%               & 0.404 & \textbf{93.9\%}          & 0.640 & 315*         \\
      & 10   & 99.6\%          & 0.832 & \textbf{95.9\%} & 0.534 & 80.1\%               & 0.534 & 96.8\%                   & 0.662 & 606*         \\
18000 & 1    & \textbf{97.3\%} & 5.056 & N.A.            & N.A.  & N.A.                 & N.A.  & N.A.                     & N.A.  & 156*         \\
      & 2    & \textbf{98.5\%} & 1.884 & 45.7\%          & 0.246 & 33.6\%               & 0.246 & 79.3\%                   & 0.683 & 241*         \\
      & 5    & 99.6\%          & 1.202 & 87.6\%          & 0.508 & 67.1\%               & 0.508 & \textbf{95.5\%}          & 0.805 & 498*         \\
      & 10   & 99.9\%          & 1.065 & \textbf{97.1\%} & 0.673 & 81.6\%               & 0.673 & 97.8\%                   & 0.833 & 925*         \\
21000 & 1    & \textbf{97.6\%} & 6.400 & N.A.            & N.A.  & N.A.                 & N.A.  & N.A.                     & N.A.  & 191*         \\
      & 2    & \textbf{98.8\%} & 2.385 & 47.1\%          & 0.310 & 34.5\%               & 0.310 & 81.4\%                   & 0.859 & 319*         \\
      & 5    & 99.8\%          & 1.520 & 89.4\%          & 0.639 & 68.6\%               & 0.639 & \textbf{96.5\%}          & 1.012 & 703*         \\
      & 10   & 100\%           & 1.348 & \textbf{97.5\%} & 0.845 & 83.1\%               & 0.845 & 98.3\%                   & 1.047 & 1343*        \\
25000 & 1    & \textbf{97.4\%} & 7.259 & N.A.            & N.A.  & N.A.                 & N.A.  & N.A.                     & N.A.  & 208*         \\
      & 2    & \textbf{98.7\%} & 2.712 & 46.5\%          & 0.362 & 35.2\%               & 0.362 & 80.8\%                   & 1.004 & 350*         \\
      & 5    & 99.7\%          & 1.731 & 88.8\%          & 0.746 & 69.9\%               & 0.746 & \textbf{96.0\%}          & 1.183 & 777*         \\
      & 10   & 99.9\%          & 1.536 & 98.3\%          & 0.988 & 84.5\%               & 0.988 & \textbf{97.9\%}          & 1.224 & 1487*        
\end{tblr}
}
\end{table}

\subsection{Logistic Regression with Independent Covariates}\label{sec:logistic regression}

Figure \ref{figure_logistic_vs_B} presents the coverage probabilities and confidence interval widths of 95\%-level confidence intervals for three typical choices of parameters: $\beta_1=1$, $\beta_{301}=-1$ and $\beta_{601}=0$. We observe that all cheap bootstrap coverage probabilities are close to or larger than the nominal level 95\% while other bootstrap method coverages are below 90\% except for the standard error bootstrap for $\beta_{601}=0$ and $B\ge5$. Besides, cheap bootstrap interval widths are larger than others but decay very fast for the first few $B$'s, in line with our observation in the previous linear regression example. In fact, it is already quite close to other bootstrap widths for $\beta_{601}=0$ and $B=10$. Figure \ref{figure_logistic} reports the box plot of the coverage probabilities and confidence interval widths of all $\beta_i$'s with $B=1,2,5,10$. We distinguish between $\beta_i\neq0$ and $\beta_i=0$ since the former has wider widths than the latter. For $\beta_{i}\neq0$, the cheap bootstrap widths shrink more slowly so that almost all cheap bootstrap coverage probabilities are 100\% but other bootstrap method coverages are still below 90\% in almost all cases. For $\beta_{i}=0$, the cheap bootstrap with any $B$, standard error bootstrap with $B=5,10$ and basic bootstrap with $B=10$ have coverage probabilities close to the nominal level 95\%. In other cases, most of the coverage probabilities are below 85\%. A similar decay rate for the cheap bootstrap interval width is also observed here: it decreases by around $2/3$ from $B=1$ to $B=2$ and by around $4/5$ from $B=1$ to $B=10$.

\begin{figure}[htbp]
\centering

\subfloat[Coverage probability for $\beta_1=1$]{\includegraphics[width=0.45\textwidth]{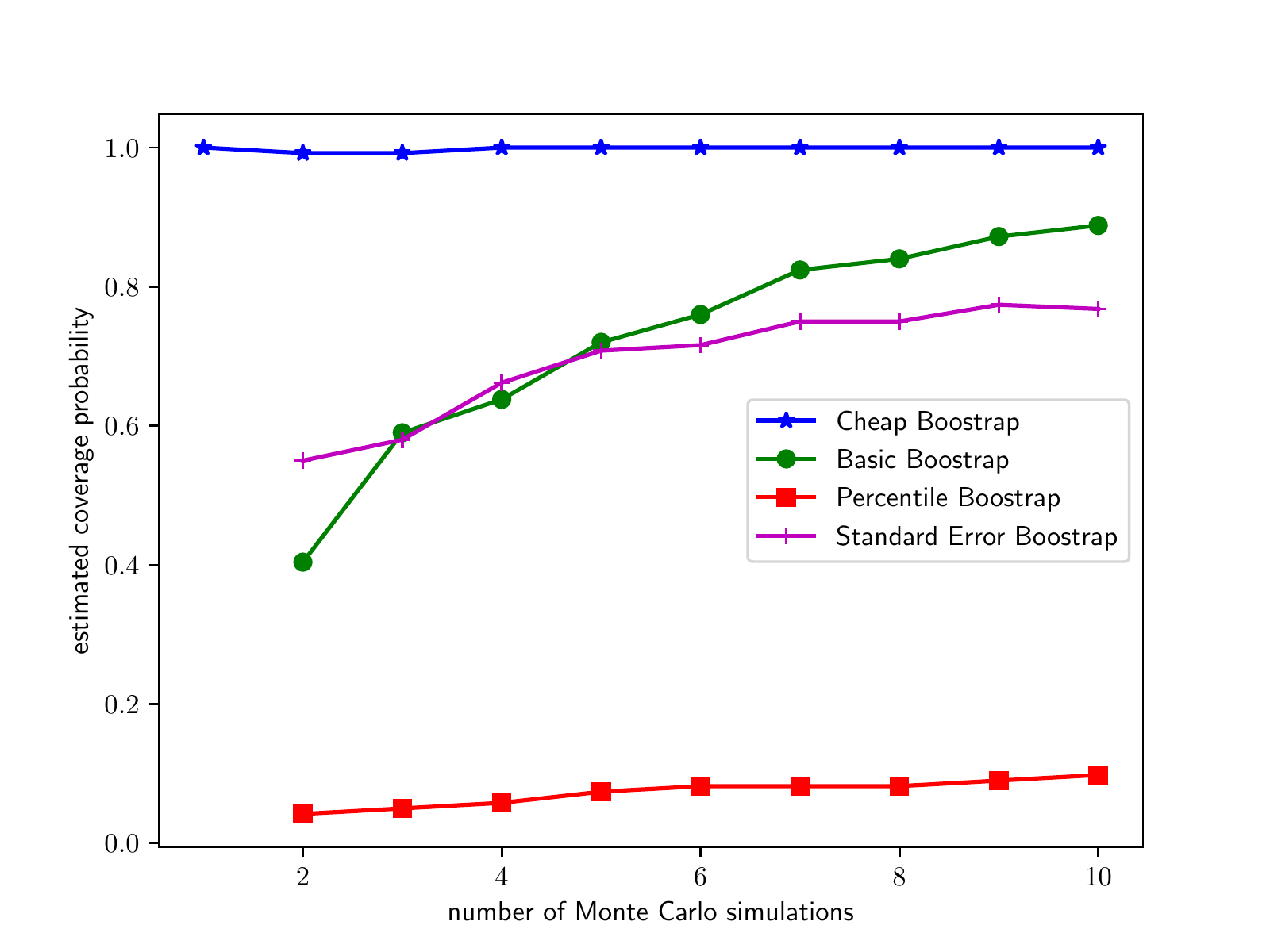}}
\subfloat[Confidence interval width for $\beta_1=1$]{\includegraphics[width=0.45\textwidth]{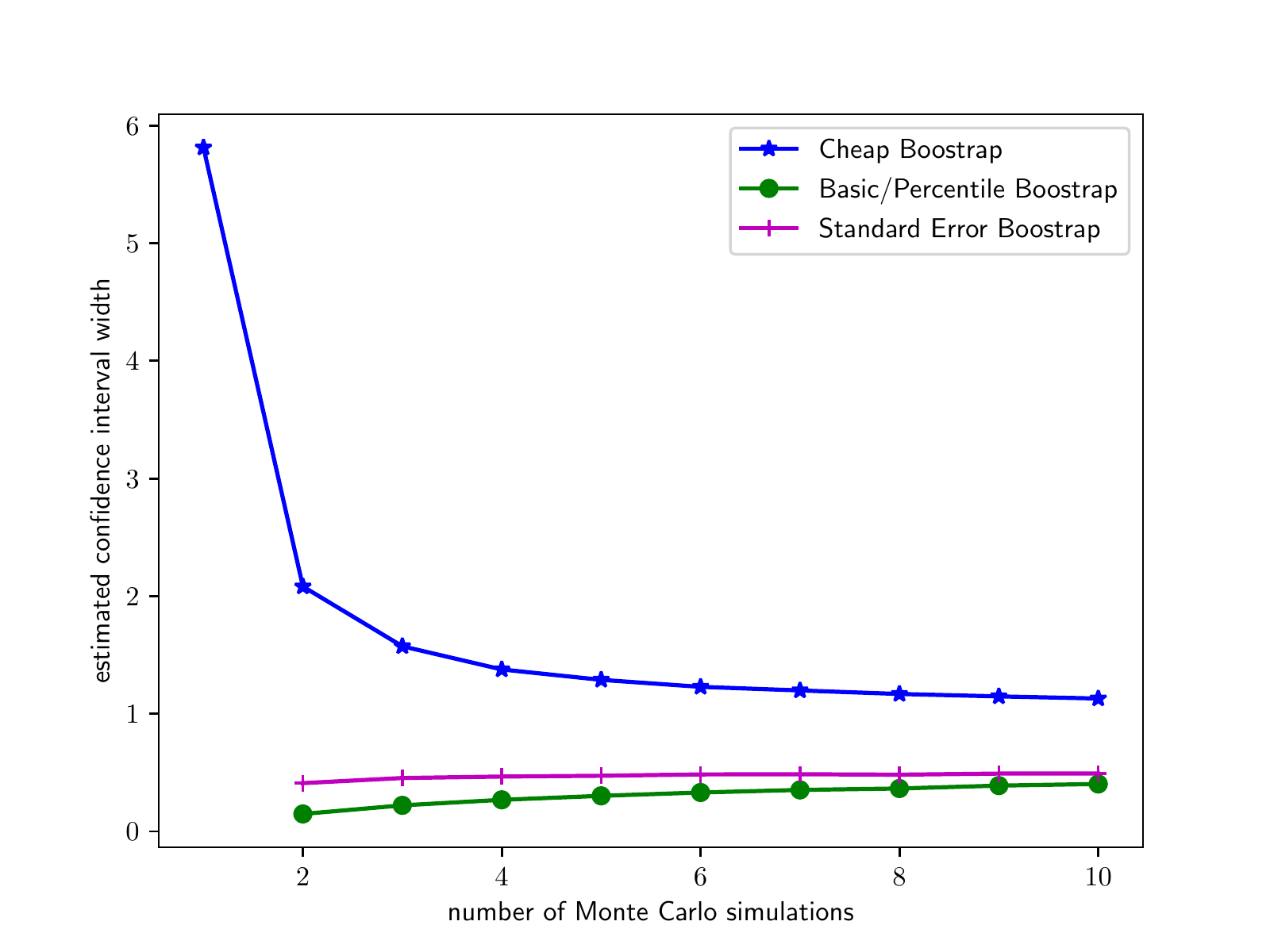}}\\
\subfloat[Coverage probability for $\beta_{301}=-1$]{\includegraphics[width=0.45\textwidth]{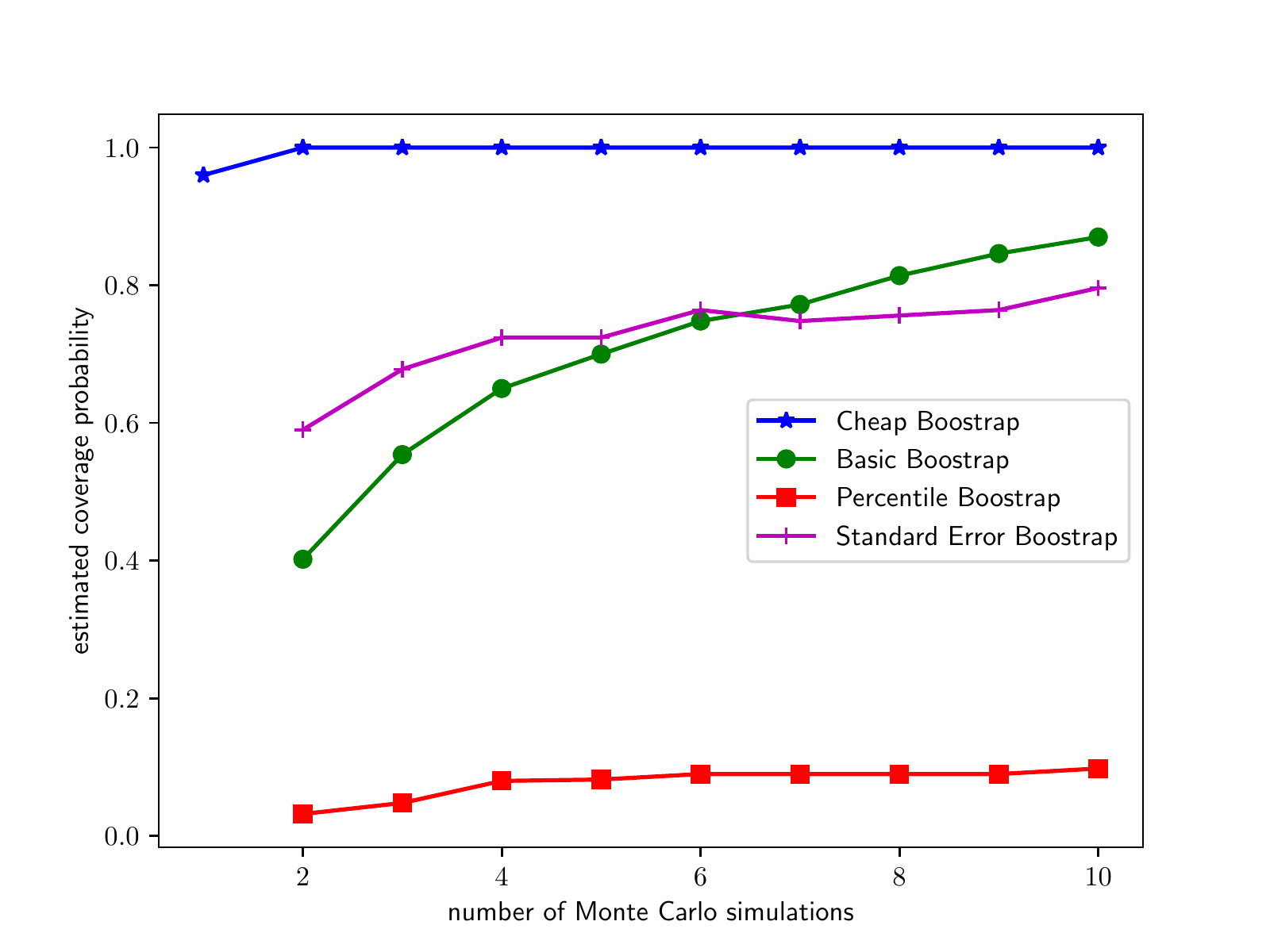}}
\subfloat[Confidence interval width for $\beta_{301}=-1$]{\includegraphics[width=0.45\textwidth]{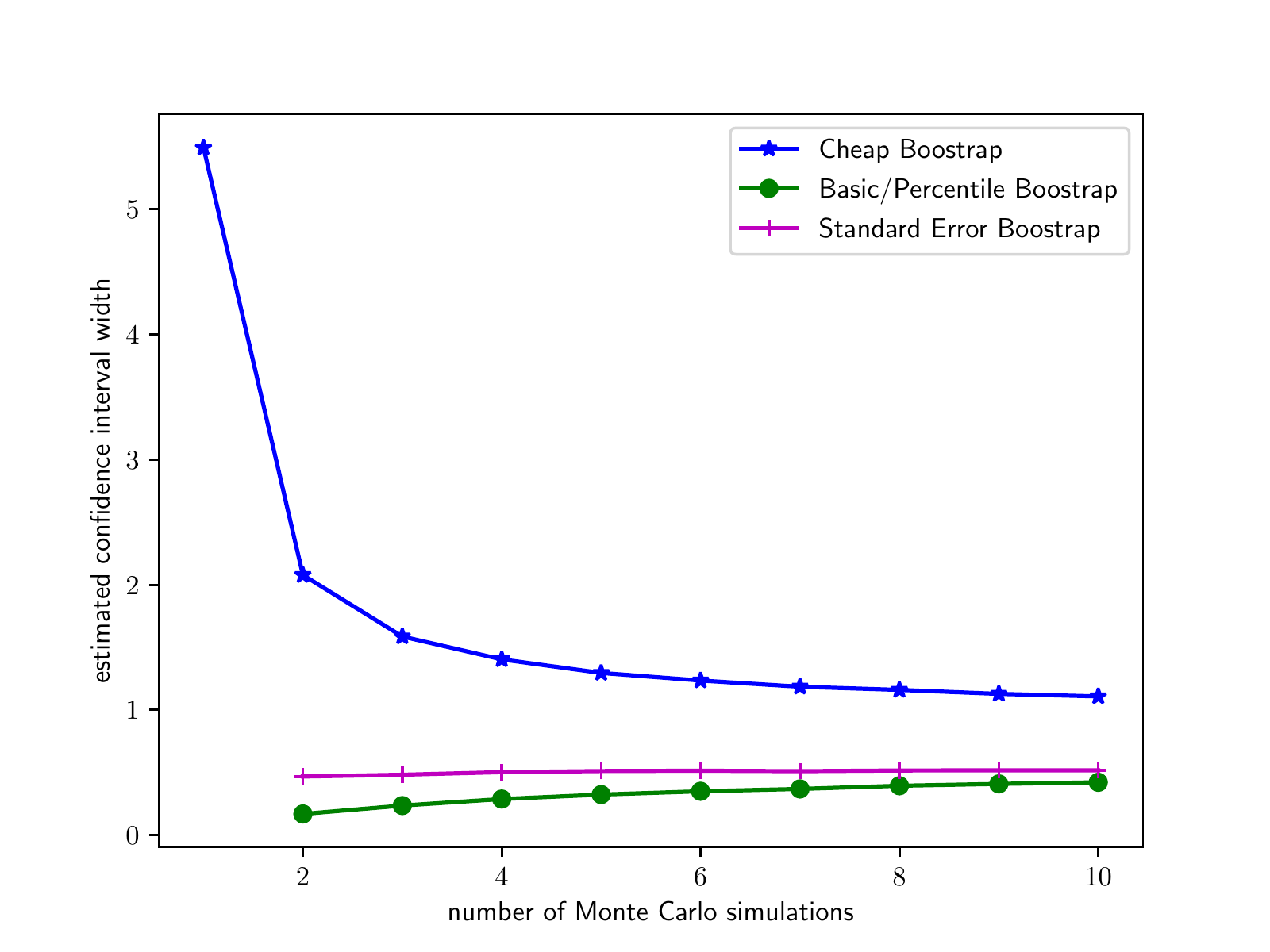}}\\
\subfloat[Coverage probability for $\beta_{601}=0$]{\includegraphics[width=0.45\textwidth]{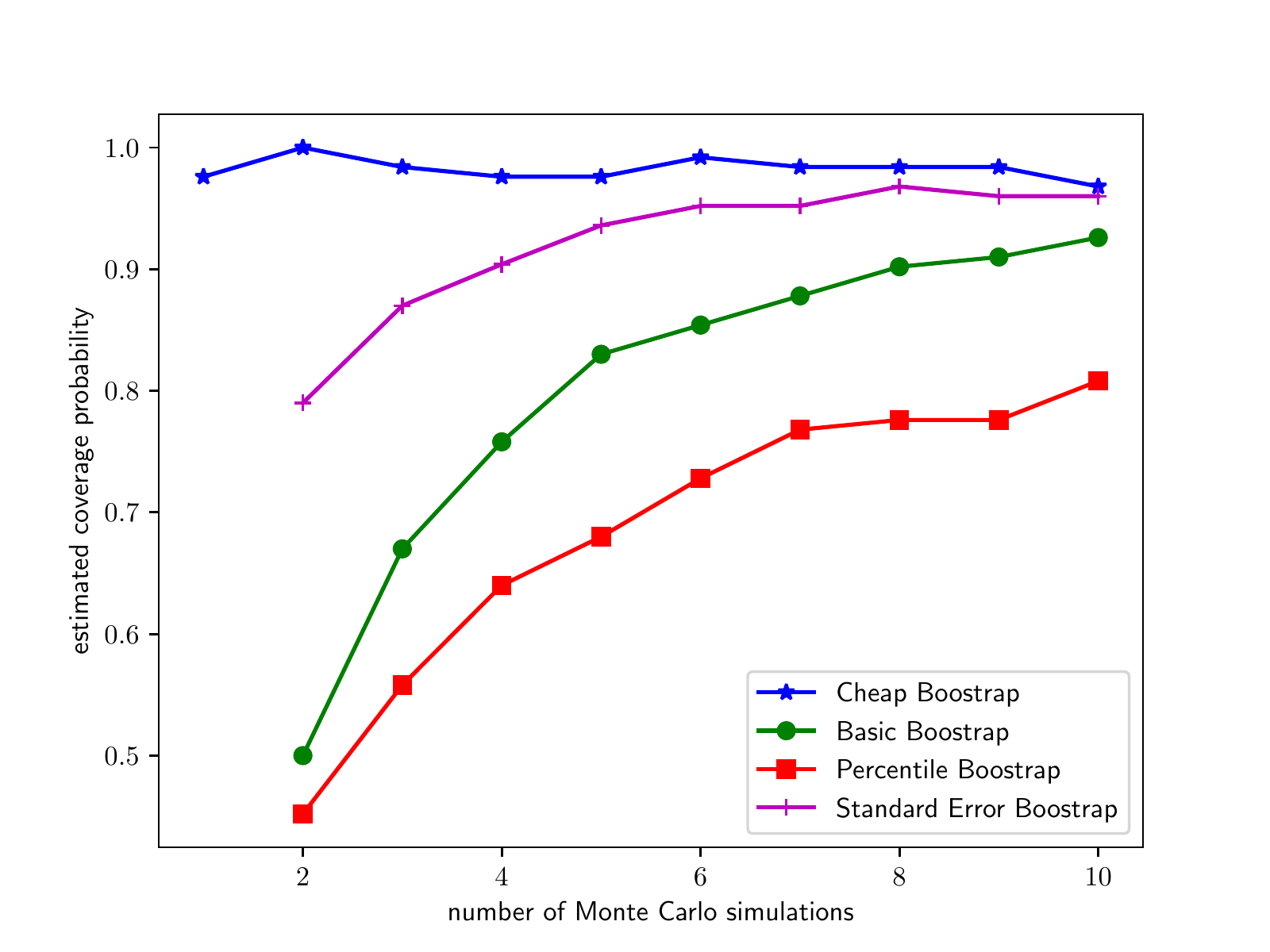}}
\subfloat[Confidence interval width for $\beta_{601}=0$]{\includegraphics[width=0.45\textwidth]{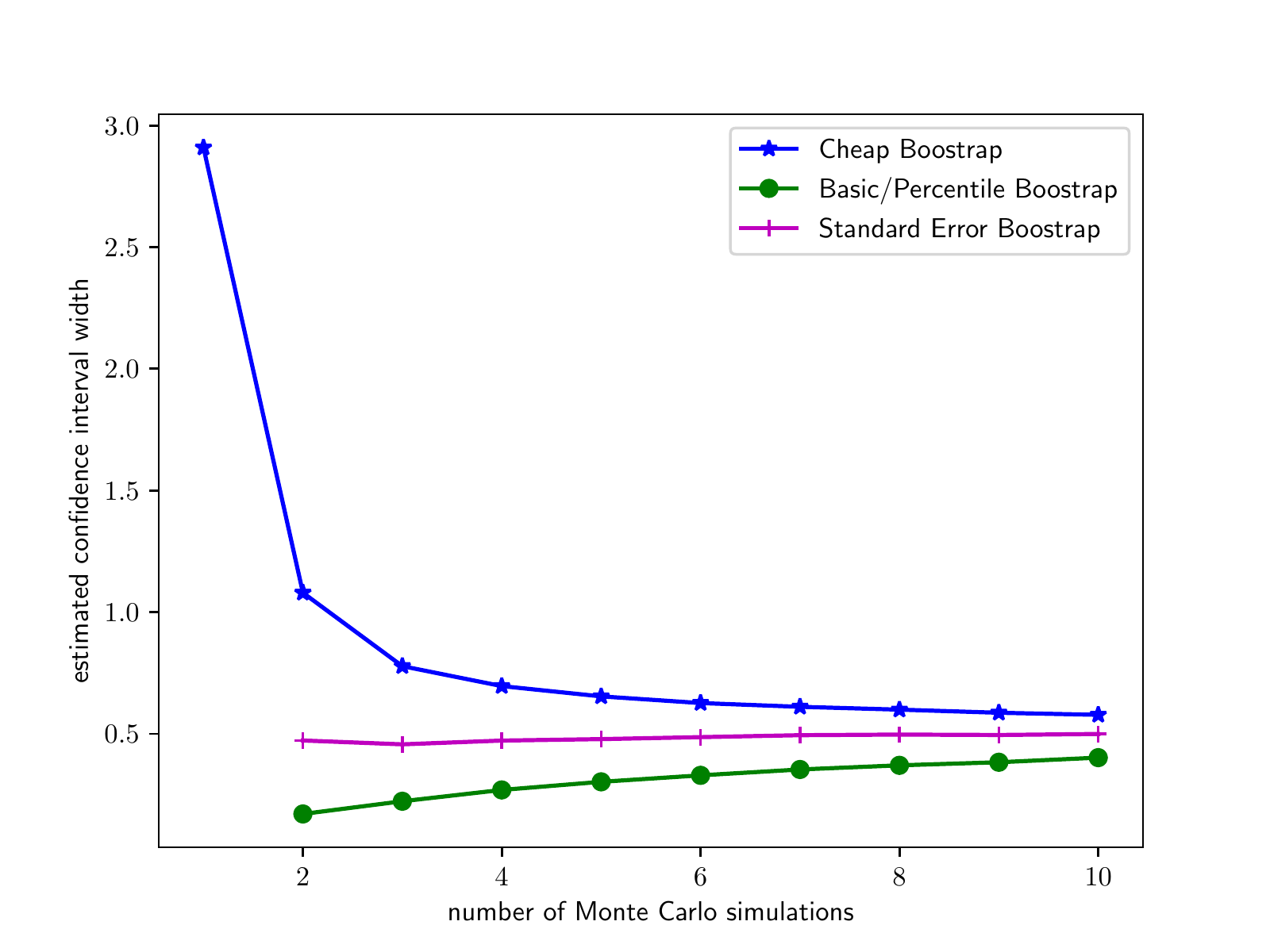}}\\

\caption{Empirical coverage probabilities and confidence interval widths for different numbers of resamples in a logistic regression.}%
\label{figure_logistic_vs_B}

\end{figure}

\begin{figure}[ht]
\centering

\subfloat[Coverage probability $\beta_i\neq0$]{\includegraphics[width=0.45\textwidth]{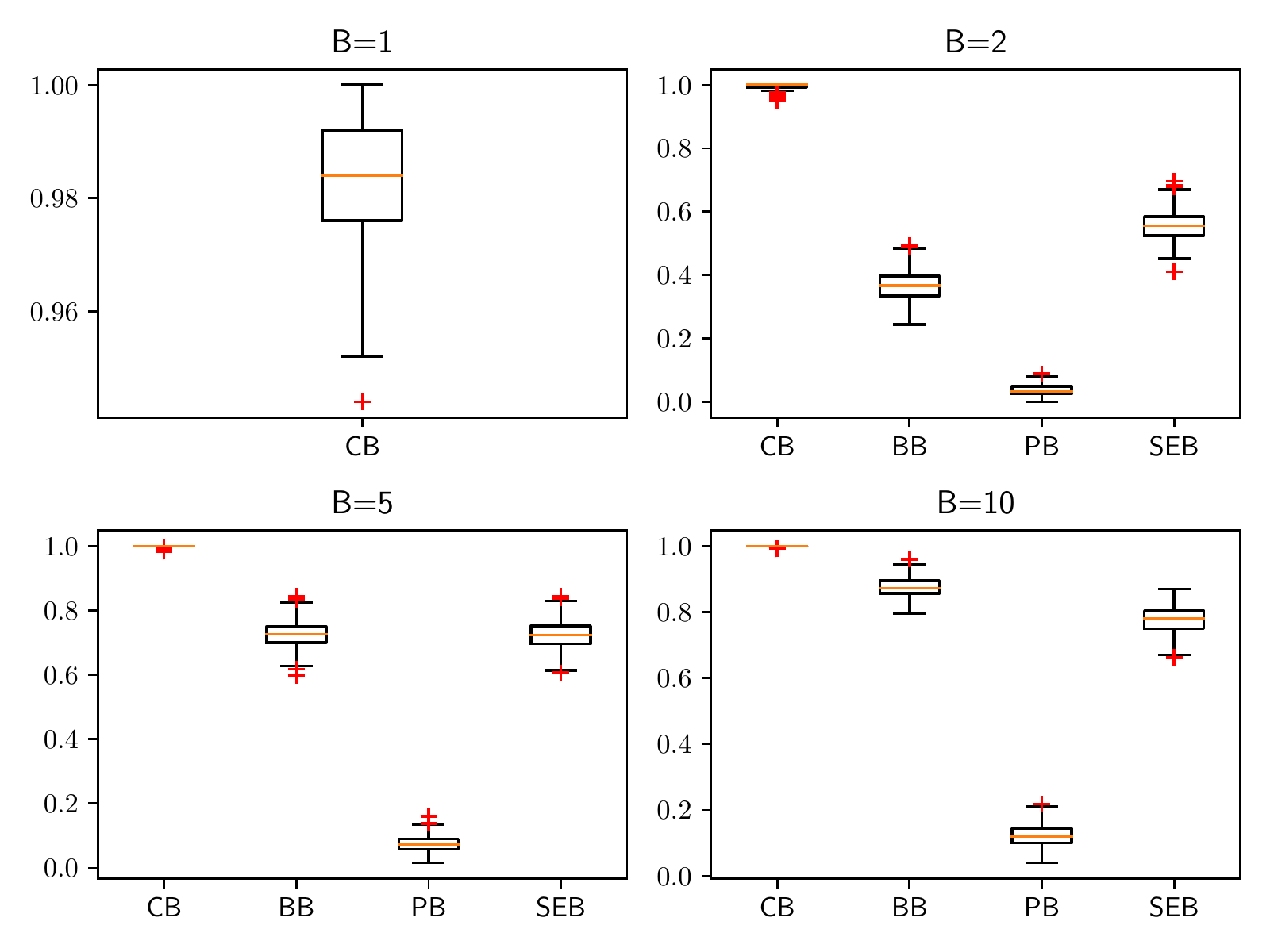}}
\subfloat[Confidence interval width $\beta_i\neq0$]{\includegraphics[width=0.45\textwidth]{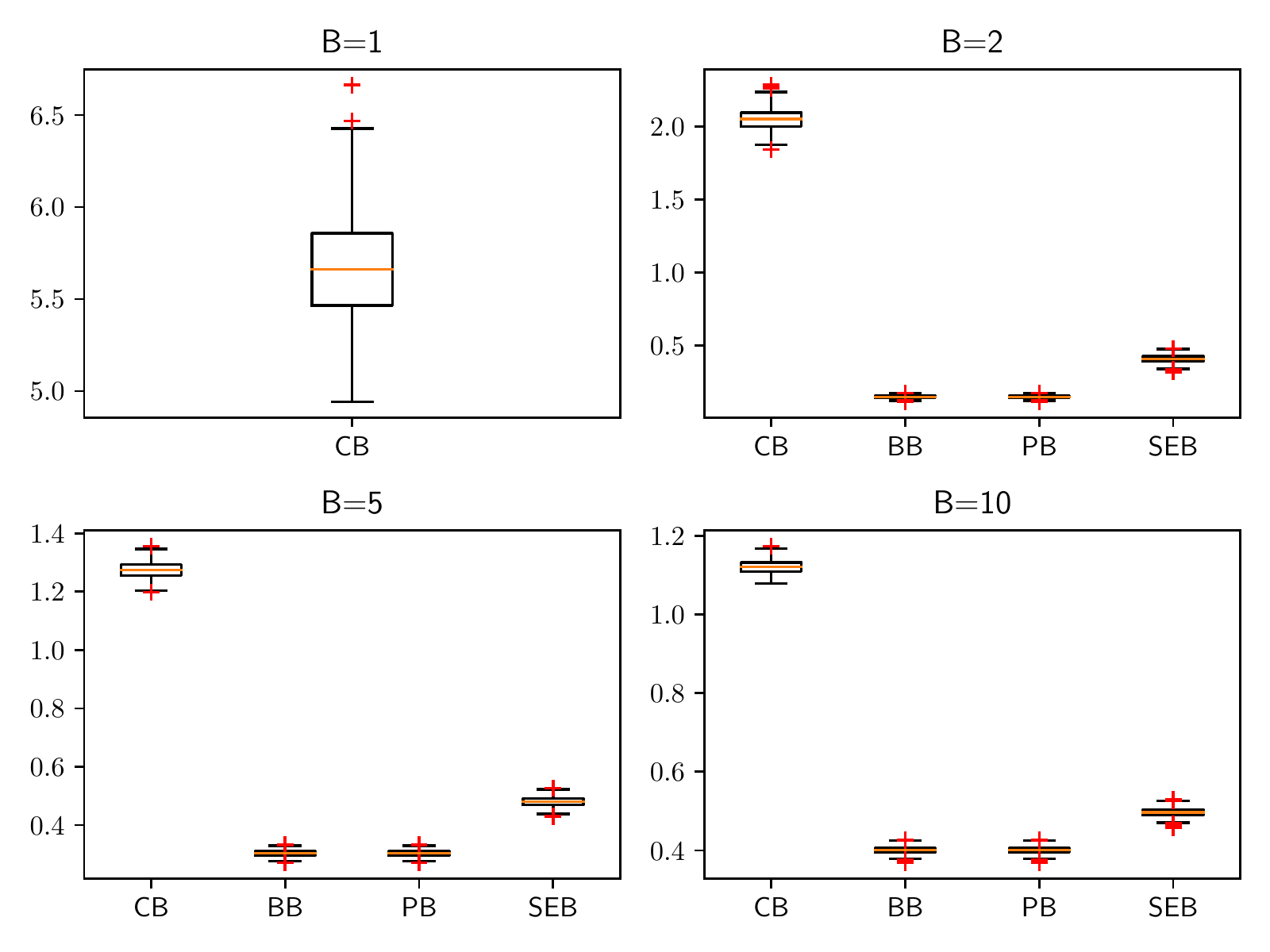}}\\
\subfloat[Coverage probability $\beta_i=0$]{\includegraphics[width=0.45\textwidth]{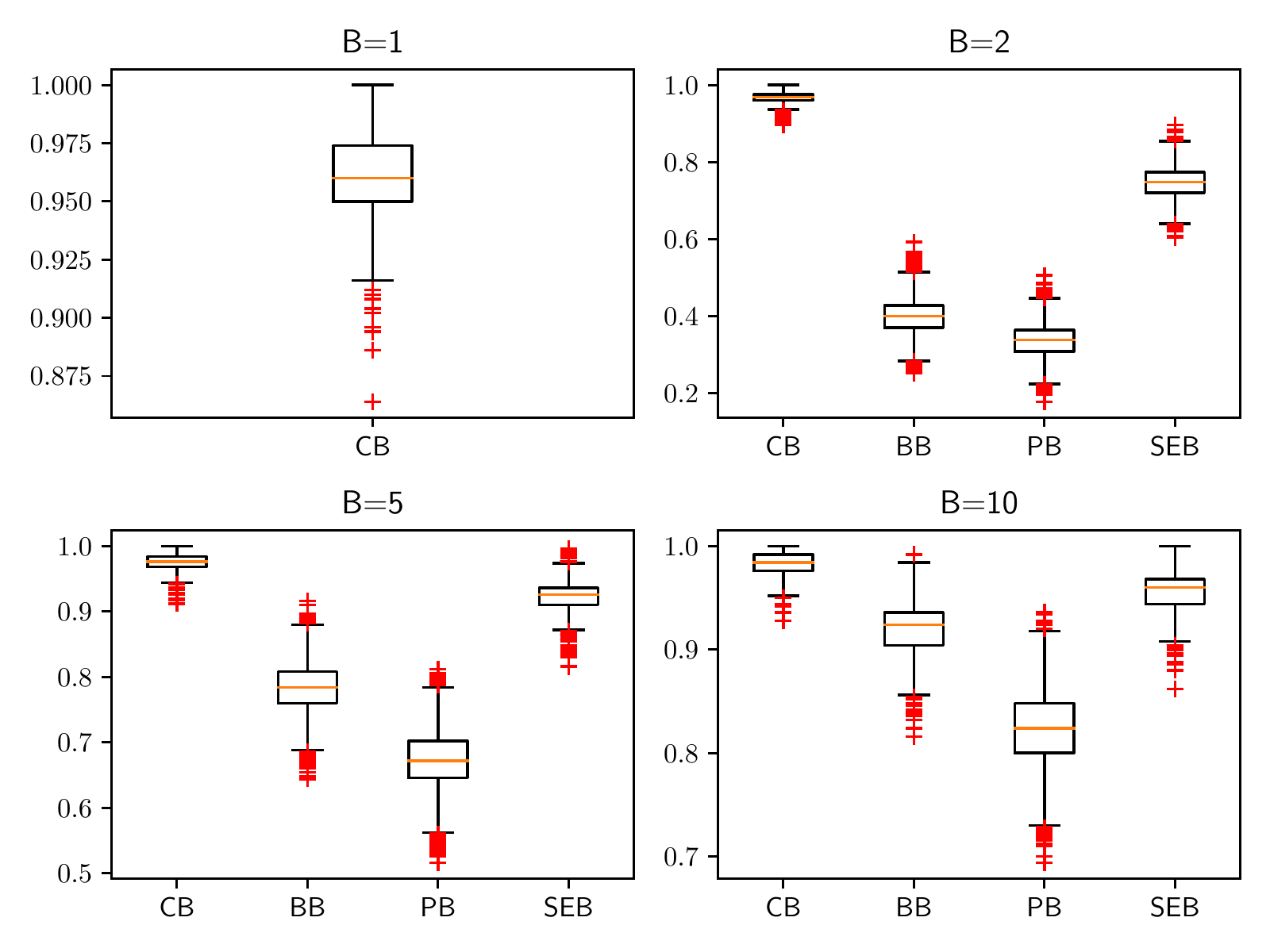}}
\subfloat[Confidence interval width $\beta_i=0$]{\includegraphics[width=0.45\textwidth]{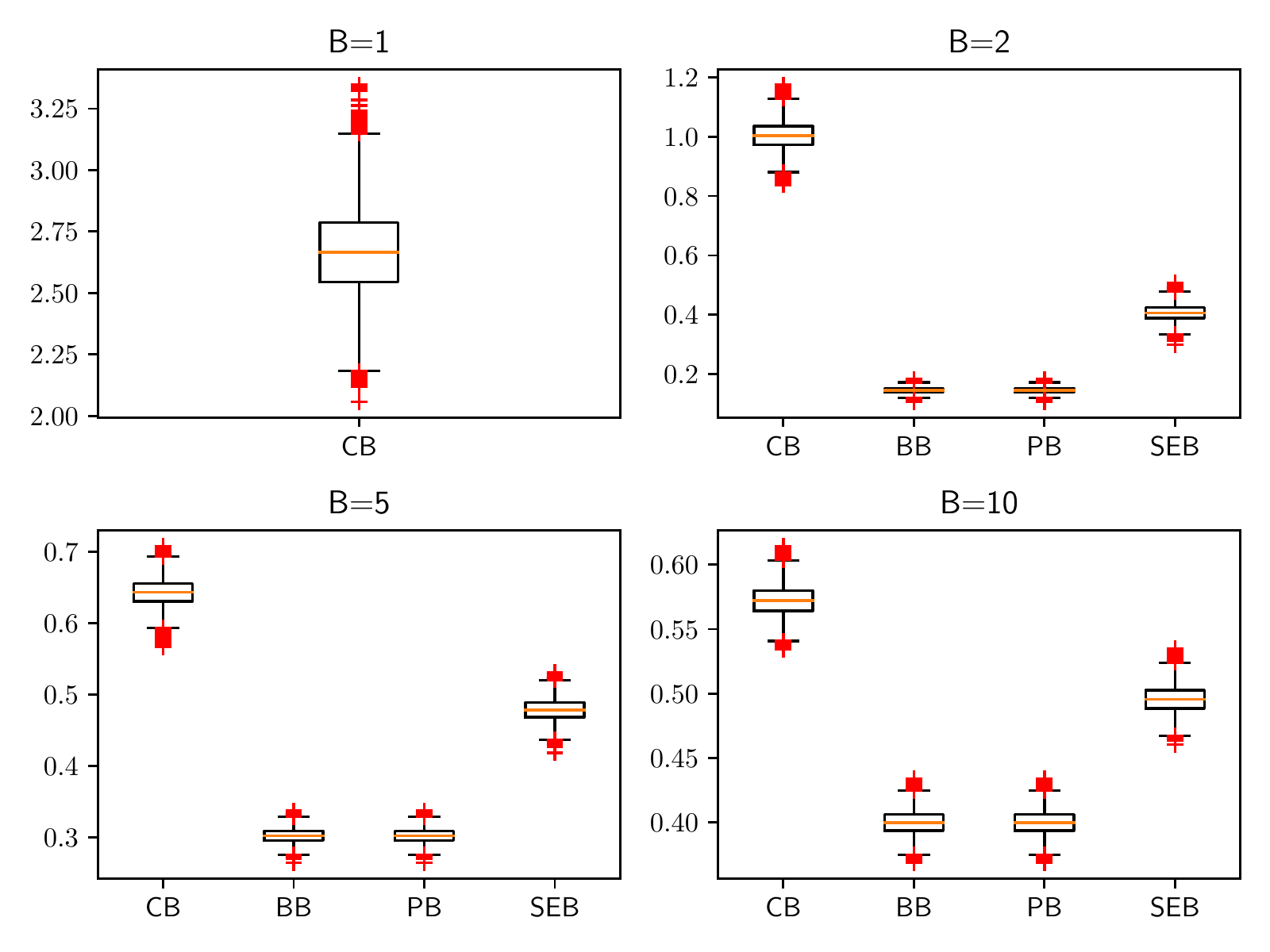}}

\caption{Box plots of empirical coverage probabilities and confidence interval widths of all $\beta_i$'s for different numbers of resamples in a logistic regression.}%
\label{figure_logistic}

\end{figure}

\subsection{Computer Network}\label{sec:computer model}
We detail the specifications of the computer communication network simulation model; similar models have been used in \citet{cheng1997sensitivity,lin2015single,lam2022subsampling}. This network can be represented by an undirected graph in Figure \ref{network}. The four nodes denote message processing units and the four edges are transport channels. For every pair of nodes $i,j$ ($i\neq j$), there are external messages which enter into node $i$ from the external and are to be transmitted to node $j$ through a prescribed path. Their arrival time follows a Poisson process with parameter $\lambda_{i,j}$ showed in Table \ref{network_parameter1}. All the message lengths (unit: bits) are i.i.d. following a common exponential distribution with mean $300$ bits. Suppose each unit spends $0.001$ second to process a message passing it. We assume the node storage is unlimited but the channel storage is restricted to $275000$ bits. Message speed in transport channels is $150000$ miles per second and channel $i$ has length $100i$ miles. Therefore, it takes $l/275000+100i/150000$ seconds for a message with length $l$ bits to pass channel $i$. Suppose the network is empty at the beginning. The performance measure of interest is the steady-state average delay for the messages where delay means the time from the entering node to the destination node. It has approximate true value $7.05\times10^{-3}$. This example has $13$ unknown input distributions, i.e., $12$ inter-arrival time distributions $\mathrm{Exp}(\lambda_{i,j})$ and one message length distribution $\mathrm{Exp}(1/300)$, for which we have data sizes from $3\times10^{4}$ to $6\times10^{4}$. Given input distributions $P_1,\ldots,P_{13}$, the performance measure of this system can be computed accurately by
\[
\psi(P_{1},\ldots,P_{13})=E_{P_{1},\ldots,P_{13}}\left[  \frac{1}{9500}\sum_{k=501}^{10000}D_{k}\right]  ,
\]
where $D_{k}$ is the delay for the $k$-th message. The point estimator of $\psi(P_{1},\ldots,P_{13})$ is taken as $\hat{\psi}=\psi(\hat{P}_{1,n_{1}},\ldots,\hat{P}_{13,n_{13}})$ where each $\hat{P}_{i,n_{i}}$ is the empirical distribution of $n_{i}$ i.i.d. samples $\{X_{i,j},j=1,\ldots,n_{i}\}$ from the $i$-th input distribution $P_{i}$. 

Next we construct the bootstrap estimator $\hat{\psi}^{\ast b}$. For each $b=1,\ldots,B$ and $i=1,\ldots,13$, we sample with replacement the data $\{X_{i,j},j=1,\ldots,n_{i}\}$ to obtain the bootstrap resamples $\{X_{i,j}^{\ast b},j=1,\ldots,n_{i}\}$ and denote the resample empirical distribution by $\hat{P}_{i,n_{i}}^{\ast b}$. The sampling procedure is conducted independently for different $b$ and $i$. The bootstrap estimator is taken as $\hat{\psi}^{\ast b}=\psi(\hat{P}_{1,n_{1}}^{\ast b},\ldots,\hat{P}_{13,n_{13}}^{\ast b})$. The cheap bootstrap confidence interval is still constructed as in (\ref{cheap CI}).

\begin{figure}[htbp]
\centering
\includegraphics[width=0.45\textwidth]{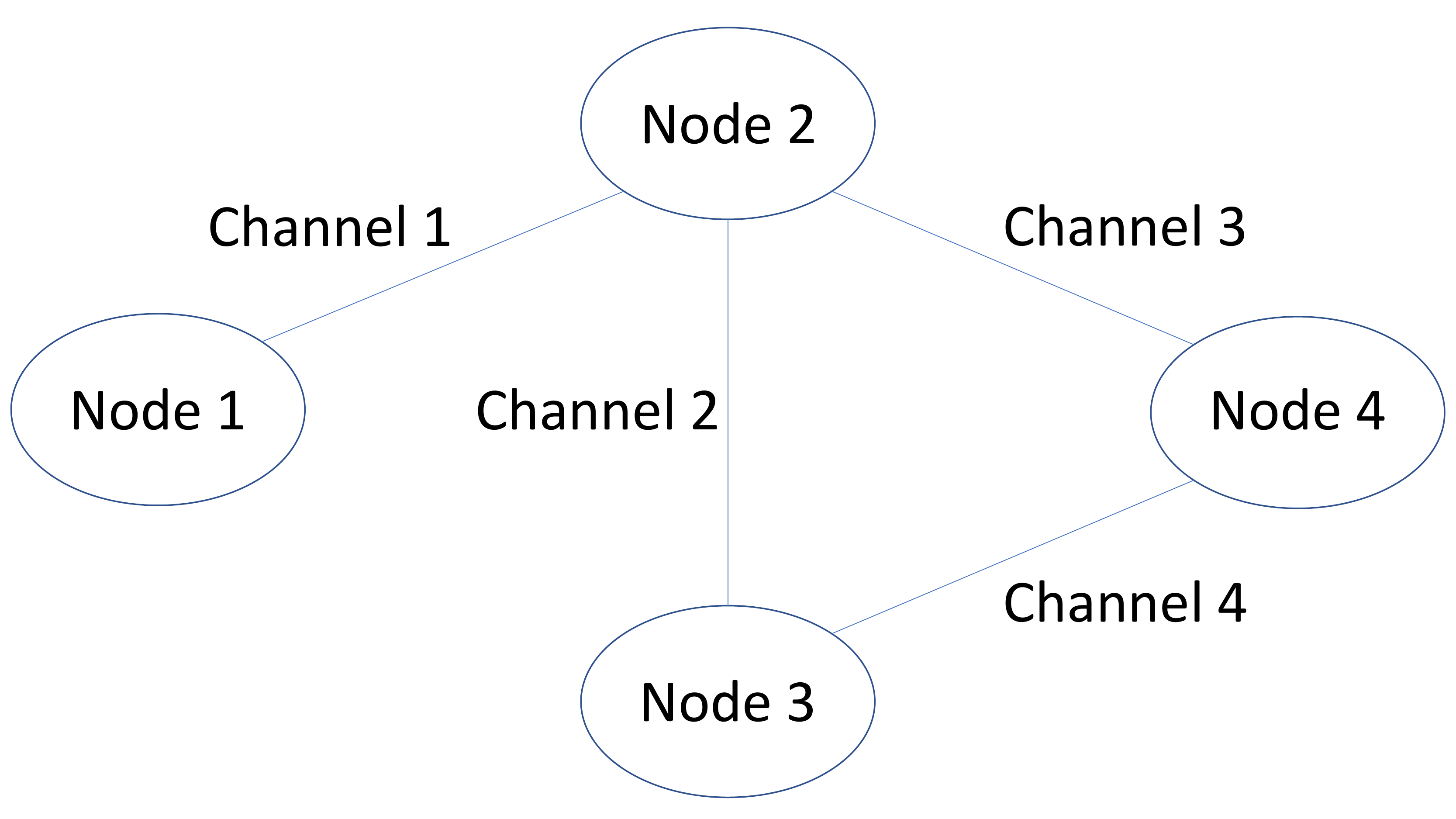}

\caption{A computer network with four nodes and four channels.}
\label{network}
\end{figure}

\begin{table}[ht]
\caption{Arrival rates $\lambda_{i,j}$ of messages to be transmitted from node $i$ to node $j$.}
\medskip
\label{network_parameter1}
\centering
\begin{tabular}{@{}ccccc@{}}
\toprule
         & \multicolumn{4}{c}{Node $j$} \\ \cmidrule(l){2-5}
Node $i$ & 1     & 2     & 3     & 4    \\ 
\midrule
1        & N.A.  & 40    & 30    & 35   \\
2        & 50    & N.A.  & 45    & 15   \\
3        & 60    & 15    & N.A.  & 20   \\
4        & 25    & 30    & 40    & N.A. \\ 
\bottomrule
\end{tabular}
\end{table}





The results for the above configuration can be found in the row ``Stochastic simulation'' of Table \ref{num_table} and the corresponding discussions can be found in Section \ref{sec:num}.

To investigate the robustness of the cheap bootstrap or other methods, we consider an alternative configuration where computer network is the same but the input models are different. More concretely, all 13 input models (12 inter-arrival time distributions and one message length distribution) are changed to Gamma distributions $\mathrm{Gamma}(\alpha,\beta)$ which have densities of the form
\[
f(x)=\frac{\beta^{\alpha}}{\Gamma(\alpha)}x^{\alpha-1}e^{-\beta x},x>0.
\]
The message length distribution follows $\mathrm{Gamma}(2.5,1/200)$ and the parameters for the inter-arrival time distributions $\mathrm{Gamma}(\alpha_{i,j},\beta_{i,j})$ are given in Table \ref{network_parameter2}. Under the new input distributions, the true steady-state mean delay is approximately $0.0109$. Figure \ref{figure_network2} reports the results. The cheap bootstrap coverage probabilities are close to the nominal level 95\% for any $B$ while those of the basic bootstrap and standard error bootstrap are below 60\% and 90\% respectively even for $B=10$. Percentile bootstrap also performs well when $B\ge7$ perhaps because of the skewness of the estimates. But this is not always the case in view of the previous numerical results.

\begin{table}[ht]
\caption{Parameters $(\alpha_{i,j},\beta_{i,j})$ for the inter-arrival time distribution of messages to be transmitted from node $i$ to node $j$.}
\medskip
\label{network_parameter2}
\centering
\begin{tabular}{@{}ccccc@{}}
\toprule
         & \multicolumn{4}{c}{Node $j$}                  \\ \cmidrule(l){2-5} 
Node $i$ & 1         & 2         & 3         & 4         \\ 
\midrule
1        & N.A.      & $(1.5, 60)$ & $(0.7, 40)$ & $(1.3, 50)$ \\
2        & $(2, 80)$   & N.A.      & $(1.5, 65)$ & $(0.6, 20)$ \\
3        & $(3, 100)$  & $(0.5, 25)$ & N.A.      & $(1.2, 30)$ \\
4        & $(0.8, 40)$ & $(1.1, 50)$ & $(0.9, 35)$ & N.A.      \\ 
\bottomrule
\end{tabular}
\end{table}

\begin{figure}[htbp]
\centering

\subfloat[Coverage probability]{\includegraphics[width=0.45\textwidth]{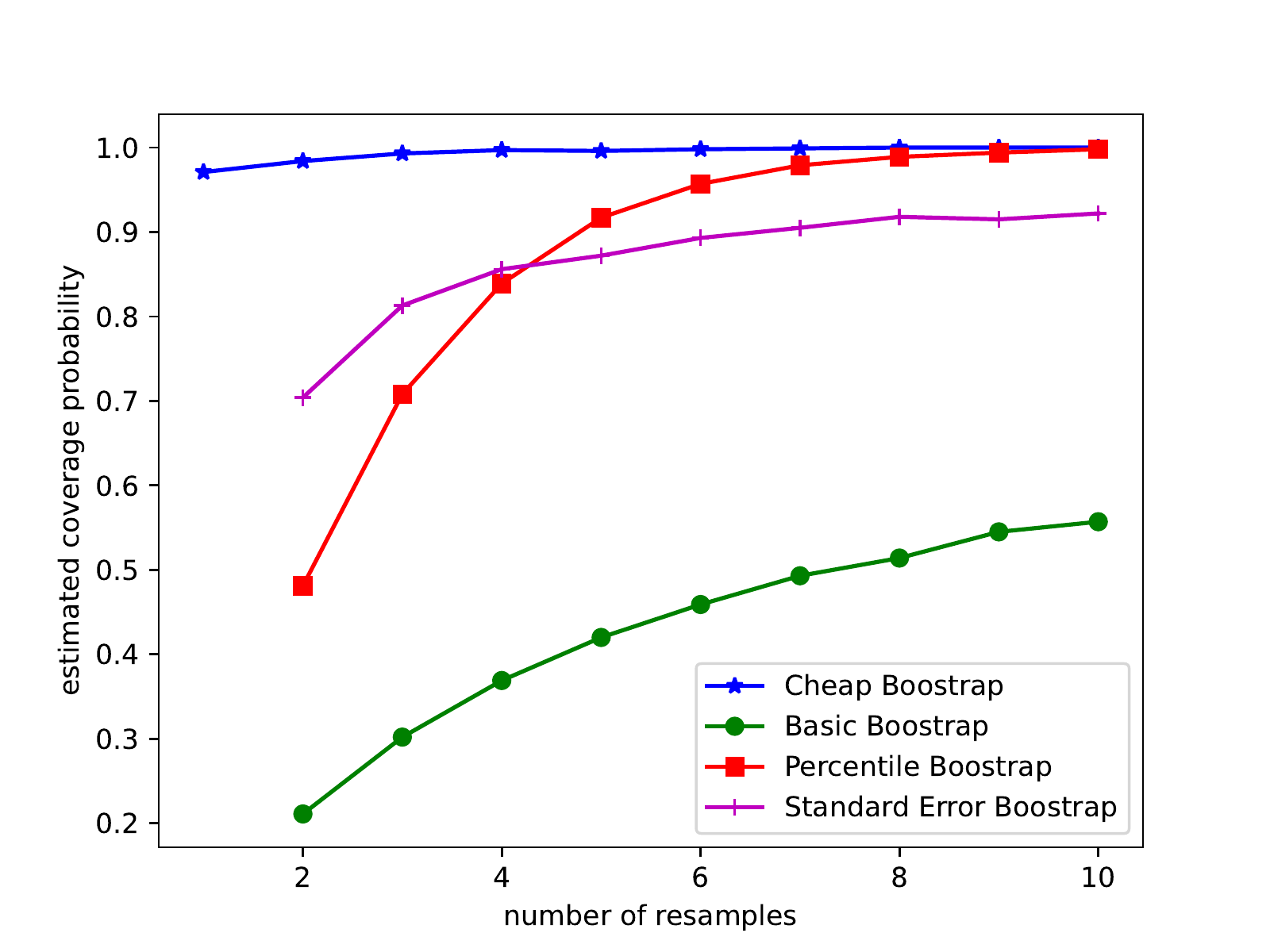}}
\subfloat[Confidence interval width]{\includegraphics[width=0.45\textwidth]{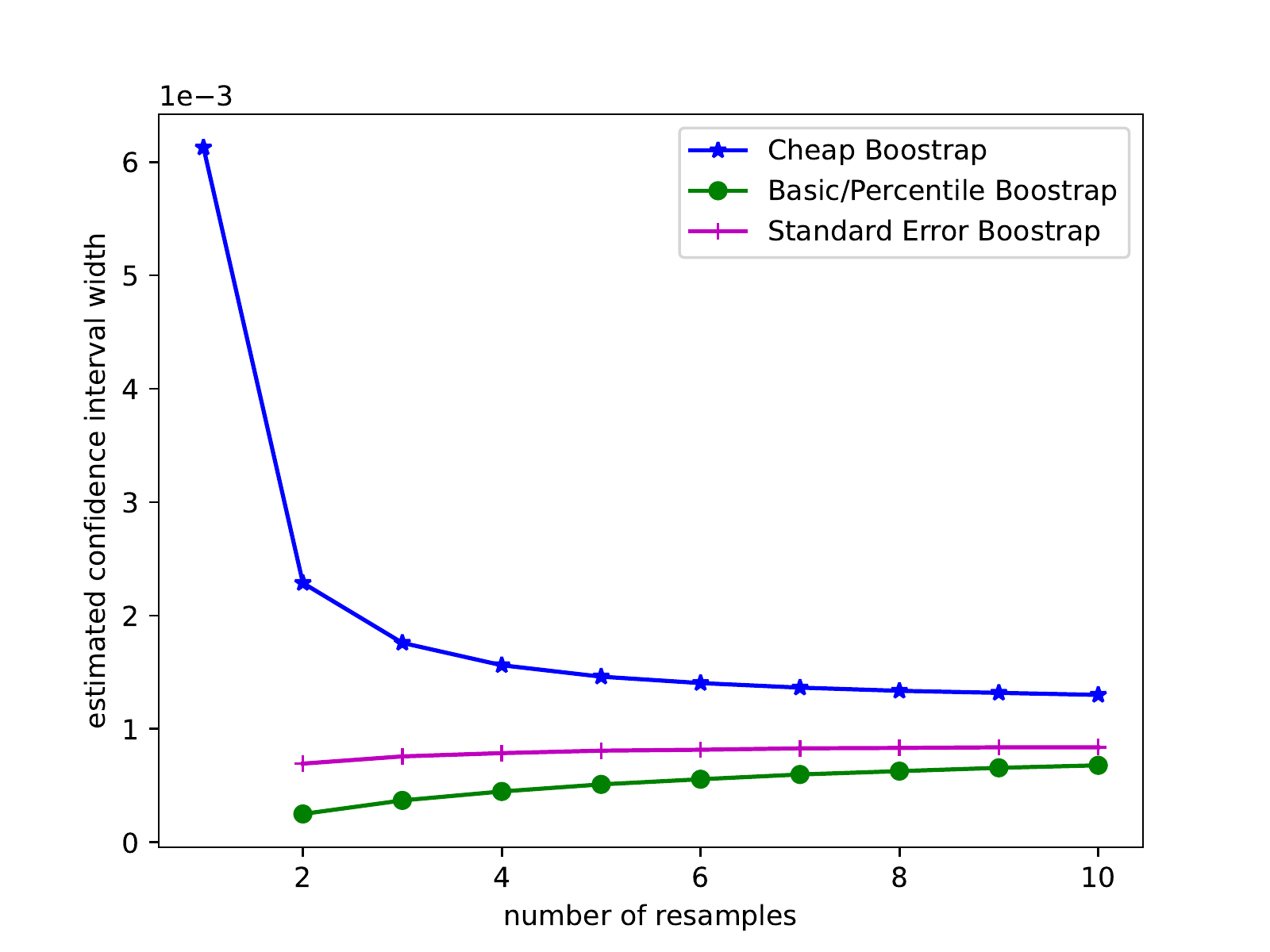}}

\caption{Empirical coverage probabilities and confidence interval widths for different numbers of resamples in a computer communication network.}%
\label{figure_network2}

\end{figure}

\subsection{Real World Problem}\label{sec:real data}
The data we use is the RCV1-v2 data in \citet{lewis2004rcv1}. This dataset contains $n=804414$ manually categorized newswire stories with a total of $p=47236$ features. We compare the confidence interval widths of the logistic regression parameters for the four bootstrap methods, by using all the observations to run the logistic regression and estimate the parameters. That is, the observation matrix is of the size $804414\times47236\approx4\times10^{10}$. There are up to $103$ different categories for all these newswire stories. As in \citet{singh2009parallel} and \citet{balakrishnan2008algorithms}, we only use the ``economics'' (``ECAT'') as the $+1$ label, i.e., the label $Y$ is $1$ if the newswire story is in ``economics'' and $0$ if not, which leads to $119920$ positive labels. Besides, we add $l_2$ regularization to this logistic regression as in \citet{singh2009parallel} and \citet{balakrishnan2008algorithms}. 

To run this logistic regression, we use sklearn.linear\_model.LogisticRegression (a machine learning package in Python) in the virtual machine c2-standard-8 in Google Cloud Platform, which takes about 30-40 minutes to run one bootstrap resample. Therefore, the common bootstrap methods which require $B=50$ or $100$ would be computationlly expensive.

In Section \ref{sec:num}, we report and discuss the average interval widths over all $\beta_i$'s for the four bootstrap methods. Here we display the results for three individual parameters, namely the first three $\beta_i$'s, in Figure \ref{figure_real_l2}. Since we are only able to run one experimental repetition in this real world example, the confidence interval widths contain some noises and thus we cannot observe the monotonicity of the widths when $B$ increases. But we still see the cheap bootstrap confidence interval widths are wider than others, with general trends that resemble the average interval widths in our synthetic examples. This suggests that the cheap bootstrap confidence intervals would have higher and closer-to-nominal coverages than the other methods.


\begin{figure}[htbp]
\centering
\includegraphics[width=0.45\textwidth]{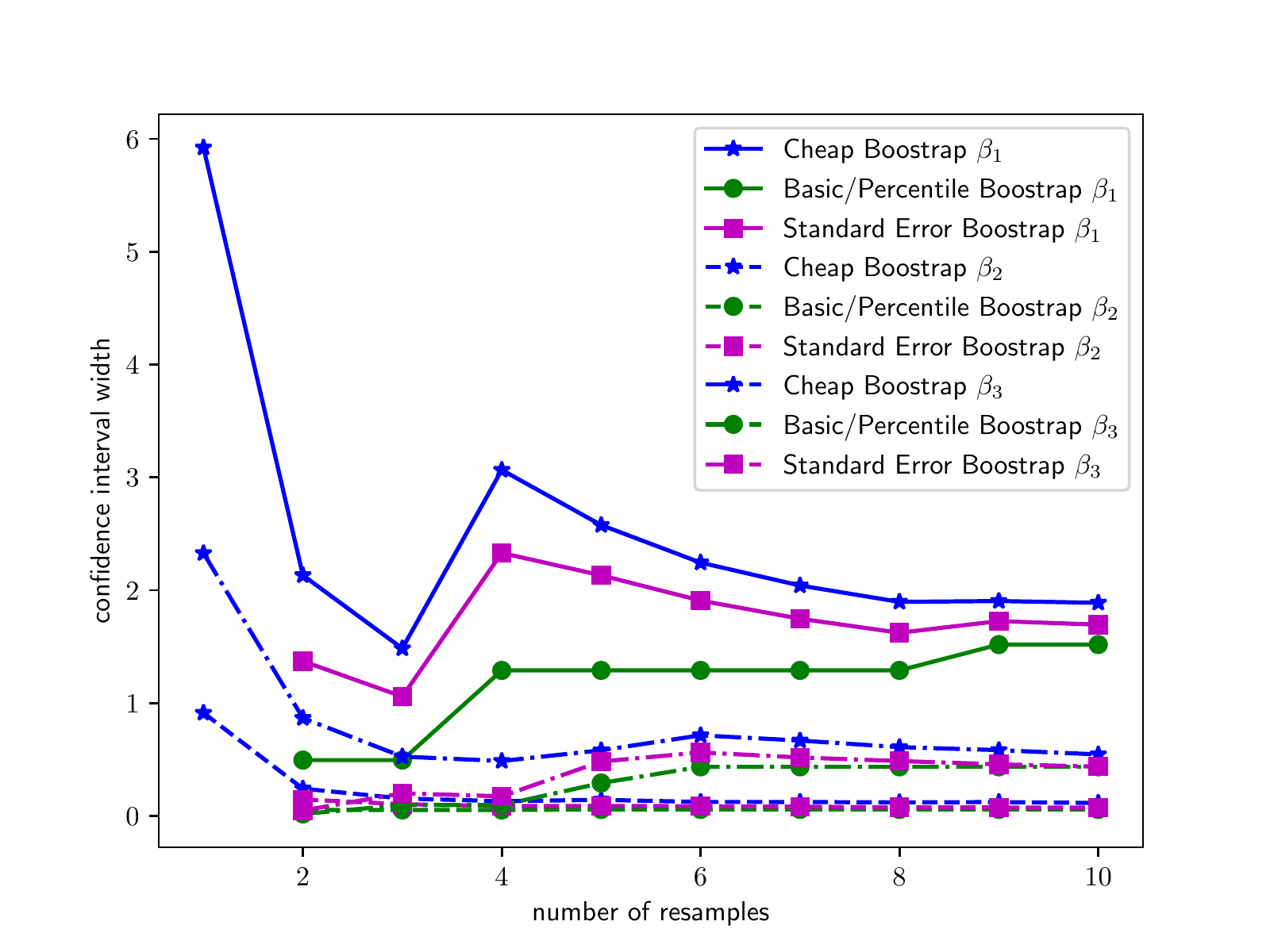}

\caption{Confidence interval widths of the first three $\beta$'s for different numbers of resamples in a real world logistic regression.}
\label{figure_real_l2}

\end{figure}

\subsection{Numerical Experiments with a Lower Nominal Level}\label{sec:lower_level}

In this section, we conduct a simulation study with the nominal level 70\% to further support the validity of the cheap bootstrap. We choose the ellipsoidal estimation and sinusoidal estimation presented in Section \ref{sec:num} as our model. All the settings are the same except that we use a different sample size $n=4\times10^4$ and a different dimension $p=9000$. Table \ref{alpha70_table} presents the empirical coverage and average interval width over $1000$ experimental repetitions. We can observe that the cheap bootstrap coverage probabilities are still close to the nominal level, and are the closest among all methods in all cases except two (sinusoidal $B=5$ and $B=10$) where percentile and standard error bootstraps each outperforms slightly. Regarding these exceptional cases, we note that the outperformance of the percentile bootstrap is likely a coincidence, because as a quantile-based method it cannot construct a symmetric 70\% level confidence interval from only 5 resamples. We have used the minimum and maximum of the 5 resamples to construct this percentile bootstrap confidence interval, whose actual nominal level should be close to 100\%. So it is likely by accident that percentile bootstrap coverage is closest to 70\% and in fact this coverage is far away from its actual nominal level around 100\%.

\begin{table}[ht]
\centering
\caption{Coverage probabilities (Pro.), confidence interval widths (Wid.) and running time (unit: second) of the numerical examples. The closest coverage probability to the nominal 70\% level among all methods in each setting is \textbf{bold}.}
\label{alpha70_table}
\medskip
\resizebox{\columnwidth}{!}{
\begin{tblr}{
  cells = {c},
  cell{1}{1} = {r=2}{},
  cell{1}{2} = {r=2}{},
  cell{1}{3} = {c=2}{},
  cell{1}{5} = {c=2}{},
  cell{1}{7} = {c=2}{},
  cell{1}{9} = {c=2}{},
  cell{1}{11} = {r=2}{},
  hline{1,11} = {-}{0.08em},
  hline{3,7} = {1-10}{0.03em},
  hline{3,7} = {11}{},
}
Example     & $B$ & Cheap Bootstrap &                      & Basic Bootstrap &                    & Percentile Bootstrap &                    & Standard Error Bootstrap &                    & Running Time \\
            &     & Pro.            & Wid.                 & Pro.            & Wid.               & Pro.                 & Wid.               & Pro.                     & Wid.               &              \\
~           & 1   & \textbf{73.0\%} & $9.828\times10^{-3}$   & N.A.            & N.A.               & N.A.                 & N.A.               & N.A.                     & N.A.               & 2            \\
Ellipsoidal & 2   & \textbf{70.0\%} & $7.568\times10^{-3}$   & 30.9\%          & $2.197\times10^{-3}$ & 7.7\%                & $2.197\times10^{-3}$ & 34.3\%                   & $3.220\times10^{-3}$ & 4            \\
estimation  & 5   & \textbf{68.5\%} & $6.639\times10^{-3}$ & 64.7\%          & $4.429\times10^{-3}$ & 16.5\%               & $4.429\times10^{-3}$ & 42.3\%                   & $3.717\times10^{-3}$ & 10           \\
            & 10  & \textbf{66.9\%} & $6.323\times10^{-3}$   & 60.1\%          & $3.756\times10^{-3}$ & 10.8\%               & $3.756\times10^{-3}$ & 41.8\%                   & $3.804\times10^{-3}$ & 20           \\
            & 1   & \textbf{70.3\%} & 0.148                & N.A.            & N.A.               & N.A.                 & N.A.               & N.A.                     & N.A.               & 2            \\
Sinusoidal  & 2   & \textbf{72.0\%} & 0.116                & 34.9\%          & 0.052              & 33.3\%               & 0.052              & 52.4\%                   & 0.077              & 4            \\
estimation  & 5   & 72.2\%          & 0.104                & 67.6\%          & 0.108              & \textbf{68.6\%}      & 0.108              & 65.6\%                   & 0.091              & 10           \\
            & 10  & 72.3\%          & 0.100                & 65.8\%          & 0.093              & 63.7\%               & 0.093              & \textbf{68.8\%}          & 0.094              & 19           
\end{tblr}
}
\end{table}

\subsection{Coverage Error Behavior with Respect to $B$ and $n$}\label{sec:trend}

In this section, we numerically study the cheap bootstrap coverage error behavior with respect to $B$ and $n$ and illustrate how it aligns with our theoretical bounds. We choose the model as the sinusoidal estimation in Section \ref{sec:num}. We fix the dimension $p=9000$ and vary $B$ and $n$. Figure \ref{colormap} displays the colormaps of the absolute values of empirical coverage errors (nominal level 95\%), where the $x$-axis represents $B$ and $y$-axis represents $n$. We cut the results of the basic and percentile bootstraps for the first few $B$'s because their errors are too large. From Figure \ref{colormap} (a), it appears that the cheap bootstrap coverage error does not change much in this regime of $n$. This matches to some extent Theorem \ref{CB_coverage_nonlinear} and Corollary \ref{concise_CB_coverage_nonlinear} that guarantee the coverage error would decrease as $n$ increases with a slow rate $1/\sqrt n$.  Further, when we fix $n$, the cheap bootstrap coverage error seems to lack clear trend and otherwise be quite stable as $B$ changes. On the other hand, the basic and percentile bootstrap coverage errors show a clear decreasing trend as $B$ increases. Their different behaviors are attributed to the different ideas behind them. The basic bootstrap and percentile bootstrap are quantile-based methods. As $B$ increases, the bootstrap quantile estimate is closer and closer to the true quantile, which leads to the improvement on the coverage error. However, the cheap bootstrap method relies on a totally different idea, i.e., it relies on approximate independence of the resamples from the original estimator and thus a $t$-distribution-based (with degree of freedom $B$) confidence interval can be constructed. Different $B$ just means a different pivotal $t$-distribution. There is no evident reason that the pivotal $t$-distribution with a larger degree of freedom $B$ will lead to a smaller coverage error.

\begin{figure}[ht]
\centering

\subfloat[Cheap bootstrap]{\includegraphics[width=0.45\textwidth]{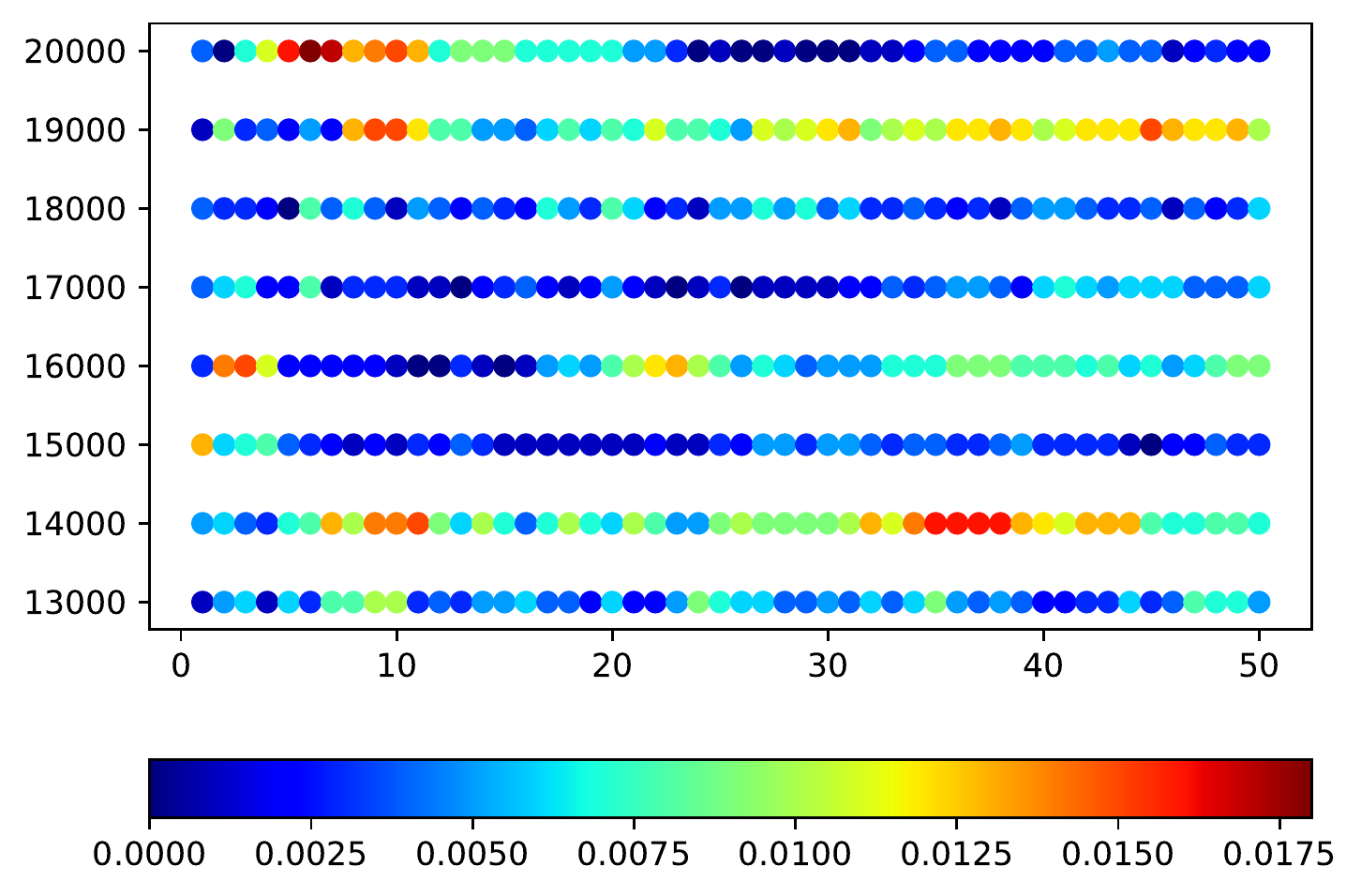}}\\
\subfloat[Basic bootstrap]{\includegraphics[width=0.45\textwidth]{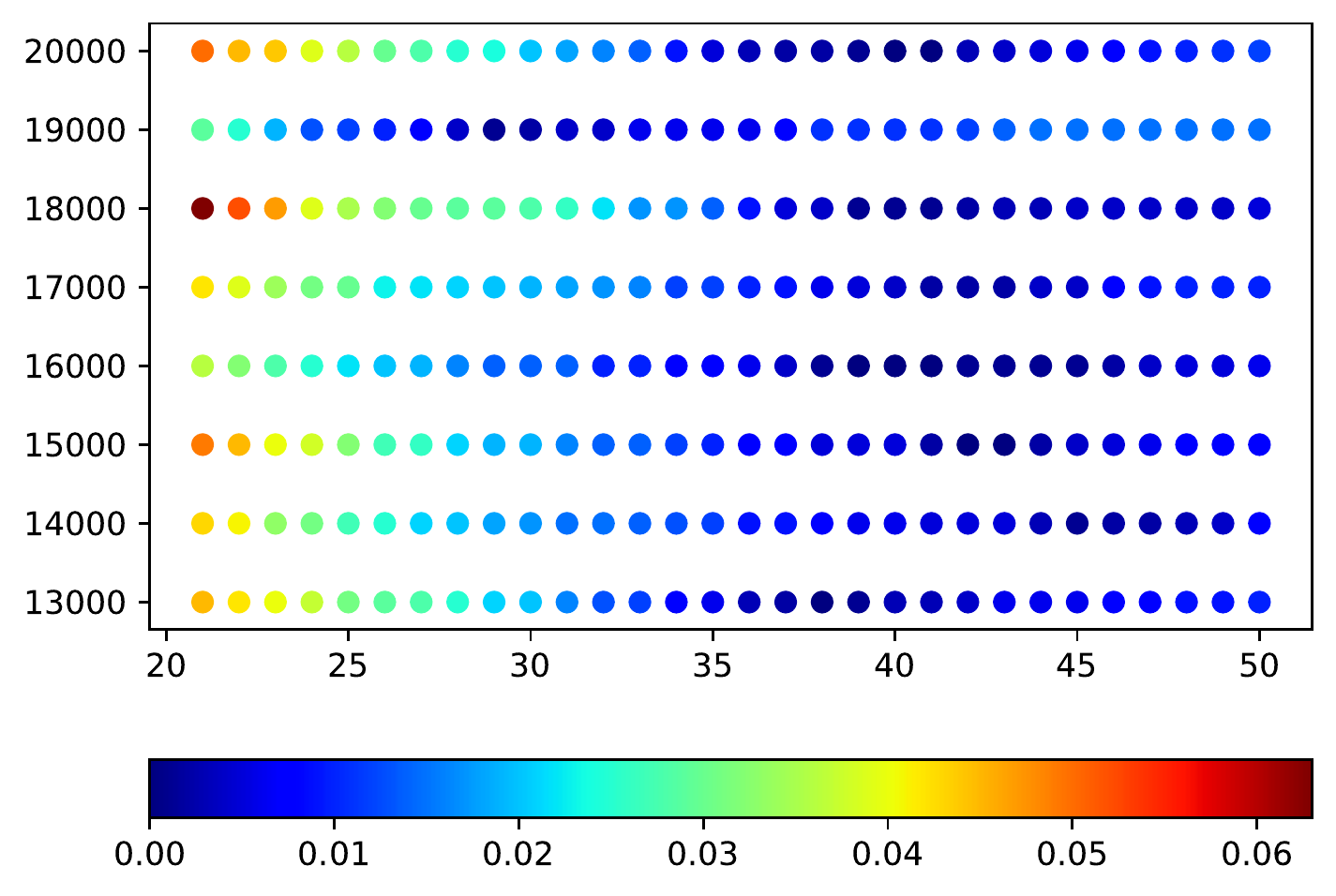}}
\subfloat[Percentile bootstrap]{\includegraphics[width=0.45\textwidth]{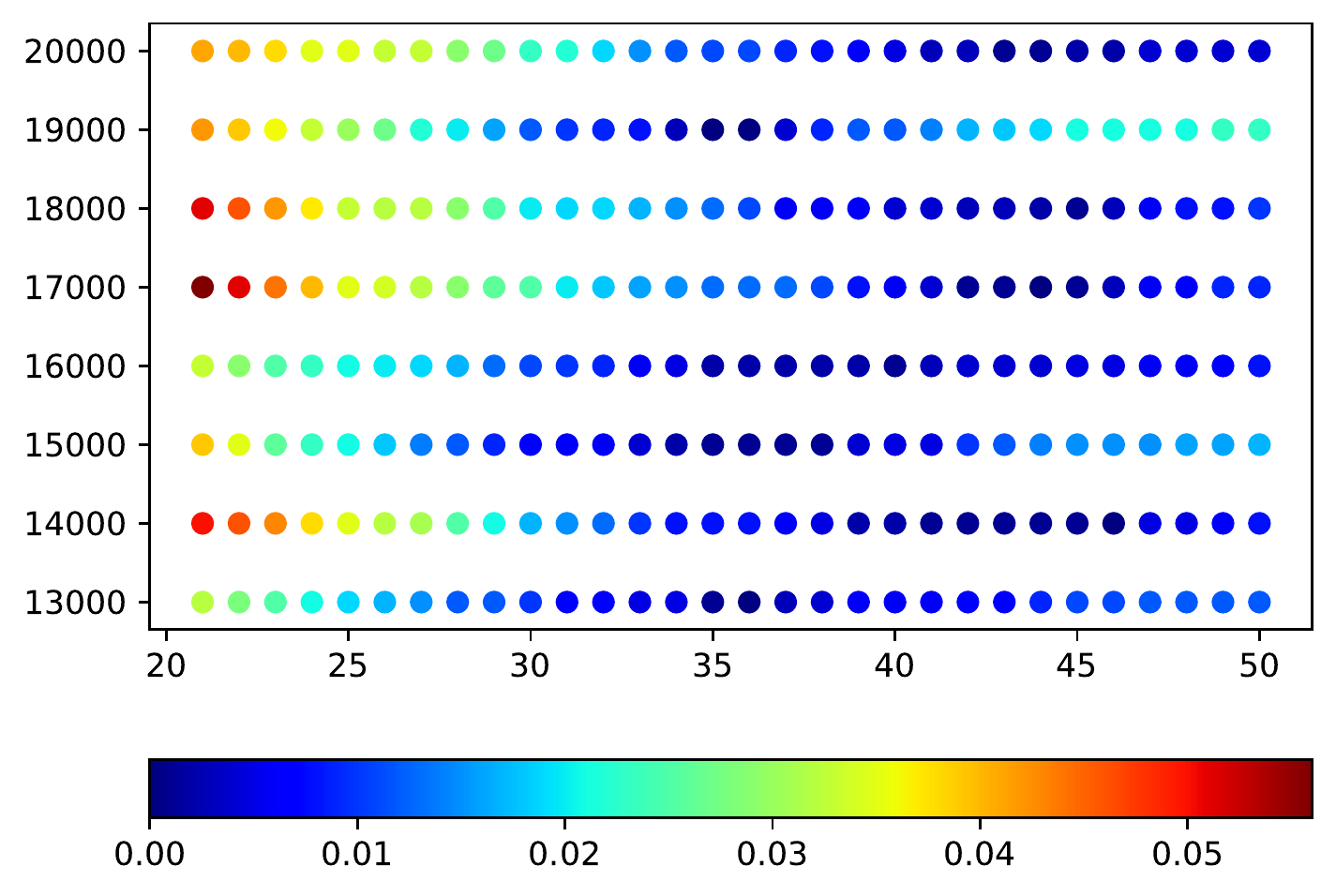}}

\caption{Colormaps of the absolute values of empirical coverage errors (nominal level 95\%) for the sinusoidal estimation.}%
\label{colormap}

\end{figure}

\section{Proofs}\label{sec:proofs}

\begin{proof}[Proof of Theorem \ref{general_CB_coverage}:]
We define $Q^{\ast}$ as the distribution of $\sqrt{n}(\hat{\psi}_{n}^{\ast}-\hat{\psi}_{n})$ conditional on $X_{1},\ldots,X_{n}$. With the repeated bootstrap resampling, we have that $\sqrt{n}(\hat{\psi}_{n}^{\ast b}-\hat{\psi}_{n}),b=1,\ldots,B$ are independent conditional on $X_{1},\ldots,X_{n}$. Then we can write the coverage probability as%
\begin{align*}
& P(\psi\in\lbrack\hat{\psi}_{n}-t_{B,1-\alpha/2}S_{n,B},\hat{\psi}_{n}+t_{B,1-\alpha/2}S_{n,B}])\\
& =P\left(  \left\vert \frac{\sqrt{n}(\hat{\psi}_{n}-\psi)}{\sqrt{\frac{1}{B}\sum_{b=1}^{B}(\sqrt{n}(\hat{\psi}_{n}^{\ast b}-\hat{\psi}_{n}))^{2}}}\right\vert \leq t_{B,1-\alpha/2}\right)  \\
& =E\left[  \int\cdots\int_{|\sqrt{n}(\hat{\psi}_{n}-\psi)|\leq t_{B,1-\alpha/2}\sqrt{\frac{1}{B}\sum_{b=1}^{B}z_{b}^{2}}}dQ^{\ast}(z_{B})\cdots dQ^{\ast}(z_{1})\right]
\end{align*}
where the expectation $E$ is taken with respect to $X_{1},\ldots,X_{n}$. If we write $\mathcal{A}$ as the event that%
\begin{equation}
\sup_{x\in\mathbb{R}}\left\vert Q^{\ast}((-\infty,x])-\Phi\left(  \frac{x}{\sigma}\right)  \right\vert \equiv\sup_{x\in\mathbb{R}}\left\vert P^{\ast}(\sqrt{n}(\hat{\psi}_{n}^{\ast}-\hat{\psi}_{n})\leq x)-\Phi\left(  \frac{x}{\sigma}\right)  \right\vert \leq\mathcal{E}_{2},\label{bootstrap_CLT_bound1}%
\end{equation}
then we know that $P(\mathcal{A}^{c})\leq\beta$. We consider the coverage probability intersected with $\mathcal{A}$, i.e.,
\begin{equation}
E\left[  \int\cdots\int_{|\sqrt{n}(\hat{\psi}_{n}-\psi)|\leq t_{B,1-\alpha/2}\sqrt{\frac{1}{B}\sum_{b=1}^{B}z_{b}^{2}}}dQ^{\ast}(z_{B})\cdots dQ^{\ast}(z_{1});\mathcal{A}\right]  .\label{coverage_probability_underE}%
\end{equation}
Note that conditional on $X_{1},\ldots,X_{n}$ and given $z_{1},\ldots,z_{B-1}$, the integral region for $z_{B}$ can be written as%
\[
\left\{  z_{B}:|\sqrt{n}(\hat{\psi}_{n}-\psi)|\leq t_{B,1-\alpha/2}\sqrt{\frac{1}{B}\sum_{b=1}^{B}z_{b}^{2}}\right\}  =(-\infty,-q]\cup\lbrack q,\infty)
\]
for some $q\geq0$. Therefore, applying (\ref{bootstrap_CLT_bound1}), we have%
\[
\left\vert \int_{|\sqrt{n}(\hat{\psi}_{n}-\psi)|\leq t_{B,1-\alpha/2}\sqrt{\frac{1}{B}\sum_{b=1}^{B}z_{b}^{2}}}dQ^{\ast}(z_{B})-\int_{|\sqrt{n}(\hat{\psi}_{n}-\psi)|\leq t_{B,1-\alpha/2}\sqrt{\frac{1}{B}\sum_{b=1}^{B}z_{b}^{2}}}dP_{0}(z_{B})\right\vert \leq2\mathcal{E}_{2},
\]
where $P_{0}$ is the distribution of $N(0,\sigma^{2})$. Plugging it into (\ref{coverage_probability_underE}), we have%
\begin{align*}
&  E\left[  \int\cdots\int_{|\sqrt{n}(g(\bar{X}_{n})-g(\mu))|\leq t_{B,1-\alpha/2}\sqrt{\frac{1}{B}\sum_{b=1}^{B}z_{b}^{2}}}dQ^{\ast}(z_{B})\cdots dQ^{\ast}(z_{1});\mathcal{A}\right]  \\
&  =E\left[  \int\cdots\int\int_{|\sqrt{n}(\hat{\psi}_{n}-\psi)|\leq t_{B,1-\alpha/2}\sqrt{\frac{1}{B}\sum_{b=1}^{B}z_{b}^{2}}}dP_{0}(z_{B})dQ^{\ast}(z_{B-1})\cdots dQ^{\ast}(z_{1});\mathcal{A}\right]  +R_{B},
\end{align*}
where the error $R_{B}$ satisfies%
\[
|R_{B}|\leq E\left[  \int\cdots\int2\mathcal{E}_{2}dQ^{\ast}(z_{B-1})\cdots dQ^{\ast}(z_{1});\mathcal{A}\right]  \leq2\mathcal{E}_{2}%
\]
By the same argument, we can further replace the remaining $Q^{\ast}(z_{i})$'s by $P_{0}(z_{i})$'s and obtain%
\begin{align*}
&  E\left[  \int\cdots\int_{|\sqrt{n}(\hat{\psi}_{n}-\psi)|\leq t_{B,1-\alpha/2}\sqrt{\frac{1}{B}\sum_{b=1}^{B}z_{b}^{2}}}dQ^{\ast}(z_{B})\cdots dQ^{\ast}(z_{1});\mathcal{A}\right]  \\
&  =E\left[  \int\cdots\int_{|\sqrt{n}(\hat{\psi}_{n}-\psi)|\leq t_{B,1-\alpha/2}\sqrt{\frac{1}{B}\sum_{b=1}^{B}z_{b}^{2}}}dP_{0}(z_{B})\cdots dP_{0}(z_{1});\mathcal{A}\right]  +\sum_{b=1}^{B}R_{b},
\end{align*}
where each error $R_{b}$ satisfies%
\[
|R_{b}|\leq2\mathcal{E}_{2}.
\]
Therefore, the coverage probability satisfies%
\begin{align}
&  E\left[  \int\cdots\int_{|\sqrt{n}(\hat{\psi}_{n}-\psi)|\leq t_{B,1-\alpha/2}\sqrt{\frac{1}{B}\sum_{b=1}^{B}z_{b}^{2}}}dQ^{\ast}(z_{B})\cdots dQ^{\ast}(z_{1})\right]  \nonumber\\
&  =E\left[  \int\cdots\int_{|\sqrt{n}(\hat{\psi}_{n}-\psi)|\leq t_{B,1-\alpha/2}\sqrt{\frac{1}{B}\sum_{b=1}^{B}z_{b}^{2}}}dQ^{\ast}(z_{B})\cdots dQ^{\ast}(z_{1});\mathcal{A}\right]  \nonumber\\
&  +E\left[  \int\cdots\int_{|\sqrt{n}(\hat{\psi}_{n}-\psi)|\leq t_{B,1-\alpha/2}\sqrt{\frac{1}{B}\sum_{b=1}^{B}z_{b}^{2}}}dQ^{\ast}(z_{B})\cdots dQ^{\ast}(z_{1});\mathcal{A}^{c}\right]  \nonumber\\
&  =E\left[  \int\cdots\int_{|\sqrt{n}(\hat{\psi}_{n}-\psi)|\leq t_{B,1-\alpha/2}\sqrt{\frac{1}{B}\sum_{b=1}^{B}z_{b}^{2}}}dP_{0}(z_{B})\cdots dP_{0}(z_{1});\mathcal{A}\right]  +\sum_{b=1}^{B}R_{b}\nonumber\\
&  +E\left[  \int\cdots\int_{|\sqrt{n}(\hat{\psi}_{n}-\psi)|\leq t_{B,1-\alpha/2}\sqrt{\frac{1}{B}\sum_{b=1}^{B}z_{b}^{2}}}dQ^{\ast}(z_{B})\cdots dQ^{\ast}(z_{1});\mathcal{A}^{c}\right]  \nonumber\\
&  =E\left[  \int\cdots\int_{|\sqrt{n}(\hat{\psi}_{n}-\psi)|\leq t_{B,1-\alpha/2}\sqrt{\frac{1}{B}\sum_{b=1}^{B}z_{b}^{2}}}dP_{0}(z_{B})\cdots dP_{0}(z_{1})\right]  +R_{\mathcal{A}^{c}}+\sum_{b=1}^{B}R_{b},\label{coverage_error1}%
\end{align}
where the additional error $R_{\mathcal{A}^{c}}$ is given by%
\begin{align*}
R_{\mathcal{A}^{c}} &  =E\left[  \int\cdots\int_{|\sqrt{n}(\hat{\psi}_{n}-\psi)|\leq t_{B,1-\alpha/2}\sqrt{\frac{1}{B}\sum_{b=1}^{B}z_{b}^{2}}}dQ^{\ast}(z_{B})\cdots dQ^{\ast}(z_{1});\mathcal{A}^{c}\right]  \\
&  -E\left[  \int\cdots\int_{|\sqrt{n}(\hat{\psi}_{n}-\psi)|\leq t_{B,1-\alpha/2}\sqrt{\frac{1}{B}\sum_{b=1}^{B}z_{b}^{2}}}dP_{0}(z_{B})\cdots dP_{0}(z_{1});\mathcal{A}^{c}\right]
\end{align*}
and it satisfies $|R_{\mathcal{A}^{c}}|\leq P(\mathcal{A}^{c})\leq\beta$. Now we will handle the distribution of $\sqrt{n}(\hat{\psi}_{n}-\psi)$ which is denoted by $Q_{0}$. Note that by Fubini's theorem we have
\begin{align*}
&  E\left[  \int\cdots\int_{|\sqrt{n}(\hat{\psi}_{n}-\psi)|\leq t_{B,1-\alpha/2}\sqrt{\frac{1}{B}\sum_{b=1}^{B}z_{b}^{2}}}dP_{0}(z_{B})\cdots dP_{0}(z_{1})\right]  \\
&  =\int\cdots\int_{|z_{0}|\leq t_{B,1-\alpha/2}\sqrt{\frac{1}{B}\sum_{b=1}^{B}z_{b}^{2}}}dQ_{0}(z_{0})dP_{0}(z_{B})\cdots dP_{0}(z_{1}).
\end{align*}
Given $z_{1},\ldots,z_{B}$, consider the innermost integral with respect to $Q_{0}$. By the finite-sample accuracy for $Q_{0}$, i.e.,
\[
\sup_{x\in\mathbb{R}}\left\vert P(\sqrt{n}(\hat{\psi}_{n}-\psi)\leq x)-\Phi\left(  \frac{x}{\sigma}\right)  \right\vert \equiv\sup_{x\in\mathbb{R}}\left\vert Q_{0}((-\infty,x])-P_{0}((-\infty,x])\right\vert\leq\mathcal{E}_{1},
\]
we have
\[
\left\vert \int_{|z_{0}|\leq t_{B,1-\alpha/2}\sqrt{\frac{1}{B}\sum_{b=1}^{B}z_{b}^{2}}}dQ_{0}(z_{0})-\int_{|z_{0}|\leq t_{B,1-\alpha/2}\sqrt{\frac{1}{B}\sum_{b=1}^{B}z_{b}^{2}}}dP_{0}(z_{0})\right\vert \leq2\mathcal{E}_{1}.
\]
Therefore,
\begin{align}
&  E\left[  \int\cdots\int_{|\sqrt{n}(\hat{\psi}_{n}-\psi)|\leq t_{B,1-\alpha/2}\sqrt{\frac{1}{B}\sum_{b=1}^{B}z_{b}^{2}}}dP_{0}(z_{B})\cdots dP_{0}(z_{1})\right]  \nonumber\\
&  =\int\cdots\int_{|z_{0}|\leq t_{B,1-\alpha/2}\sqrt{\frac{1}{B}\sum_{b=1}^{B}z_{b}^{2}}}dP_{0}(z_{0})dP_{0}(z_{B})\cdots dP_{0}(z_{1})+R_{0}\nonumber\\
&  =1-\alpha+R_{0},\label{coverage_error2}%
\end{align}
where the second equality follows from%
\[
\frac{Z_{0}}{\sqrt{\frac{1}{B}\sum_{b=1}^{B}Z_{b}^{2}}}\overset{d}{=}t_{B}%
\]
for i.i.d. $Z_{i}\sim N(0,\sigma^{2}),i=0,\ldots,B$ and the error $R_{0}$ satisfies%
\[
|R_{0}|\leq2\mathcal{E}_{1}.
\]
Plugging (\ref{coverage_error2}) into (\ref{coverage_error1}), we have%
\begin{align*}
&  E\left[  \int\cdots\int_{|\sqrt{n}(\hat{\psi}_{n}-\psi)|\leq t_{B,1-\alpha/2}\sqrt{\frac{1}{B}\sum_{b=1}^{B}z_{b}^{2}}}dQ^{\ast}(z_{B})\cdots dQ^{\ast}(z_{1})\right]  \\
&  =1-\alpha+R_{\mathcal{A}^{c}}+\sum_{b=0}^{B}R_{b}:=1-\alpha+R,
\end{align*}
where the overall error satisfies%
\[
|R|\leq2\mathcal{E}_{1}+2B\mathcal{E}_{2}+\beta.
\]

\end{proof}

\begin{proof}[Proof of Theorem \ref{general_other_coverages}:]
Recall that for a cumulative distribution function $F$ of a random variable, the $q$-th quantile is defined as $F^{-1}(q)=\inf\{x:F(x)\geq q\}$. We first prove a useful result: if the cumulative distribution functions of two random variables $X$ and $Y$ satisfy%
\begin{equation}
\sup_{t\in\mathbb{R}}|F_{X}(t)-F_{Y}(t)|\leq\varepsilon,\label{Fx_Fy_difference}%
\end{equation}
then for any $\alpha\in\lbrack0,1]$,
\begin{equation}
F_{Y}^{-1}(\alpha-\varepsilon)\leq F_{X}^{-1}(\alpha)\leq F_{Y}^{-1}(\alpha+\varepsilon). \label{general_quantile_inequality}%
\end{equation}
To prove it, we note that if $\alpha-\varepsilon<0$ then $-\infty=F_{Y}^{-1}(\alpha-\varepsilon)\leq F_{X}^{-1}(\alpha)$ trivially holds and if $\alpha+\varepsilon>1$ then $F_{X}^{-1}(\alpha)\leq F_{Y}^{-1}(\alpha+\varepsilon)=\infty$ trivially holds. So we assume $0\leq\alpha-\varepsilon\leq\alpha+\varepsilon\leq1$. Now let's prove the first inequality $F_{Y}^{-1}(\alpha-\varepsilon)\leq F_{X}^{-1}(\alpha)$. By the definition of $F_{X}^{-1}$ and right-continuity of $F_X$, we know that $F_{X}(F_{X}^{-1}(\alpha))\geq\alpha$. Therefore, by (\ref{Fx_Fy_difference}), we know that $F_{Y}(F_{X}^{-1}(\alpha))\geq \alpha-\varepsilon$, which implies $F_{Y}^{-1}(\alpha-\varepsilon)\leq F_{X}^{-1}(\alpha)$ by the definition of $F_{Y}^{-1}(\alpha-\varepsilon)$. This proves the first inequality in (\ref{general_quantile_inequality}). Interchanging the role of $X$ and $Y$, we have $F_{X}^{-1}(\alpha-\varepsilon)\leq F_{Y}^{-1}(\alpha)$. Replacing $\alpha$ by $\alpha+\varepsilon$, we obtain the second inequality $F_{X}^{-1}(\alpha)\leq F_{Y}^{-1}(\alpha+\varepsilon)$.

Now we consider the basic bootstrap. We write $\mathcal{A}$ as the event%
\[
\sup_{x\in\mathbb{R}}\left\vert P^{\ast}(\sqrt{n}(\hat{\psi}_{n}^{\ast}-\hat{\psi}_{n})\leq x)-\Phi\left(  \frac{x}{\sigma}\right)  \right\vert\leq\mathcal{E}_{2}.
\]
By our assumption, we have $P(\mathcal{A}^{c})\leq\beta$. Note that if $q_{\alpha/2}$ and $q_{1-\alpha/2}$ are the $\alpha/2$-th and $(1-\alpha/2)$-th quantiles of $\hat{\psi}_{n}^{\ast}-\hat{\psi}_{n}$ given $X_{1},\ldots,X_{n}$, then $\sqrt{n}q_{\alpha/2}$ and $\sqrt{n}q_{1-\alpha/2}$ are the $\alpha/2$-th and $(1-\alpha/2)$-th quantiles of $\sqrt{n}(\hat{\psi}_{n}^{\ast}-\hat{\psi}_{n})$ respectively given $X_{1},\ldots,X_{n}$. By the inequality (\ref{general_quantile_inequality}), when $\mathcal{A}$ happens, we have%
\[
\sigma z_{\alpha/2-\mathcal{E}_{2}}\leq\sqrt{n}q_{\alpha/2}\leq\sigma z_{\alpha/2+\mathcal{E}_{2}},
\]
where $z_{q}$ is the $q$-th quantile of the standard normal. This inequality implies%
\begin{equation}
P(\sqrt{n}(\hat{\psi}_{n}-\psi)<\sigma z_{\alpha/2-\mathcal{E}_{2}};\mathcal{A})\leq P(\sqrt{n}(\hat{\psi}_{n}-\psi)<\sqrt{n}q_{\alpha/2};\mathcal{A})\label{lower_bound}%
\end{equation}
and%
\begin{equation}
P(\sqrt{n}(\hat{\psi}_{n}-\psi)<\sqrt{n}q_{\alpha/2};\mathcal{A})\leq P(\sqrt{n}(\hat{\psi}_{n}-\psi)<\sigma z_{\alpha/2+\mathcal{E}_{2}};\mathcal{A}).\label{upper_bound}%
\end{equation}
Next, we notice that%
\begin{align*}
&  \sup_{x\in\mathbb{R}}\left\vert P(\sqrt{n}(\hat{\psi}_{n}-\psi)\leq x)-\Phi\left(  \frac{x}{\sigma}\right)  \right\vert \leq\mathcal{E}_{1}\\
\Leftrightarrow &  \sup_{x\in\mathbb{R}}\left\vert P(\sqrt{n}(\hat{\psi}_{n}-\psi)<x)-\Phi\left(  \frac{x}{\sigma}\right)  \right\vert \leq\mathcal{E}_{1}.
\end{align*}
Thus, (\ref{lower_bound}) implies that%
\begin{align*}
&  P(\sqrt{n}(\hat{\psi}_{n}-\psi)<\sqrt{n}q_{\alpha/2};\mathcal{A})\\
&  \geq P(\sqrt{n}(\hat{\psi}_{n}-\psi)<\sigma z_{\alpha/2-\mathcal{E}_{2}};\mathcal{A})\\
&  \geq P(\sqrt{n}(\hat{\psi}_{n}-\psi)<\sigma z_{\alpha/2-\mathcal{E}_{2}})-P(\mathcal{A}^{c})\\
&  \geq\Phi\left(  \frac{\sigma z_{\alpha/2-\mathcal{E}_{2}}}{\sigma}\right)-\mathcal{E}_{1}-\beta\\
&  =\frac{\alpha}{2}-\mathcal{E}_{1}-\mathcal{E}_{2}-\beta.
\end{align*}
Similarly, (\ref{upper_bound}) implies that%
\[
P(\sqrt{n}(\hat{\psi}_{n}-\psi)<\sqrt{n}q_{\alpha/2};\mathcal{A})\leq P(\sqrt{n}(\hat{\psi}_{n}-\psi)<\sigma z_{\alpha/2+\mathcal{E}_{2}})\leq\frac{\alpha}{2}+\mathcal{E}_{1}+\mathcal{E}_{2}.
\]
Therefore, we have the following two-sided bound%
\[
\frac{\alpha}{2}-\mathcal{E}_{1}-\mathcal{E}_{2}-\beta\leq P(\sqrt{n}(\hat{\psi}_{n}-\psi)<\sqrt{n}q_{\alpha/2};\mathcal{A})\leq\frac{\alpha}{2}+\mathcal{E}_{1}+\mathcal{E}_{2}.
\]
For the $(1-\alpha/2)$-th quantile, we can also derive a similar bound%
\[
1-\frac{\alpha}{2}-\mathcal{E}_{1}-\mathcal{E}_{2}-\beta\leq P(\sqrt{n}(\hat{\psi}_{n}-\psi)\leq\sqrt{n}q_{1-\alpha/2};\mathcal{A})\leq1-\frac{\alpha}{2}+\mathcal{E}_{1}+\mathcal{E}_{2}.
\]
So we have
\[
|P(\sqrt{n}q_{\alpha/2}\leq\sqrt{n}(\hat{\psi}_{n}-\psi)\leq\sqrt{n}q_{1-\alpha/2};\mathcal{A})-(1-\alpha)|\leq2\mathcal{E}_{1}+2\mathcal{E}_{2}+\beta,
\]
which gives rise to%
\begin{align*}
&  |P(\sqrt{n}q_{\alpha/2}\leq\sqrt{n}(\hat{\psi}_{n}-\psi)\leq\sqrt{n}q_{1-\alpha/2})-(1-\alpha)|\\
&  \leq|P(\sqrt{n}q_{\alpha/2}\leq\sqrt{n}(\hat{\psi}_{n}-\psi)\leq\sqrt{n}q_{1-\alpha/2};\mathcal{A})-(1-\alpha)|\\
&  +P(\sqrt{n}q_{\alpha/2}\leq\sqrt{n}(\hat{\psi}_{n}-\psi)\leq\sqrt{n}q_{1-\alpha/2};\mathcal{A}^{c})\\
&  \leq2\mathcal{E}_{1}+2\mathcal{E}_{2}+\beta+P(\mathcal{A}^{c}\mathcal{)}\\
&  \leq2\mathcal{E}_{1}+2\mathcal{E}_{2}+2\beta,
\end{align*}
or equivalently%
\[
|P(\hat{\psi}_{n}-q_{1-\alpha/2}\leq\psi\leq\hat{\psi}_{n}-q_{\alpha/2})-(1-\alpha)|\leq2\mathcal{E}_{1}+2\mathcal{E}_{2}+2\beta.
\]

The result for the percentile bootstrap follows similarly but we need to use the symmetry of $N(0,\sigma^{2})$. Note that%
\begin{align*}
&  \sup_{x\in\mathbb{R}}\left\vert P(\sqrt{n}(\hat{\psi}_{n}-\psi)\leq x)-\Phi\left(  \frac{x}{\sigma}\right)  \right\vert \leq\mathcal{E}_{1}\\
\Leftrightarrow &  \sup_{x\in\mathbb{R}}\left\vert P(\sqrt{n}(\hat{\psi}_{n}-\psi)<x)-\Phi\left(  \frac{x}{\sigma}\right)  \right\vert \leq\mathcal{E}_{1}%
\end{align*}
and the latter can be rewritten as%
\begin{align*}
&  \sup_{x\in\mathbb{R}}\left\vert P(\sqrt{n}(\hat{\psi}_{n}-\psi)<x)-\Phi\left(  \frac{x}{\sigma}\right)  \right\vert \\
= &  \sup_{x\in\mathbb{R}}\left\vert P(\sqrt{n}(\hat{\psi}_{n}-\psi)<-x)-\Phi\left(  \frac{-x}{\sigma}\right)  \right\vert \\
= &  \sup_{x\in\mathbb{R}}\left\vert P(\sqrt{n}(\psi-\hat{\psi}_{n})>x)-\Phi\left(  \frac{-x}{\sigma}\right)  \right\vert \\
= &  \sup_{x\in\mathbb{R}}\left\vert (1-P(\sqrt{n}(\psi-\hat{\psi}_{n})>x))-\left(  1-\Phi\left(  \frac{-x}{\sigma}\right)  \right)  \right\vert\\
= &  \sup_{x\in\mathbb{R}}\left\vert P(\sqrt{n}(\psi-\hat{\psi}_{n})\leq x)-\Phi\left(  \frac{x}{\sigma}\right)  \right\vert \leq\mathcal{E}_{1},
\end{align*}
where the last equality uses the symmetry of $N(0,\sigma^{2})$, i.e., $1-\Phi(-x/\sigma)=\Phi(x/\sigma)$. Moreover, if $q_{\alpha/2}$ and $q_{1-\alpha/2}$ are the $\alpha/2$-th and $(1-\alpha/2)$-th quantiles of $\hat{\psi}_{n}^{\ast}$ given $X_{1},\ldots,X_{n}$, then $\sqrt{n}(q_{\alpha/2}-\hat{\psi}_{n})$ and $\sqrt{n}(q_{1-\alpha/2}-\hat{\psi}_{n})$ are the $\alpha/2$-th and $(1-\alpha/2)$-th quantiles of $\sqrt{n}(\hat{\psi}_{n}^{\ast}-\hat{\psi}_{n})$ given $X_{1},\ldots,X_{n}$. Therefore, the proof for the basic bootstrap also applies if we replace $\sqrt{n}q_{\alpha/2}$, $\sqrt{n}q_{1-\alpha/2}$ and $\sqrt{n}(\hat{\psi}_{n}-\psi)$ in that proof by $\sqrt{n}(q_{\alpha/2}-\hat{\psi}_{n})$, $\sqrt{n}(q_{1-\alpha/2}-\hat{\psi}_{n})$ and $\sqrt{n}(\psi-\hat{\psi}_{n})$ respectively. In particular, the final result now reads as follows%
\[
|P(\sqrt{n}(q_{\alpha/2}-\hat{\psi}_{n})\leq\sqrt{n}(\psi-\hat{\psi}_{n})\leq\sqrt{n}(q_{1-\alpha/2}-\hat{\psi}_{n}))-(1-\alpha)|\leq2\mathcal{E}_{1}+2\mathcal{E}_{2}+2\beta
\]
or equivalently%
\[
|P(q_{\alpha/2}\leq\psi\leq q_{1-\alpha/2})-(1-\alpha)|\leq2\mathcal{E}_{1}+2\mathcal{E}_{2}+2\beta.
\]
This completes our proof.
\end{proof}

To prove Theorem \ref{CB_coverage_nonlinear}, we first need to prove two lemmas regarding (\ref{general_Berry_Esseen}) and
(\ref{general_bootstrap_CLT}) respectively.

The following lemma is from Theorem 2.11 in \citet{pinelis2016optimal} which establishes the Berry-Esseen theorem in the multivariate delta method in the form of (\ref{general_Berry_Esseen}).

\begin{lemma}
\label{berry_esseen_delta_method}Suppose that $X_{1},\ldots,X_{n}$ are i.i.d. random vectors in $\mathbb{R}^{p}$ satisfying $E[X]=\mu$, $Var(X)=\Sigma$, $m_{31}:=E[|\nabla g(\mu)^{\top}(X-\mu)|^{3}]<\infty$ and $m_{32}:=E[||X-\mu||^{3}]<\infty$. Suppose $g(x)$ satisfies Assumption \ref{smoothness_assumption} and $\sigma^{2}:=\nabla g(\mu)^{\top}\Sigma\nabla g(\mu)>0$. Then there is a universal constant $C>0$ s.t.%
\begin{align*}
&  \sup_{x\in\mathbb{R}}\left\vert P(\sqrt{n}(g(\bar{X}_{n})-g(\mu))\leq x)-\Phi\left(  \frac{x}{\sigma}\right)  \right\vert \\
&  \leq C\left(  \frac{m_{31}}{\sqrt{n}\sigma^{3}}+\frac{C_{H_{g}}m_{31}^{1/3}tr(\Sigma)}{\sqrt{n}\sigma^{2}}+\frac{C_{H_{g}}m_{32}^{2/3}}{n^{5/6}\sigma}+\frac{C_{H_{g}}m_{31}^{1/3}m_{32}^{2/3}}{n\sigma^{2}}\right).
\end{align*}

\end{lemma}

\begin{proof}[Proof of Lemma \ref{berry_esseen_delta_method}:]
Define $f(x)=g(x+\mu)-g(\mu)$ and its linearization $L(x)=\nabla g(\mu)^{\top}x$. Then by the second order Taylor expansion of $f(x)$ and boundedness property of $H_{g}$ in Assumption \ref{smoothness_assumption}, we can see (2.1) in \citet{pinelis2016optimal} holds for $M_{\epsilon}=C_{H_{g}}$ and any $\epsilon>0$. By Theorem 2.11 in \citet{pinelis2016optimal} with $V=X-\mu$, $c_{\ast}=1/2$ and $\epsilon\rightarrow\infty$, we have%
\begin{align}
&  \sup_{x\in\mathbb{R}}\left\vert P(\sqrt{n}(g(\bar{X}_{n})-g(\mu))\leq x)-\Phi\left(  \frac{x}{\sigma}\right)  \right\vert \nonumber\\
&  \leq\frac{\mathfrak{K}_{0}+\mathfrak{K}_{1}m_{31}/\sigma^{3}}{\sqrt{n}}+\frac{\mathfrak{K}_{20}+\mathfrak{K}_{21}m_{31}^{1/3}/\sigma}{\sqrt{n}}tr(\Sigma)+\frac{\mathfrak{K}_{30}+\mathfrak{K}_{31}m_{31}^{1/3}/\sigma}{\sqrt{n}}m_{32}^{2/3}, \label{explicit_berry_esseen}%
\end{align}
where the additional term $\mathfrak{K}_{\epsilon}$ in Theorem 2.11 vanishes as $\epsilon\rightarrow\infty$. By the definition of these $\mathfrak{K}$'s with $c_{\ast}=1/2$ in (2.30) in \citet{pinelis2016optimal}, we can see there is a universal constant $C>0$ s.t.%
\[
\mathfrak{K}_{0}\leq C,\mathfrak{K}_{1}\leq C,\mathfrak{K}_{20}\leq C\frac{C_{H_{g}}}{\sigma},\mathfrak{K}_{21}\leq C\frac{C_{H_{g}}}{\sigma},\mathfrak{K}_{30}\leq C\frac{C_{H_{g}}}{\sigma n^{1/3}},\mathfrak{K}_{31}\leq C\frac{C_{H_{g}}}{\sigma n^{1/2}}.
\]
Moreover, by Holder's inequality, $m_{31}=E[|\nabla g(\mu)^{\top}(X-\mu)|^{3}]\geq E[|\nabla g(\mu)^{\top}(X-\mu)|^{2}]^{3/2}=\sigma^{3}$, which implies that $\mathfrak{K}_{0},\mathfrak{K}_{20}$ can be absorbed into $\mathfrak{K}_{1}m_{31}/\sigma^{3},\mathfrak{K}_{21}m_{31}^{1/3}/\sigma$ respectively by choosing a larger $C$. Therefore, (\ref{explicit_berry_esseen}) can be written as%
\begin{align*}
&  \sup_{x\in\mathbb{R}}\left\vert P(\sqrt{n}(g(\bar{X}_{n})-g(\mu))\leq x)-\Phi\left(  \frac{x}{\sigma}\right)  \right\vert \\
&  \leq C\left(  \frac{m_{31}}{\sqrt{n}\sigma^{3}}+\frac{C_{H_{g}}m_{31}^{1/3}tr(\Sigma)}{\sqrt{n}\sigma^{2}}+\frac{C_{H_{g}}m_{32}^{2/3}}{n^{5/6}\sigma}+\frac{C_{H_{g}}m_{31}^{1/3}m_{32}^{2/3}}{n\sigma^{2}}\right).
\end{align*}
This concludes our proof.
\end{proof}

Next we prove the finite-sample accuracy (\ref{general_bootstrap_CLT}) for the bootstrap estimator by extracting the dependence on problem parameters in Theorem 4.2 in \citet{zhilova2020nonclassical} and combining it with Lemma \ref{berry_esseen_delta_method}.

\begin{lemma}
\label{highdim_CLT_nonlinear} Suppose the conditions in Theorem \ref{CB_coverage_nonlinear} hold. Then with probability at least $1-6/n$ we have
\begin{align*}
&  \sup_{x\in\mathbb{R}}\left\vert P^{\ast}(\sqrt{n}(g(\bar{X}_{n}^{\ast})-g(\bar{X}_{n}))\leq x)-\Phi\left(  \frac{x}{\sigma}\right)  \right\vert \\
&  \leq C\left(  \frac{m_{31}}{\sqrt{n}\sigma^{3}}+\frac{C_{H_{g}}m_{31}^{1/3}tr(\Sigma)}{\sqrt{n}\sigma^{2}}+\frac{C_{H_{g}}m_{32}^{2/3}}{n^{5/6}\sigma}+\frac{C_{H_{g}}m_{31}^{1/3}m_{32}^{2/3}}{n\sigma^{2}}\right.\\
&  +\frac{C_{H_{g}}\tau^{2}}{C_{\nabla g}\sqrt{\lambda_{\Sigma}}}\left(1+\frac{\log n}{p}\right)  \sqrt{\frac{p}{n}}+\frac{||E[(X-\mu)^{3}]||}{\lambda_{\Sigma}^{3/2}}\frac{1}{\sqrt{n}}+\frac{\tau^{3}}{\lambda_{\Sigma}^{3/2}}\left(  1+\frac{\log n}{p}\right)  ^{3/2}\frac{1}{\sqrt{n}}\\
&  \left.  +\frac{\tau^{4}\sqrt{p}}{\lambda_{\Sigma}^{2}n}\left(  1+\frac{\log n}{p}\right)  ^{1/2}+\frac{\tau^{2}\sqrt{p}}{\lambda_{\Sigma}n}\left(1+\frac{\log n}{p}\right)  ^{1/2}+\frac{\tau^{3}\sqrt{p}}{\lambda_{\Sigma}^{3/2}n}\left(  1+\frac{\log n}{p}\right)  \right) \\
&  +C_{1}\left(  \frac{\tau^{4}(\log n)^{3/2}}{\lambda_{\Sigma}^{2}\sqrt{n}}+\frac{\tau^{2}(\log n)^{3/2}}{\lambda_{\Sigma}\sqrt{n}}+\frac{\tau^{3}}{\lambda_{\Sigma}^{3/2}\sqrt{n}}\left(  1+\frac{\log n}{p}\right)^{1/2}(\log n+\log p)\sqrt{\log n}\right)  ,
\end{align*}
where $m_{31},m_{32}$ and $\sigma^{2}$ are defined in Lemma \ref{berry_esseen_delta_method}, $C$ is a universal constant and $C_{1}$ only depends on $C_{X}$.
\end{lemma}

\begin{proof}[Proof of Lemma \ref{highdim_CLT_nonlinear}:]
We use Theorem 4.2 in \citet{zhilova2020nonclassical} to prove this lemma. First, we verify the conditions of Theorem 4.2 with $K=3$. For $K=3$, by Remark 2.1, we can take $U_{i}\equiv0,Y_{i}-\mu=Z_{i}-\mu\sim N(0,\Sigma_{z})$ independent of $X_{1},\ldots,X_{n}$ with $\Sigma_{z}=\Sigma$ which satisfies (2.1) and (2.2) for approximating $X_{i}-\mu$ (since $X_{i}$ is not centered in Theorem 4.2, $Y_{i}$ should also be non-centered). In this case, we can see $C_{z}:=||\Sigma_{z}^{-1/2}||=||\Sigma^{-1/2}||=1/\sqrt{\lambda_{\Sigma}}$. Similarly, $C_{X}:=||\Sigma^{-1/2}||=1/\sqrt{\lambda_{\Sigma}}$. Other conditions about $f$ in Theorem 4.2 have already been assumed in the statement of this lemma. Thus, by Theorem 4.2, it holds with probability at least $1-6e^{-x}$ for $x>0$:%
\begin{align}
&  \sup_{t\in\mathbb{R}}|P(\sqrt{n}(g(\bar{X}_{n})-g(\mu))\leq t)-P^{\ast}(\sqrt{n}(g(\bar{X}_{n}^{\ast})-g(\bar{X}_{n}))\leq t)|\nonumber\\
&  \leq2C_{H_{g}}\sqrt{\frac{1}{\lambda_{\Sigma}}}\tau^{2}\left(1+2\sqrt{\frac{x}{p}}+\frac{2x}{p}\right)  \frac{1}{C_{\nabla g}}\sqrt{\frac{p}{n}}+C_{\mathcal{B},\text{i.i.d.}}\frac{1}{\lambda_{\Sigma}^{3/2}}C_{M,3}\frac{1}{\sqrt{n}}\nonumber\\
&  +2C_{\mathcal{B},\text{i.i.d.}}\left(  1+2\sqrt{\frac{x}{p}}+\frac{2x}{p}\right)  ^{3/2}\frac{1}{\lambda_{\Sigma}^{3/2}}\tau^{3}\frac{1}{\sqrt{n}}+R_{1,n,3}\label{accuracy_bootstrap1}%
\end{align}
where $C_{\mathcal{B},\text{i.i.d.}}>0$ is a constant only depending on $K$ and thus is a universal constant for $K=3$, $C_{M,3}$ is defined as%
\[
C_{M,3}=||E[(X-\mu)^{3}]||+||E[(Y_{1}-\mu)^{3}]||
\]
and $R_{1,n,3}$ is defined in (B.14) as%
\begin{equation}
R_{1,n,3}=\frac{\tau^{4}C_{x,2}}{\lambda_{\Sigma}^{2}\sqrt{2n}}+\frac{4\tau^{2}C_{x,2}}{\lambda_{\Sigma}\sqrt{2n}}+\frac{\tilde{C}_{\phi,2}\tau^{3}C_{x,3}}{2\lambda_{\Sigma}^{3/2}\sqrt{n}}.\label{remainder}%
\end{equation}
Since $Y_{1}-\mu\sim N(0,\Sigma_{z})$, the tensor power $(Y_{1}-\mu)^{3}$ has expectation zero and thus $C_{M,3}$ can be simplified as $C_{M,3}=||E[(X-\mu)^{3}]||$. $\tilde{C}_{\phi,2}$ in (\ref{remainder}) is defined in (A.13) and according to Remark A.1, it depends on the choice of $\phi(t)$ which is a $K=3$ times continuously differentiable smooth approximation of the indicator function $1\{t\leq0\}$. Once we fix such $\phi(t)$, $\tilde{C}_{\phi,2}$ is a universal constant. $C_{x,2}$ and $C_{x,3}$ in (\ref{remainder}) are defined in (B.27) of Theorem B.1 as%
\[
C_{x,2}=C_{1}((x+\log n)\vee\sqrt{x+\log n})\sqrt{2x}+2\sqrt{\frac{p}{n}}\left(  1+2\sqrt{\frac{x+\log n}{p}}+\frac{2(x+\log n)}{p}\right)  ^{1/2},
\]%
\begin{align*}
C_{x,3} &  =C_{1}\left(  1+2\sqrt{\frac{x+\log n+\log p}{p}}+\frac{2(x+\log n+\log p)}{p}\right)  ^{1/2}\\
& \times((x+\log n+\log p)\vee\sqrt{x+\log n+\log p})\sqrt{2x}\\
&  +3\sqrt{\frac{p}{n}}\left(  1+2\sqrt{\frac{x+\log n}{p}}+\frac{2(x+\log n)}{p}\right)  ,
\end{align*}
where $C_{1}$ is a constant only depending on the density bound $C_{X}$ based on the proof of Theorem B.1. Besides, since we assume $n\geq3$, we know that $x+\log n\geq\sqrt{x+\log n}$ and $x+\log n+\log p\geq\sqrt{x+\log n+\log p}$ holds for any $x>0$. Therefore, $C_{x,2}$ and $C_{x,3}$ can be simplified as%
\[
C_{x,2}=C_{1}(x+\log n)\sqrt{2x}+2\sqrt{\frac{p}{n}}\left(  1+2\sqrt{\frac{x+\log n}{p}}+\frac{2(x+\log n)}{p}\right)  ^{1/2},
\]%
\begin{align*}
C_{x,3} &  =C_{1}\left(  1+2\sqrt{\frac{x+\log n+\log p}{p}}+\frac{2(x+\log n+\log p)}{p}\right)  ^{1/2}(x+\log n+\log p)\sqrt{2x}\\
&  +3\sqrt{\frac{p}{n}}\left(  1+2\sqrt{\frac{x+\log n}{p}}+\frac{2(x+\log n)}{p}\right)  .
\end{align*}
Moreover, for any $y>0$, we always have $1+y\leq1+2\sqrt{y}+2y\leq4(1+y)$. Therefore, the remainder term (\ref{remainder}) can be bounded in a more compact way up to some constants as follows:%
\begin{align*}
R_{1,n,3} &  \leq C_{1}\left(  \frac{\tau^{4}}{\lambda_{\Sigma}^{2}\sqrt{n}}(x+\log n)\sqrt{x}+\frac{\tau^{2}}{\lambda_{\Sigma}\sqrt{n}}(x+\log n)\sqrt{x}\right.  \\
&  \left.  +\frac{\tau^{3}}{\lambda_{\Sigma}^{3/2}\sqrt{n}}\left(1+\frac{x+\log n+\log p}{p}\right)  ^{1/2}(x+\log n+\log p)\sqrt{x}\right)  \\
&  +C\left(  \frac{\tau^{4}\sqrt{p}}{\lambda_{\Sigma}^{2}n}\left(1+\frac{x+\log n}{p}\right)  ^{1/2}+\frac{\tau^{2}\sqrt{p}}{\lambda_{\Sigma}n}\left(  1+\frac{x+\log n}{p}\right)  ^{1/2}+\frac{\tau^{3}\sqrt{p}}{\lambda_{\Sigma}^{3/2}n}\left(  1+\frac{x+\log n}{p}\right)  \right)  ,
\end{align*}
where $C$ is a universal constant and $C_{1}$ only depends on $C_{X}$. Plugging it back to (\ref{accuracy_bootstrap1}), we can similarly write (\ref{accuracy_bootstrap1}) in a more compact way:
\begin{align*}
&  \sup_{t\in\mathbb{R}}|P(\sqrt{n}(g(\bar{X}_{n})-g(\mu))\leq t)-P^{\ast}(\sqrt{n}(g(\bar{X}_{n}^{\ast})-g(\bar{X}_{n}))\leq t)|\\
&  \leq C\left(  \frac{C_{H_{g}}\tau^{2}}{C_{\nabla g}\sqrt{\lambda_{\Sigma}}}\left(  1+\frac{x}{p}\right)  \sqrt{\frac{p}{n}}+\frac{||E[(X-\mu)^{3}]||}{\lambda_{\Sigma}^{3/2}}\frac{1}{\sqrt{n}}+\frac{\tau^{3}}{\lambda_{\Sigma}^{3/2}}\left(  1+\frac{x}{p}\right)  ^{3/2}\frac{1}{\sqrt{n}}\right.\\
&  \left.  +\frac{\tau^{4}\sqrt{p}}{\lambda_{\Sigma}^{2}n}\left(1+\frac{x+\log n}{p}\right)  ^{1/2}+\frac{\tau^{2}\sqrt{p}}{\lambda_{\Sigma}n}\left(  1+\frac{x+\log n}{p}\right)  ^{1/2}+\frac{\tau^{3}\sqrt{p}}{\lambda_{\Sigma}^{3/2}n}\left(  1+\frac{x+\log n}{p}\right)  \right)  \\
&  +C_{1}\left(  \frac{\tau^{4}}{\lambda_{\Sigma}^{2}\sqrt{n}}(x+\log n)\sqrt{x}+\frac{\tau^{2}}{\lambda_{\Sigma}\sqrt{n}}(x+\log n)\sqrt{x}\right.\\
&  \left.  +\frac{\tau^{3}}{\lambda_{\Sigma}^{3/2}\sqrt{n}}\left(1+\frac{x+\log n+\log p}{p}\right)  ^{1/2}(x+\log n+\log p)\sqrt{x}\right)  ,
\end{align*}
where $C$ is a universal constant and $C_{1}$ only depends on $C_{X}$. Now we choose $x=\log n$. Then with probability at least $1-6/n$, we have%
\begin{align*}
&  \sup_{t\in\mathbb{R}}|P(\sqrt{n}(g(\bar{X}_{n})-g(\mu))\leq t)-P^{\ast}(\sqrt{n}(g(\bar{X}_{n}^{\ast})-g(\bar{X}_{n}))\leq t)|\\
&  \leq C\left(  \frac{C_{H_{g}}\tau^{2}}{C_{\nabla g}\sqrt{\lambda_{\Sigma}}}\left(  1+\frac{\log n}{p}\right)  \sqrt{\frac{p}{n}}+\frac{||E[(X-\mu)^{3}]||}{\lambda_{\Sigma}^{3/2}}\frac{1}{\sqrt{n}}+\frac{\tau^{3}}{\lambda_{\Sigma}^{3/2}}\left(  1+\frac{\log n}{p}\right)  ^{3/2}\frac{1}{\sqrt{n}}\right.  \\
&  \left.  +\frac{\tau^{4}\sqrt{p}}{\lambda_{\Sigma}^{2}n}\left(  1+\frac{\log n}{p}\right)  ^{1/2}+\frac{\tau^{2}\sqrt{p}}{\lambda_{\Sigma}n}\left(1+\frac{\log n}{p}\right)  ^{1/2}+\frac{\tau^{3}\sqrt{p}}{\lambda_{\Sigma}^{3/2}n}\left(  1+\frac{\log n}{p}\right)  \right)  \\
&  +C_{1}\left(  \frac{\tau^{4}(\log n)^{3/2}}{\lambda_{\Sigma}^{2}\sqrt{n}}+\frac{\tau^{2}(\log n)^{3/2}}{\lambda_{\Sigma}\sqrt{n}}+\frac{\tau^{3}}{\lambda_{\Sigma}^{3/2}\sqrt{n}}\left(  1+\frac{\log n}{p}\right)^{1/2}(\log n+\log p)\sqrt{\log n}\right)  ,
\end{align*}
where $C$ is a universal constant, $C_{1}$ only depends on $C_{X}$ and we have absorbed $\log p/p$ into the constant term $1$ due to $\log p/p\leq1$.

Now we combine the above bound with Lemma \ref{berry_esseen_delta_method} in this paper. Since $X$ is sub-Gaussian, the moment conditions in Lemma \ref{berry_esseen_delta_method} hold. Moreover, since $\Sigma$ is positive definite and $||\nabla g(\mu)||>0$, $\sigma^{2}=\nabla g(\mu)^{\top}\Sigma\nabla g(\mu)>0$ is also satisfied. Therefore, by Lemma \ref{berry_esseen_delta_method} and triangular inequality, we obtain the desired bound with probability at least $1-6/n$
\begin{align*}
&  \sup_{x\in\mathbb{R}}\left\vert P^{\ast}(\sqrt{n}(g(\bar{X}_{n}^{\ast})-g(\bar{X}_{n}))\leq x)-\Phi\left(  \frac{x}{\sigma}\right)  \right\vert \\
&  \leq C\left(  \frac{m_{31}}{\sqrt{n}\sigma^{3}}+\frac{C_{H_{g}}m_{31}^{1/3}tr(\Sigma)}{\sqrt{n}\sigma^{2}}+\frac{C_{H_{g}}m_{32}^{2/3}}{n^{5/6}\sigma}+\frac{C_{H_{g}}m_{31}^{1/3}m_{32}^{2/3}}{n\sigma^{2}}\right.\\
&  +\frac{C_{H_{g}}\tau^{2}}{C_{\nabla g}\sqrt{\lambda_{\Sigma}}}\left(1+\frac{\log n}{p}\right)  \sqrt{\frac{p}{n}}+\frac{||E[(X-\mu)^{3}]||}{\lambda_{\Sigma}^{3/2}}\frac{1}{\sqrt{n}}+\frac{\tau^{3}}{\lambda_{\Sigma}^{3/2}}\left(  1+\frac{\log n}{p}\right)  ^{3/2}\frac{1}{\sqrt{n}}\\
&  \left.  +\frac{\tau^{4}\sqrt{p}}{\lambda_{\Sigma}^{2}n}\left(  1+\frac{\log n}{p}\right)  ^{1/2}+\frac{\tau^{2}\sqrt{p}}{\lambda_{\Sigma}n}\left(1+\frac{\log n}{p}\right)  ^{1/2}+\frac{\tau^{3}\sqrt{p}}{\lambda_{\Sigma}^{3/2}n}\left(  1+\frac{\log n}{p}\right)  \right) \\
&  +C_{1}\left(  \frac{\tau^{4}(\log n)^{3/2}}{\lambda_{\Sigma}^{2}\sqrt{n}}+\frac{\tau^{2}(\log n)^{3/2}}{\lambda_{\Sigma}\sqrt{n}}+\frac{\tau^{3}}{\lambda_{\Sigma}^{3/2}\sqrt{n}}\left(  1+\frac{\log n}{p}\right)^{1/2}(\log n+\log p)\sqrt{\log n}\right)  ,
\end{align*}
where $C$ is a universal constant and $C_{1}$ only depends on $C_{X}$.
\end{proof}

\begin{proof}[Proof of Theorem \ref{CB_coverage_nonlinear}:]
Plugging the bounds in Lemma \ref{berry_esseen_delta_method} and Lemma \ref{highdim_CLT_nonlinear} into Theorem \ref{general_CB_coverage}, we obtain the desired finite sample bound for the cheap bootstrap coverage accuracy. Besides, the error $\mathcal{E}_{1}$ can be absorbed in $\mathcal{E}_{2}$.
\end{proof}

We make a remark about the proof of Theorem \ref{CB_coverage_nonlinear}. 
In \citet{zhilova2020nonclassical}, it appears that a generalized version of the Hanson-Wright inequality (equation (B.32)) that allows dependent components is derived as a middle step of the proof of Theorem B.1. If the classical Hanson-Wright inequality (Theorem 1.1 in \citet{rudelson2013hanson}) is used in that proof instead, then our Theorem \ref{CB_coverage_nonlinear} would be changed accordingly to have an additional assumption that $X$ has independent components but no longer needs to be continuously distributed with a bounded density. In this case, $C_{1}$ can be taken as a universal constant.

\begin{proof}[Proof of Corollary \ref{concise_CB_coverage_nonlinear}:]
It suffices to show that if $\tau=O(1)$ and $||\nabla g(\mu)||^{2}=O(p)$, then $tr(\Sigma)=O(p)$, $m_{31}=O(p^{3/2})$ and $m_{32}=O(p^{3/2})$. In fact, if these orders hold, with other orders assumed in Corollary \ref{concise_CB_coverage_nonlinear}, we can easily get the desired order after absorbing the small order terms into large order terms.

Now let us prove $tr(\Sigma)=O(p)$, $m_{31}=O(p^{3/2})$ and $m_{32}=O(p^{3/2})$ provided $\tau=O(1)$ and $||\nabla g(\mu)||^{2}=O(p)$. Recall that $X$ is assumed to be sub-Gaussian, i.e.,
\begin{equation}
E[\exp(a^{\top}(X-\mu))]\leq\exp(||a||^{2}\tau^{2}/2),\forall a\in\mathbb{R}^{p}.\label{sub-Gaussianity}%
\end{equation}
for some $\tau^{2}>0$. Therefore, $X_{i}-\mu_{i},i=1,\ldots,p$ are sub-Gaussian random variables with sub-Gaussian norm $\tau$ up to a universal constant (see \citet{vershynin2018high} Section 2.5). For simplicity, we write $a\lesssim b$ if $a\leq Cb$ for a universal constant $C>0$. By Proposition 2.5.2 (ii) in \citet{vershynin2018high}, $E[|X_{i}-\mu_{i}|^{2}]=\Sigma_{ii}\lesssim\tau^{2}$ and $E[|X_{i}-\mu_{i}|^{4}]\lesssim\tau^{4}$. Therefore, we can see $tr(\Sigma)\lesssim\tau^{2}p=O(p)$. By H\"{o}lder's inequality,%
\begin{align*}
m_{32}  & =E[||X-\mu||^{3}]\leq E[||X-\mu||^{4}]^{3/4}=\left(  \sum_{i,j=1}^{p}E\left[  (X_{i}-\mu_{i})^{2}(X_{j}-\mu_{j})^{2}\right]  \right)^{3/4}\\
& \leq\left(  \sum_{i,j=1}^{p}\sqrt{E\left[  (X_{i}-\mu_{i})^{4}]E[(X_{j}-\mu_{j})^{4}\right]  }\right)  ^{3/4}\lesssim\left(  \sum_{i,j=1}^{p}\tau^{4}\right)  ^{3/4}=\tau^{3}p^{3/2}=O(p^{3/2}).
\end{align*}
Moreover, (\ref{sub-Gaussianity}) also implies that $\nabla g(\mu)^{\top}(X-\mu)$ is sub-Gaussian with sub-Gaussian norm $||\nabla g(\mu)||\tau$ up to a universal constant. By Proposition 2.5.2 (ii) in \citet{vershynin2018high}, we have $m_{31}=E[|\nabla g(\mu)^{\top}(X-\mu)|^{3}]\lesssim||\nabla g(\mu)||^{3}\tau^{3}=O(p^{3/2})$. This concludes our proof.
\end{proof}

\begin{proof}[Proof of Theorem \ref{CB_coverage_subexp}:]
We will apply Theorem \ref{general_CB_coverage} to prove this theorem. The finite-sample bound (\ref{general_Berry_Esseen}) can be obtained by the Berry-Esseen theorem:
\begin{equation}
\sup_{x\in\mathbb{R}}\left\vert P(\sqrt{n}(g_{1}^{\top}\bar{X}_{n}-g_{1}^{\top}\mu)\leq x)-\Phi\left(  \frac{x}{\sigma}\right)  \right\vert \leq\frac{CE[|g_{1}^{\top}(X-\mu)|^{3}]}{\sigma^{3}\sqrt{n}}.\label{1D_Berry_Esseen}%
\end{equation}
Next, we need to find a bound for%
\[
\sup_{x\in\mathbb{R}}\left\vert P^{\ast}(\sqrt{n}(g_{1}^{\top}\bar{X}_{n}^{\ast}-g_{1}^{\top}\bar{X}_{n})\leq x)-\Phi\left(  \frac{x}{\sigma}\right)  \right\vert .
\]
We consider the i.i.d. centered random variables $g_{1}^{\top}(X_{i}-\mu),i=1,\ldots,n$ which have non-degenerate variance $\sigma^{2}=g_{1}^{\top}\Sigma g_{1}>0$. We apply Theorem 2.5 in \citet{lopes2022central} to $g_{1}^{\top}(X_{i}-\mu)$'s by choosing $Y=N(0,g_{1}^{\top}\Sigma g_{1})$ such that $\varrho=1$, $\Delta=0$, $\omega_{1}=||g_{1}^{\top}(X-\mu)/\sigma||_{\psi_{1}}=||g_{1}^{\top}(X-\mu)||_{\psi_{1}}/\sigma$ and obtain that with probability at least $1-C/n$,%
\[
\sup_{x\in\mathbb{R}}|P^{\ast}(\sqrt{n}(g_{1}^{\top}\bar{X}_{n}^{\ast}-g_{1}^{\top}\bar{X}_{n})\leq x)-P(\sqrt{n}g_{1}^{\top}(\bar{X}_{n}-\mu)\leq x)|\leq\frac{C||g_{1}^{\top}(X-\mu)||_{\psi_{1}}^{4}\log^{11}(n)}{\sigma^{4}\sqrt{n}},
\]
where $C>0$ is a universal constant. By the triangle inequality and Berry-Esseen bound (\ref{1D_Berry_Esseen}), the following holds with probability at least $1-C/n$%
\begin{equation}
\sup_{x\in\mathbb{R}}\left\vert P^{\ast}(\sqrt{n}(g_{1}^{\top}\bar{X}_{n}^{\ast}-g_{1}^{\top}\bar{X}_{n})\leq x)-\Phi\left(  \frac{x}{\sigma}\right)  \right\vert \leq\frac{CE[|g_{1}^{\top}(X-\mu)|^{3}]}{\sigma^{3}\sqrt{n}}+\frac{C||g_{1}^{\top}(X-\mu)||_{\psi_{1}}^{4}\log^{11}(n)}{\sigma^{4}\sqrt{n}}, \label{bootstrap_CLT_subexp}%
\end{equation}
where $C>0$ is a universal constant.

By Theorem \ref{general_CB_coverage} and the bounds (\ref{1D_Berry_Esseen}) and (\ref{bootstrap_CLT_subexp}), we finally get%
\begin{align*}
& \left\vert P(g(\mu)\in[g(\bar{X}_{n})-t_{B,1-\alpha/2}S_{n,B},g(\bar{X}_{n})+t_{B,1-\alpha/2}S_{n,B}])  -(1-\alpha)\right\vert\\
&  \leq \frac{C}{n}+BC\left(  \frac{E[|g_{1}^{\top}(X-\mu)|^{3}]}{\sigma^{3}\sqrt{n}}+\frac{||g_{1}^{\top}(X-\mu)||_{\psi_{1}}^{4}\log^{11}(n)}{\sigma^{4}\sqrt{n}}\right)  +\frac{CE[|g_{1}^{\top}(X-\mu)|^{3}]}{\sigma^{3}\sqrt{n}},
\end{align*}
where $C$ is a universal constant. Furthermore, the last term can be absorbed into the second term by using a larger constant $2C$, which completes our proof.
\end{proof}

\begin{proof}[Proof of Theorem \ref{CB_coverage_moments}:]
We use Theorem \ref{general_CB_coverage} to prove this theorem. As in the proof of Theorem \ref{CB_coverage_subexp}, (\ref{general_Berry_Esseen}) is given by the Berry-Esseen theorem in (\ref{1D_Berry_Esseen}) and we only need to bound
\[
\sup_{x\in\mathbb{R}}\left\vert P^{\ast}(\sqrt{n}(g_{1}^{\top}\bar{X}_{n}^{\ast}-g_{1}^{\top}\bar{X}_{n})\leq x)-\Phi\left(  \frac{x}{\sigma}\right)  \right\vert .
\]
In this regard, we will use the results in \citet{chernozhukov2020nearly}. Note that their results only apply for at least three-dimensional random vectors so we consider the following setting. Suppose we have $3n$ i.i.d. random variables $g_{1}^{\top}(X_{ij}-\mu),i=1,\ldots,n,j=1,2,3$ where each $X_{ij}$ has the same distribution as $X_{1}$. Then we can construct $n$ three-dimensional i.i.d. random vectors as $\tilde{X}_{i}:=(g_{1}^{\top}(X_{i1}-\mu),g_{1}^{\top}(X_{i2}-\mu),g_{1}^{\top}(X_{i3}-\mu))^{\top},i=1,\ldots,n$ whose components have common variance $\sigma^{2}=g_{1}^{\top}\Sigma g_{1}$. Then we can see that conditions (E.3) and (M) are satisfied for%
\[
B_{n}=\max\left\{  3E[|g_{1}^{\top}(X-\mu)/\sigma|^{q}]^{1/q},\sqrt{E[|g_{1}^{\top}(X-\mu)/\sigma|^{4}]}\right\}  .
\]
Therefore, by Corollary 3.2 in \citet{chernozhukov2020nearly} ($\sigma_{\ast,W}=1$ since $\tilde{X}_{i}$ has independent components), we have with probability at least $1-1/\sqrt{n}$ that%
\begin{align*}
&  \sup_{A\in\mathcal{R}}|P^{\ast}(\sqrt{n}(\bar{\tilde{X}}_{n}^{\ast}-\bar{\tilde{X}}_{n})\in A)-P(N(0,\sigma^{2}I_{3\times3})\in A)|\\
&  \leq C_{1}B_{n}\left(  \frac{\log3\log n\sqrt{\log(3\sqrt{n})}}{\sqrt{n}}+\frac{\log3\sqrt{\log(3n)}}{n^{1/2-3/(2q)}}\right) \\
&  \leq C_{1}\max\left\{  3E[|g_{1}^{\top}(X-\mu)/\sigma|^{q}]^{1/q},\sqrt{E[|g_{1}^{\top}(X-\mu)/\sigma|^{4}]}\right\}  \left(  \frac{(\log n)^{3/2}}{\sqrt{n}}+\frac{\sqrt{\log n}}{n^{1/2-3/(2q)}}\right)  ,
\end{align*}
where $C_{1}>0$ denotes a constant depending only on $q$ which are different for its two appearances and $\mathcal{R}$ contains all the hyperrectangles in $\mathbb{R}^{3}$. In particular, if we only focus on the first component of $\tilde{X}_{i}$, that is, we choose $A=(-\infty,x]\times\mathbb{R\times R}$, we have with probability at least $1-1/\sqrt{n}$ that
\begin{align}
&  \sup_{x\in\mathbb{R}}\left\vert P^{\ast}(\sqrt{n}(g_{1}^{\top}\bar{X}_{n}^{\ast}-g_{1}^{\top}\bar{X}_{n})\leq x)-\Phi\left(  \frac{x}{\sigma}\right)  \right\vert \nonumber\\
&  \leq C_{1}\max\left\{  3E[|g_{1}^{\top}(X-\mu)/\sigma|^{q}]^{1/q},\sqrt{E[|g_{1}^{\top}(X-\mu)/\sigma|^{4}]}\right\}  \left(  \frac{(\log n)^{3/2}}{\sqrt{n}}+\frac{\sqrt{\log n}}{n^{1/2-3/(2q)}}\right)  .\label{bootstrap_CLT_moments}%
\end{align}

By Theorem \ref{general_CB_coverage} and the bounds (\ref{1D_Berry_Esseen}) and (\ref{bootstrap_CLT_moments}), we then obtain
\begin{align*}
& \left\vert P(g(\mu)\in[g(\bar{X}_{n})-t_{B,1-\alpha/2}S_{n,B},g(\bar{X}_{n})+t_{B,1-\alpha/2}S_{n,B}])  -(1-\alpha)\right\vert\\
&  \leq \frac{2}{\sqrt{n}}+B  C_{1}\max\left\{  E[|g_{1}^{\top}(X-\mu)/\sigma|^{q}]^{1/q},\sqrt{E[|g_{1}^{\top}(X-\mu)/\sigma|^{4}]}\right\}  \\
&    \times\left(  \frac{(\log n)^{3/2}}{\sqrt{n}}+\frac{\sqrt{\log n}}{n^{1/2-3/(2q)}}\right)  +\frac{CE[|g_{1}^{\top}(X-\mu)|^{3}]}{\sigma^{3}\sqrt{n}},
\end{align*}
where $C$ is a universal constant and $C_{1}$ is a constant depending only on $q$. Finally notice that%
\[
\frac{(\log n)^{3/2}}{\sqrt{n}}=o\left(  \frac{\sqrt{\log n}}{n^{1/2-3/(2q)}}\right)  ,
\]
and $E[|g_{1}^{\top}(X-\mu)|^{3}]/\sigma^{3}\geq1$. We can absorb $(\log n)^{3/2}/\sqrt{n}$ into $\sqrt{\log n}/n^{1/2-3/(2q)}$ and absorb $2/\sqrt{n}$ into $CE[|g_{1}^{\top}(X-\mu)|^{3}]/\sigma^{3}\sqrt{n}$ (with larger constants $C_{1}$ and $C$), which leads to%
\begin{align*}
& \left\vert P(g(\mu)\in\lbrack g(\bar{X}_{n})-t_{B,1-\alpha/2}S_{n,B},g(\bar{X}_{n})+t_{B,1-\alpha/2}S_{n,B}])-(1-\alpha)\right\vert \\
& \leq BC_{1}\max\left\{  E[|g_{1}^{\top}(X-\mu)/\sigma|^{q}]^{1/q},\sqrt{E[|g_{1}^{\top}(X-\mu)/\sigma|^{4}]}\right\}  \frac{\sqrt{\log n}}{n^{1/2-3/(2q)}}+\frac{CE[|g_{1}^{\top}(X-\mu)|^{3}]}{\sigma^{3}\sqrt{n}}.
\end{align*}

\end{proof}

\begin{proof}[Proof of Theorem \ref{general_CB_coverage2}:]
In view of Theorem \ref{general_CB_coverage}, it suffices to show
\begin{align*}
&  |P(\psi\in\lbrack\hat{\psi}_{n}-t_{B,1-\alpha/2}S_{n,B},\hat{\psi}_{n}+t_{B,1-\alpha/2}S_{n,B}])-(1-\alpha)|\\
&  \leq2\mathcal{E}_{1}+2\mathcal{E}_{4}+\sqrt{\frac{2}{\pi}}|t_{B,1-\alpha/2}-z_{1-\alpha/2}|+\sqrt{\frac{2}{\pi}}\frac{\mathcal{E}_{3}}{\sigma}t_{B,1-\alpha/2}.
\end{align*}
Then taking the minimum of the two bounds, we can get the desired result. We write $\mathcal{A}$ as the event that%
\begin{align*}
&  \left\vert \sqrt{\frac{1}{B}\sum_{b=1}^{B}(\sqrt{n}(\hat{\psi}_{n}^{\ast b}-\hat{\psi}_{n}))^{2}}-\sigma\right\vert \leq\mathcal{E}_{3}\\
\Leftrightarrow &  ~\sigma-\mathcal{E}_{3}\leq\sqrt{\frac{1}{B}\sum_{b=1}^{B}(\sqrt{n}(\hat{\psi}_{n}^{\ast b}-\hat{\psi}_{n}))^{2}}\leq\sigma+\mathcal{E}_{3}.
\end{align*}
Then we have $P(\mathcal{A}^{c})\leq\mathcal{E}_{4}$. Note that the confidence interval can be written as%
\begin{align*}
& \psi\in\lbrack\hat{\psi}_{n}-t_{B,1-\alpha/2}S_{n,B},\hat{\psi}_{n}+t_{B,1-\alpha/2}S_{n,B}]\\
\Leftrightarrow & \left\{  \psi:\left\vert \frac{\sqrt{n}(\hat{\psi}_{n}-\psi)}{\sqrt{\frac{1}{B}\sum_{b=1}^{B}(\sqrt{n}(\hat{\psi}_{n}^{\ast b}-\hat{\psi}_{n}))^{2}}}\right\vert \leq t_{B,1-\alpha/2}\right\}
\end{align*}
Therefore, we have%
\begin{align*}
&  P\left(  \frac{|\sqrt{n}(\hat{\psi}_{n}-\psi)|}{\sigma-\mathcal{E}_{3}}\leq t_{B,1-\alpha/2};\mathcal{A}\right)  \\
&  \leq P(  \psi\in\lbrack\hat{\psi}_{n}-t_{B,1-\alpha/2}S_{n,B},\hat{\psi}_{n}+t_{B,1-\alpha/2}S_{n,B}];\mathcal{A})  \\
&  \leq P\left(  \frac{|\sqrt{n}(\hat{\psi}_{n}-\psi)|}{\sigma+\mathcal{E}_{3}}\leq t_{B,1-\alpha/2};\mathcal{A}\right)  .
\end{align*}
For the upper bound, we can further bound it as follows%
\begin{align*}
&  P\left(  \frac{|\sqrt{n}(\hat{\psi}_{n}-\psi)|}{\sigma+\mathcal{E}_{3}}\leq t_{B,1-\alpha/2};\mathcal{A}\right)  \\
&  \leq P\left(  \frac{|\sqrt{n}(\hat{\psi}_{n}-\psi)|}{\sigma+\mathcal{E}_{3}}\leq t_{B,1-\alpha/2}\right)  \\
&  =P(-(\sigma+\mathcal{E}_{3})t_{B,1-\alpha/2}\leq\sqrt{n}(\hat{\psi}_{n}-\psi)\leq(\sigma+\mathcal{E}_{3})t_{B,1-\alpha/2}).
\end{align*}
By means of the finite-sample accuracy in (\ref{appendix:general_Berry_Esseen}), we have%
\begin{align*}
&  P(-(\sigma+\mathcal{E}_{3})t_{B,1-\alpha/2}\leq\sqrt{n}(\hat{\psi}_{n}-\psi)\leq(\sigma+\mathcal{E}_{3})t_{B,1-\alpha/2})\\
&  \leq\Phi\left(  \frac{\sigma+\mathcal{E}_{3}}{\sigma}t_{B,1-\alpha/2}\right)  -\Phi\left(  -\frac{\sigma+\mathcal{E}_{3}}{\sigma}t_{B,1-\alpha/2}\right)  +2\mathcal{E}_{1}\\
&  \leq\Phi(z_{1-\alpha/2})-\Phi(-z_{1-\alpha/2})+2\mathcal{E}_{1}+\sqrt{\frac{2}{\pi}}\left\vert \frac{\sigma+\mathcal{E}_{3}}{\sigma}t_{B,1-\alpha/2}-z_{1-\alpha/2}\right\vert \\
&  \leq1-\alpha+2\mathcal{E}_{1}+\sqrt{\frac{2}{\pi}}|t_{B,1-\alpha/2}-z_{1-\alpha/2}|+\sqrt{\frac{2}{\pi}}\frac{\mathcal{E}_{3}}{\sigma}t_{B,1-\alpha/2},
\end{align*}
where $z_{1-\alpha/2}$ is the $(1-\alpha/2)$-th quantile of the standard normal and the second inequality is due to the $1/\sqrt{2\pi}$-Lipschitz property of $\Phi(\cdot)$. For the lower bound, by a similar argument we can obtain%
\begin{align*}
&  P\left(  \frac{|\sqrt{n}(\hat{\psi}_{n}-\psi)|}{\sigma-\mathcal{E}_{3}}\leq t_{B,1-\alpha/2};\mathcal{A}\right)  \\
&  \geq P\left(  \frac{|\sqrt{n}(\hat{\psi}_{n}-\psi)|}{\sigma-\mathcal{E}_{3}}\leq t_{B,1-\alpha/2}\right)  -P(\mathcal{A}^{c})\\
&  \geq1-\alpha-2\mathcal{E}_{1}-\sqrt{\frac{2}{\pi}}|t_{B,1-\alpha/2}-z_{1-\alpha/2}|-\sqrt{\frac{2}{\pi}}\frac{\mathcal{E}_{3}}{\sigma}t_{B,1-\alpha/2}-P(\mathcal{A}^{c})\\
&  \geq1-\alpha-2\mathcal{E}_{1}-\sqrt{\frac{2}{\pi}}|t_{B,1-\alpha/2}-z_{1-\alpha/2}|-\sqrt{\frac{2}{\pi}}\frac{\mathcal{E}_{3}}{\sigma}t_{B,1-\alpha/2}-\mathcal{E}_{4}.
\end{align*}
Thus, by combining the upper and lower bounds, we have the following bound for the coverage error when $\mathcal{A}$ happens%
\begin{align*}
&  |P(\psi\in\lbrack\hat{\psi}_{n}-t_{B,1-\alpha/2}S_{n,B},\hat{\psi}_{n}+t_{B,1-\alpha/2}S_{n,B}];\mathcal{A})-(1-\alpha)|\\
&  \leq2\mathcal{E}_{1}+\mathcal{E}_{4}+\sqrt{\frac{2}{\pi}}|t_{B,1-\alpha/2}-z_{1-\alpha/2}|+\sqrt{\frac{2}{\pi}}\frac{\mathcal{E}_{3}}{\sigma}t_{B,1-\alpha/2}.
\end{align*}
Finally, the overall coverage error can be bounded by%
\begin{align*}
&  |P(\psi\in\lbrack\hat{\psi}_{n}-t_{B,1-\alpha/2}S_{n,B},\hat{\psi}_{n}+t_{B,1-\alpha/2}S_{n,B}])-(1-\alpha)|\\
&  \leq|P(  \psi\in\lbrack\hat{\psi}_{n}-t_{B,1-\alpha/2}S_{n,B},\hat{\psi}_{n}+t_{B,1-\alpha/2}S_{n,B}];\mathcal{A})  -(1-\alpha)|+P(\mathcal{A}^{c})\\
&  \leq2\mathcal{E}_{1}+2\mathcal{E}_{4}+\sqrt{\frac{2}{\pi}}|t_{B,1-\alpha/2}-z_{1-\alpha/2}|+\sqrt{\frac{2}{\pi}}\frac{\mathcal{E}_{3}}{\sigma}t_{B,1-\alpha/2},
\end{align*}
which, combined with Theorem \ref{general_CB_coverage}, gives us the desired bound.
\end{proof}

\begin{proof}[Proof of Theorem \ref{other_coverage_nonlinear}:]
By means of Lemmas \ref{berry_esseen_delta_method} and \ref{highdim_CLT_nonlinear}, this directly follows from Theorem \ref{general_other_coverages}. Besides, the error $\mathcal{E}_{1}$ can be absorbed into $\mathcal{E}_{2}$.
\end{proof}

\begin{proof}[Proof of Theorem \ref{other_coverage_subexp}:]
Plugging the bounds (\ref{1D_Berry_Esseen}) and (\ref{bootstrap_CLT_subexp}) into Theorem \ref{general_other_coverages}, we get the desired result.
\end{proof}

\begin{proof}[Proof of Theorem \ref{other_coverage_moments}:]
Plugging the bounds (\ref{1D_Berry_Esseen}) and (\ref{bootstrap_CLT_moments}) into Theorem \ref{general_other_coverages}, we get the desired result.
\end{proof}

\end{document}